\newif\ifanonym

\documentclass[letterpaper,11pt]{article}
\usepackage[margin=1in]{geometry}
\usepackage[T1]{fontenc}

\usepackage{amsmath,amsthm,amssymb}
\allowdisplaybreaks
\usepackage{mathrsfs}
\usepackage{xspace}
\usepackage{afterpage}
\usepackage[dvipsnames]{xcolor}
\usepackage[framemethod=TikZ]{mdframed}
\usepackage{url}
\usepackage[backref=page]{hyperref}
\usepackage{appendix}
\usepackage{multirow}
\usepackage{enumitem}
\usepackage{algpseudocode}
\usepackage{booktabs}
\usepackage{thm-restate}
\usepackage{graphicx}
\usepackage{subcaption}
\usepackage{tabularx}
\usepackage{comment}
\usepackage{mathtools}
\usepackage[capitalize,noabbrev]{cleveref}

\hypersetup{colorlinks=true,allcolors=blue}

\theoremstyle{plain}
\newtheorem{theorem}{Theorem}[section]
\newtheorem{lemma}[theorem]{Lemma}
\newtheorem{claim}[theorem]{Claim}
\newtheorem{corollary}[theorem]{Corollary}
\newtheorem{assumption}[theorem]{Assumption}

\newtheorem{proposition}[theorem]{Proposition}

\theoremstyle{definition}
\newtheorem{definition}[theorem]{Definition}

\newtheorem*{invariant*}{Invariant}

\theoremstyle{remark}
\newtheorem{remark}[theorem]{Remark}

\newtheoremstyle{restate}{}{}{\itshape}{}{\bfseries}{~(Restated).}{.5em}{\thmnote{#3}}
\theoremstyle{restate}
\newtheorem*{restate}{}

\crefname{theorem}{Theorem}{Theorems}
\crefname{lemma}{Lemma}{Lemmas}
\crefname{claim}{Claim}{Claims}
\crefname{corollary}{Corollary}{Corollaries}
\crefname{assumption}{Assumption}{Assumptions}
\crefname{observation}{Observation}{Observations}
\crefname{question}{Question}{Questions}
\crefname{proposition}{Proposition}{Propositions}
\crefname{fact}{Fact}{Facts}

\crefname{definition}{Definition}{Definitions}
\crefname{invariant}{Invariant}{Invariants}

\crefname{remark}{Remark}{Remarks}

\crefname{equation}{Equation}{Equations}
\crefname{figure}{Figure}{Figures}

\newcommand{\R}{\mathbb{R}}

\DeclareMathOperator{\poly}{poly}

\DeclareMathOperator{\polylog}{polylog}

\DeclareMathOperator{\ddim}{ddim}

\newcommand{\diam}{\mathrm{diam}}

\newcommand{\eqdef}{\triangleq}
\renewcommand{\Comment}[1]{\textnormal{\color{OliveGreen} $\triangleright$ {#1}}}
\newcommand{\LHS}{\mathrm{LHS}}
\newcommand{\RHS}{\mathrm{RHS}}
\renewcommand{\mid}{\,\vert\,}
\newcommand{\bigmid}{\,\big\vert\,\big.}

\newcommand{\biggmid}{\,\bigg\vert\,\bigg.}

\newcommand{\eps}{\varepsilon}
\renewcommand{\epsilon}{\varepsilon}

\newcommand{\cost}{\mathrm{cost}_{z}}
\newcommand{\OPT}{\mathrm{OPT}_{z}}
\newcommand{\dist}{\mathrm{dist}}
\newcommand{\wmin}{\underline{w}}
\newcommand{\wmax}{\Bar{w}}
\newcommand{\dmin}{\underline{d}}
\newcommand{\dmax}{\Bar{d}}

\newcommand{\DSc}{{\sf c}}
\newcommand{\DSd}{{\sf d}}
\newcommand{\DScost}{{\sf cost}_{z}}
\newcommand{\DSloss}{{\sf loss}_{z}}
\newcommand{\DSvolume}{{\sf volume}}

\newcommand{\Initialize}{\textnormal{\texttt{initialize}}}
\newcommand{\InitializeSubClusterings}{\textnormal{\texttt{initialize-subclusterings}}}

\newcommand{\Insert}{\textnormal{\texttt{insert}}}
\newcommand{\InsertSubclustering}{\textnormal{\texttt{insert-subclustering}}}

\newcommand{\Delete}{\textnormal{\texttt{delete}}}
\newcommand{\DeleteSubclustering}{\textnormal{\texttt{delete-subclustering}}}

\newcommand{\SampleNoncenter}{\textnormal{\texttt{sample-noncenter}}}

\newcommand{\Cinitial}{C^{\sf init}}
\newcommand{\Cterminal}{C^{\sf term}}
\newcommand{\cinsert}{c^{\sf ins}}
\newcommand{\cdelete}{c^{\sf del}}

\newcommand{\Tinitialize}{T^{\sf init}}
\newcommand{\Tinsert}{T^{\sf ins}}
\newcommand{\Tdelete}{T^{\sf del}}

\newcommand{\LocalSearch}{\textnormal{\texttt{local-search}}}
\newcommand{\TestEffectiveness}{\textnormal{\texttt{test-effectiveness}}}

\newcommand{\cluster}{V}

\newcommand{\failure}{\textsf{failure}}

\newcommand{\CGnumber}{t}
\newcommand{\DSnumber}{s}

\newcommand{\TLSH}{T^{\sf LSH}}

\newcommand{\calB}{\mathcal{B}}

\newcommand{\calD}{\mathcal{D}}
\newcommand{\calE}{\mathcal{E}}

\newcommand{\calG}{\mathcal{G}}
\newcommand{\calH}{\mathcal{H}}

\newcommand{\calJ}{\mathcal{J}}

\newcommand{\calQ}{\mathcal{Q}}

\newcommand{\calT}{\mathcal{T}}

\newcommand{\calX}{\mathcal{X}}

\renewcommand{\Pr}{\operatorname{{\bf Pr}}}

\newcommand{\ProblemName}[1]{\textsc{#1}}

\newcommand{\kMedian}{\ProblemName{$k$-Median}\xspace}
\newcommand{\kMeans}{\ProblemName{$k$-Means}\xspace}

\newcommand{\kCenter}{\ProblemName{$k$-Center}\xspace}

\newcommand{\kzC}{\ProblemName{$(k,z)$-Clustering}\xspace}

\def\colorful{1}
\ifnum\colorful=1

\newcommand{\white}[1]{{{\color{white}#1}}}

\fi
\ifnum\colorful=0

\fi

\usepackage[obeyFinal,textsize=footnotesize]{todonotes}

\renewcommand{\tilde}[1]{\widetilde{#1}}
\renewcommand{\hat}[1]{\widehat{#1}}

\makeatletter
\def\term{\@ifnextchar[\term@optarg\term@noarg}\def\term@optarg[#1]#2{\textup{#1}\def\@currentlabel{#1}\def\cref@currentlabel{[][2147483647][]#1}\cref@label[term]{#2}}
\def\term@noarg#1{\refstepcounter{termcounter}\textup{\thetermcounter}\cref@label[term]{#1}}
\makeatother

\newcommand{\argmin}{\mathop{\mathrm{argmin}}}

\newcommand{\notni}{\not\owns}
\renewcommand{\emptyset}{\varnothing}

\title{Local Search for Clustering in Almost-linear Time}

\date{}

\ifanonym
\author{Anonymous Authors}
\else
\author{
    Shaofeng H.-C. Jiang\thanks{Peking University.
    Email: \texttt{shaofeng.jiang@pku.edu.cn}}
    \and
    Yaonan Jin\thanks{Huawei. Email: \texttt{jinyaonan@huawei.com}}
    \and
    Jianing Lou
    \thanks{Peking University. Email: \texttt{loujn@pku.edu.cn}}
    \and
    Pinyan Lu\thanks{Shanghai University of Finance and Economics, Laboratory of Interdisciplinary Research of Computation and Economics (SUFE), \& Huawei. Email: \texttt{lu.pinyan@mail.shufe.edu.cn}}
 }
\fi 

\begin{document}
\maketitle

\begin{abstract}
We propose the first {\em local search} algorithm for Euclidean clustering that attains an $O(1)$-approximation in almost-linear time.
Specifically, for Euclidean {\kMeans}, our algorithm achieves an $O(c)$-approximation in $\tilde{O}(n^{1 + 1 / c})$ time, for any constant $c \ge 1$, maintaining the same running time as the previous (non-local-search-based) approach [la~Tour and Saulpic, arXiv'2407.11217] while improving the approximation factor from $O(c^{6})$ to $O(c)$. 
The algorithm generalizes to any metric space with sparse spanners, delivering
efficient constant approximation in $\ell_p$ metrics, doubling metrics, Jaccard metrics, etc.

This generality derives from our main technical contribution: a local search algorithm on general graphs that obtains an $O(1)$-approximation in almost-linear time.
We establish this through a new $1$-swap local search framework featuring a novel swap selection rule.
At a high level, this rule ``scores'' every possible swap, based on both its modification to the clustering and its improvement to the clustering objective, and then selects those high-scoring swaps.
To implement this, we design a new data structure for maintaining approximate nearest neighbors with amortized guarantees tailored to our framework.
\end{abstract}

 \thispagestyle{empty}
\newpage

\tableofcontents
\thispagestyle{empty}
\newpage
\setcounter{page}{1}

\begingroup
\renewcommand{\cost}{\mathrm{cost}}

\section{Introduction}

Euclidean {\kMeans} clustering, a fundamental problem in combinatorial optimization, constitutes a central research direction in approximation algorithms.
Given an $n$-point dataset $X \subseteq \R^d$ in $d$-dimensional Euclidean space, the {\kMeans} objective seeks $k$ centers $C \subseteq \R^d$ minimizing:
\begin{equation*}
    \cost(X, C) ~\eqdef~ \sum_{x \in X} \min_{c\in C} \|x - c\|_2^2.
\end{equation*}

In this work, we study {\em efficient} approximation algorithms for Euclidean {\kMeans} through the lens of {\em fine-grained} complexity, aiming to achieve sub-quadratic running time in the general parameter regime, where both $k \in [n]$ and $d \ge 1$ are part of the input.
W.l.o.g., we may assume $d = O(\log n)$, as dimensionality reduction via the Johnson-Lindenstrauss Transform \cite{JL84, MMR19} incurs only a $(1 + \epsilon)$-factor approximation error.

For this general parameter regime $k \in [n]$ and $d \ge 1$, the tight tradeoff between the approximation ratio and running time is still open,
albeit some basic lower bounds are known:
On the one hand, when $d = \Omega(\log n)$, Euclidean {\kMeans} is APX-hard~\cite{AwasthiCKS15},
so a constant approximation is the best achievable.
On the other hand, despite a long line of research on constant-approximation algorithms in general metrics \cite{CharikarGTS99, CharikarG99, GuhaMMO00, JainV01, JainMS02, JainMMSV03,  MettuP03, AryaGKMMP04, MettuP04, GuptaT08, CharikarL12, LiS16, ByrkaPRST17, ANSW20, Cohen-AddadGHOS22, Cohen-Addad0LS23, GowdaPST23, CGLSO25}, which apply to our Euclidean setting,
the state-of-the-art running time remains $\tilde{O}(nk)$ \cite{MettuP04, DBLP:journals/siamcomp/Chen09}.
This bound is quadratic in the worst case, as $k$ can be linear in $n$.
Indeed, it is shown that any constant approximation for {\kMeans} in general metrics cannot run in sub-quadratic $o(n^2)$ time when $k = \Omega(n)$~\cite{MettuP04}.

A natural benchmark for clustering in high dimensions is the batch approximate nearest neighbor (ANN) problem -- the simpler task of assigning the given data points for {\em pre-specified} centers, regardless of optimization.
The state-of-the-art ANN tradeoff between accuracy and running time, which is attributed to the Locality Sensitive Hashing techniques~\cite{AndoniI06,AndoniR15} and thus called the ``LSH tradeoff'' hereafter, attains $O(c)$-approximate assignments in $n^{1 + 1 / c^2}$ time for any $c \ge 1$.\footnote{When $c = 1 + \eps$ is close enough to $1$, non-LSH techniques can improve the running time to $n^{2 - \Tilde{\Omega}(\eps^{1/3})}$ \cite{DBLP:conf/focs/AlmanCW16}.}
Note that this LSH tradeoff requires almost-linear time for a constant approximation.
Instead, whether this is attainable in near-linear time $n \polylog(n)$ is an open problem.

Regarding Euclidean \kMeans, whose objective involves squared distances,  
the LSH tradeoff translates to an $O(c)$-approximation in $n^{1 + 1 / c}$ running time.
To attain this benchmark, one approach is to incorporate a benchmark-matching construction of Euclidean spanners (e.g., \cite{HIS13}) into a near-linear time clustering algorithm on graphs (e.g., \cite{Thorup04}).
Yet, only one such graph algorithm \cite{Thorup04} has been known in the literature, and it was devised for different but related clustering problems, including \kMedian and \kCenter. Whether this algorithm can be extended to {\kMeans} remains unclear (albeit plausible).
Indeed, no explicit tradeoff for Euclidean {\kMeans} had been established until the recent work by la Tour and Saulpic \cite{lTS24}, namely an $O(c^6)$-approximation in $n^{1 + 1/c}$ time.
However, this result still fails to meet the LSH tradeoff.
In sum, a benchmark-matching tradeoff for Euclidean {\kMeans} is of fundamental interest but remains open.

\paragraph{Local search for clustering.} 
To achieve the said sub-quadratic/almost-linear fine-grained running time,
it requires us to revisit classic algorithmic techniques for clustering through the lens of efficiency.
Indeed, the study of clustering problems have inspired the development of various algorithmic techniques, 
including primal-dual~\cite{CharikarG99, JainV01}, LP rounding~\cite{CharikarGTS02}, 
and local search~\cite{AryaGKMMP04}. 
However, these techniques for clustering have been mostly studied in the general metric setting (which inherently requires $\Omega(n^2)$ time to achieve constant approximation as mentioned),
and the focus was more on the approximation ratio side instead of efficiency.
Whether or not these techniques can be adapted to sub-quadratic time is an interesting aspect that is not well understood.

Technically, this paper aims to further develop these techniques in terms of efficiency,
and we focus on the local search paradigm,
a fundamental algorithm design paradigm that has been widely used not only for clustering but also for other well-known problems such as
\ProblemName{Max-Cut} \cite{KT06},
\ProblemName{Steiner forest} \cite{0001G0MS0V18}, and
\ProblemName{Submodular Maximization} \cite{FilmusW14}.
For {\kMeans} problem specifically,
local search stands out as a versatile approach, applicable across a wide range of settings even beyond Euclidean spaces, including: 
(i) achieving a $(25+\epsilon)$-approximation in general metric spaces~\cite{GuptaT08}, with the ratio improved to $(9+\epsilon)$ for Euclidean metric spaces~\cite{KanungoMNPSW04}; (ii) offering PTAS in low-dimensional Euclidean spaces~\cite{Cohen-Addad18, Cohen-AddadKM19, FriggstadRS19}, doubling metric spaces~\cite{FriggstadRS19}, and shortest-path metrics on minor-free graphs~\cite{Cohen-AddadKM19}; and (iii) demonstrating its applicability in fully dynamic settings~\cite{BhattacharyaCGLP24, BhattacharyaCF24}.

However, despite this flourish study, no known local search algorithm for clustering achieves sub-quadratic running time, even in low-dimensional Euclidean spaces.
Hence, we aim to break this quadratic barrier for local search,
which not only helps to achieve the explicit LSH tradeoff for Euclidean {\kMeans},
but also may lead to new algorithmic insights that could benefit local search in general.

\subsection{Our results}

We propose a novel variant of $1$-swap local search that breaks the quadratic time barrier.
Notably, this local search leads to the first explicit result that realizes the LSH tradeoff for Euclidean {\kMeans}, stated in \Cref{thm:main-informal-Euclidean}.
More generally, this local search also yields new efficient clustering results beyond Euclidean spaces.

\begin{theorem}[Euclidean {\kMeans}; see \Cref{cor:lsh}]
    \label{thm:main-informal-Euclidean}
    For any constant $c \ge 1$, there is an algorithm that computes an $O(c)$-approximation for Euclidean {\kMeans} on a given $n$-point dataset with aspect ratio $\Delta > 0$,\footnote{The aspect ratio of a dataset is defined as the ratio between the maximum and minimum pairwise distances among the data points.}
    running in time $\tilde{O}(dn^{1 + 1 / c}\log(\Delta))$ and succeeding with high probability. 
\end{theorem}

We note that \Cref{thm:main-informal-Euclidean} (and all of our results presented here) work more generally to the {\kzC} problem (formally defined in \Cref{sec:prelim}), whose objective takes the $z$-th power sum of distances and
encompasses both {\kMedian} ($z = 1$) and {\kMeans} ($z = 2$).
For general {\kzC} in Euclidean space, the ratio-time tradeoff of \Cref{thm:main-informal-Euclidean} becomes $O(c^z)$ versus $\tilde{O}(dn^{1+1/c^2} \log \Delta)$.
Hence, our algorithm not only improves upon the previous $O(c^{6z})$-approximation in the same time regime~\cite{lTS24}, but also matches and generalizes the state-of-the-art for Euclidean {\kMedian} ($z = 1$), achieved by combining the graph {\kMedian} algorithm~\cite{Thorup04} with spanners~\cite{HIS13}.

Beyond the worst-case of $k = \Theta(n)$, our algorithm is also useful for moderately large values of $k = n^{1 - \epsilon}$ for any $0 <\epsilon < 1$, by running on top of coresets~\cite{Har-PeledM04, Har-PeledK07, DBLP:journals/siamcomp/Chen09, FeldmanL11, DBLP:conf/focs/SohlerW18, FeldmanSS20, HuangV20, Cohen-AddadSS21, Cohen-AddadLSS22, Cohen-AddadLSSS22, DraganovSS24, Huang0024, 0001CPSS24, Cohen-AddadD0SS25}.
Specifically, we achieve constant $O(\epsilon^{-1})$-approximation in {\em near-linear} time $\tilde{O}(nd + k^{1+\epsilon}) = \tilde{O}(nd)$,
by first constructing a coreset with size $\tilde{O}(k)$\footnote{
    Such coresets can be constructed in time $\tilde{O}(nd)$ for constant approximation, see e.g.,~\cite{DraganovSS24}.
} and then
applying our algorithm with $c = \epsilon^{-1}$.
We emphasize that previous $O(nk)$-time algorithms, even combined with coresets,
achieve near-linear time only when $k \le \tilde{O}(\sqrt{n})$,
whereas ours yields constant approximation in near-linear time for $k = n ^{1 - \epsilon}$ for {\em full range} of $\epsilon \in (0, 1)$.\footnote{While the recent almost-linear constant approximation of~\cite{lTS24} may yield similar results,
it however offers a worse ratio, i.e., $O(\epsilon^{-6})$-approximation in near-linear time.}

\paragraph{Beyond Euclidean spaces.}

\begin{table}[t]
    \centering
    \begin{tabular}{llll}
        \toprule 
        metric space & ratio & running time & reference \\
        \midrule
        Euclidean $\R^{d}$          & $O(c)$        & $dn^{1 + 1/c}$ & \Cref{cor:lsh} \\
                                    & $O(c^{6})$    & $dn^{1 + 1 / c}$ & \cite{lTS24} \\
graphs with $m$ edges & $O(1)$ & $m^{1 + o(1)}$ & \Cref{cor:general-graph}\\
        $\ell_{p}$ in $\R^{d}$, $\forall 1 \le p < 2$  & $O(c^{2})$ & $dn^{1 + 1/c}$ & \Cref{cor:lsh} \\
        doubling dimension $\ddim$  & $O(1)$        & $\Tilde{O}(2^{O(\ddim)} n)$ & \Cref{cor:doubling}  \\
                                    & $1 + \eps$ & $\Tilde{O}(2^{(1 / \eps)^{O(\ddim^{2})}}n)$ & \cite{Cohen-AddadFS21} \\
        Jaccard ($V = 2^U$)         & $O(c^{2})$    & $n^{1+ 1/c}\poly(|U|)$ & \Cref{cor:lsh}\\
\bottomrule
    \end{tabular}
    \caption{\label{tab:result}A summary of efficient algorithms for {\kMeans} in various families of metric spaces. 
    For ease of presentation, all bounds shown in this table omit $\log(\Delta)$ factors. This dependence on the aspect ratio $\Delta$ can be reduced/removed in certain metric spaces, such as doubling spaces; see \Cref{subsec:metric-clustering} for more details.}
\end{table}

In fact, our \Cref{thm:main-informal-Euclidean} for Euclidean {\kMeans} is a corollary
of a more general technical result: an almost-linear time constant approximation for {\kMeans} on a shortest-path metric of a weighted graph (see \Cref{thm:main-informal} which we discuss in more detail later).
The Euclidean result is obtained by running this graph clustering algorithm on a sparse metric spanner~\cite{HIS13}.
Moreover, this approach generalizes to other metric spaces as long as sparse spanners exist.
We give a summary of the various results we have in \Cref{tab:result},
and we highlight the following notable ones.
\begin{itemize} 
    \item (\Cref{cor:lsh}) Metric spaces that admit LSH. This includes $\ell_{p}$ metrics, $\forall 1 \le p < 2$, and Jaccard metrics. For these spaces, the construction from \cite{HIS13} for Euclidean spaces can be generalized and yields competitive parameters; see also \Cref{cor:metric-lsh} for a general LSH-parameterized statement.
    \item (\Cref{cor:doubling}) Metric spaces with $n$ points and doubling dimension $\ddim\ge 1$. We can plug in the spanner construction from, e.g.,~\cite{Har-PeledM06, Solomon14}, to obtain a constant-factor approximation for {\kMeans} in time $\Tilde{O}(2^{O(\ddim)} n)$. It is noteworthy that this running time only has a {\em singly exponential} dependence on $\ddim$, so our algorithm complements the $\mathrm{EPTAS}$ by \cite{Cohen-AddadFS21}, i.e., a better $(1+\epsilon)$-approximation but in worse time $\Tilde{O}(2^{(1 / \eps)^{O(\ddim^{2})}} n)$.
\end{itemize}

The following \Cref{thm:main-informal} presents our (technical) result for {\kMeans} on general weighted graphs $G = (V, E, w)$.
In this statement, the factor $m^{o(1)}$ in the running time is given by $2^{O(\sqrt{\log m \log\log m})} = 2^{O(\sqrt{\log n \log\log n})}$, which is smaller than any polynomial $\poly(n)$ but larger than polylogarithms $\polylog(n)$.

\begin{theorem}[\kMeans on graphs; see \Cref{cor:general-graph}]
\label{thm:main-informal}
There is an $O(1)$-approximation for {\kMeans} on (the shortest-path metric of) $m$-edge weighted undirected graphs 
that runs in time $m^{1 + o(1)}$ and succeeds with high probability.
\end{theorem}

As a standalone result, \Cref{thm:main-informal} establishes the first explicit algorithm (thus the first local search algorithm) for {\kMeans} on shortest-path metrics with an $O(1)$-approximation in almost-linear time, and the result also extends to general {\kzC}.
We also provide concrete approximation ratios for {\kMeans} and {\kMedian} (instead of a generic $O(1)$ bound) in \Cref{remark:general-graph}, which are $\approx 44 + 16\sqrt{7} \approx 86.33$ and $\approx 6$, respectively.
Compared with best known approximation ratios in general metrics through $1$-swap local search algorithms, our new algorithm only incurs moderate degeneration, such as $86.33$ versus $81$ for {\kMeans}  \cite{GuptaT08}\footnote{For {\kMeans} in Euclidean spaces, \cite{KanungoMNPSW04} provides an improved ratio of $25$ for $1$-swap local search, which relies on properties specific to Euclidean space and is thus incomparable to our graph result.}
and $6$ versus $5$ for {\kMedian} \cite{AryaGKMMP04}.
In contrast, the speedup in running time -- from high-degree polynomial to almost-linear -- may be more significant.
Compared with the previous graph {\kMedian} algorithm \cite{Thorup04} (which does not use local search), our ratio of $6$ outperforms their ratio of $9$, although their running time $\tilde{O}(m)$ is slightly better.

\subsection{Technical contributions}
\label{sec:tech_overview}

Our main technical contribution is a new $1$-swap local search framework,
featuring a novel swap selection rule that explicitly relates the computational overhead of a swap to the improvement of the cost function.
We first present the new local search framework in \Cref{sec:techoverview:super-effective-swap},
where we focus on an intermediate complexity measure called clustering recourse,
which is a benchmark for the swap selection rule.
Our local search is the first to achieve a near-linear recourse bound, and this justifies the novelty and effectiveness of our selection rule.
Then to prove \Cref{thm:main-informal},
we discuss in \Cref{sec:techoverview:implementation-graph}
a dynamic approximate near neighbor data structure in shortest-path metrics of graphs
with an amortized complexity guarantee that is tailored to our local search framework.

\subsubsection{Local search with super-effective swaps}
\label{sec:techoverview:super-effective-swap}

Our local search framework works for any metric space $(V, \dist)$ (not necessarily Euclidean),
and we specifically consider the simple yet fundamental $1$-swap local search,
which has also been studied in e.g. ~\cite{AryaGKMMP04, KanungoMNPSW04, GuptaT08}.
Given an input dataset $X \subseteq V$, 
$1$-swap local search starts with an initial, say, $\poly(n)$-approximate solution $C \in V^k$, 
and iteratively refines it by selecting a center swap $(\cinsert, \cdelete) \in (V \setminus C) \times C$ (according to some rule) 
and updating $C \gets C \cup \{\cinsert\} \setminus \{\cdelete\}$, 
until reaching a (near) local optimum that guarantees an $O(1)$-approximation.  

To implement a $1$-swap local search, one often needs an auxiliary data structure to maintain the (approximate) clustering $\calX$ upon the change of center set in each iteration.
This can be a nontrivial task and even dominates the running time of the local search.
At a high level, our local search framework takes this running time into account, and we devise swap selection rules such that the overhead from maintaining the clustering is optimized.

\paragraph{Clustering oracle.}
Specifically, our local search is designed with respect to {\em any} given data structure that approximately maintains the clustering,
and in this discussion, we assume it is given as an oracle.
This oracle takes as input the swap sequence generated by local search, sequentially handles each center insertion/deletion, and 
maintains an assignment such that each data point is assigned to an {\em (approximate) nearest neighbor} in $C$.
For technical reason,
we further assume the oracle works in a natural well-behaved manner as follows.
\begin{enumerate}[label=(\alph*)]
    \item\label{behavior:insert} Upon inserting a non-center $\cinsert \in V \setminus C$ into $C$,
    the oracle ``lazily'' updates the clustering; that is, it reassigns a point $x \in X$ to $\cinsert$ (and to no other center) if its distance to its current center is larger than its distance to $\cinsert$ by a factor $ > 1+\epsilon$. Here $\epsilon > 0$ is a parameter.
    \item\label{behavior:deletion} Upon deleting a center $\cdelete \in C$, the oracle only reassigns points currently assigned to $\cdelete$.
\end{enumerate}
These two behaviors are natural to consider, as they minimize unnecessary reassignments while ensuring that the clustering remains $(1+\epsilon)$-accurate. 
Moreover, these natural conditions/properties turn out to be very useful in designing efficient clustering oracles,
as we will discuss in \Cref{sec:techoverview:implementation-graph}.

\paragraph{Clustering recourse.}
For the sake of presentation,
we focus the discussion on an intermediate performance measure called {\em clustering recourse}.
The clustering recourse of a local search (with respect to a clustering oracle) is defined as the total number of data point reassignments performed by the oracle as a result of the swaps selected by the local search algorithm,
ignoring the additional overhead in running time for maintaining this assignment.
A bounded clustering recourse is a necessary condition for bounded running time.
Moreover,
this is useful
for specifically benchmarking the swap selection rule,
as the reassignment is a direct reflection of the work introduced by the swap,
and the recourse isolates this complexity from the running time.
To the best of our knowledge, 
the swap selection rules employed by previous local search algorithms~\cite{AryaGKMMP04,KanungoMNPSW04, GuptaT08,Cohen-Addad18, Cohen-AddadKM19, FriggstadRS19, LattanziS19, ChooGPR20} are only known to achieve $O(n^2)$ recourse,
whereas ours is the first to achieve near-linear recourse.

\paragraph{Our swap selection rule.} 
In this setting, our new selection rule guides $1$-swap local search to choose a center swap based on the current state of the clustering oracle,
and the rule allows to choose $(\cinsert, \cdelete)$ only if
it satisfies the following condition which we call 
{\em super-effective} rule.
\begin{align}
    \label{eq:techoverview:super-effective}
    \frac{\cost(X, C \cup \{\cinsert\} \setminus \{\cdelete\})}{\cost(X, C)}
    ~\le~ \exp\big(-|\calX(\cdelete)| / n\big).
\end{align}
Here, $\calX(\cdelete)$ denotes the set of points assigned to the center $\cdelete$ in the clustering $\calX$ currently maintained by the oracle.
Importantly, due to Property~\ref{behavior:deletion}, $|\calX(\cdelete)|$ is exactly the  clustering recourse incurred by the deletion of $\cdelete$.

This super-effective Rule~\eqref{eq:techoverview:super-effective} intuitively relates the objective improvement to the (one-step) clustering recourse. Notably, the rule only accounts for the recourse caused by the deletion of $\cdelete$ (which we refer to as deletion-recourse), while ignoring the recourse caused by the insertion of $\cinsert$ (insertion-recourse); this is fine since, as shown soon after, the latter can be ``charged'' for the former.
There is also a technical reason for excluding the insertion-recourse in Rule~\eqref{eq:techoverview:super-effective}: the insertion-recourse is inherently {\em unpredictable}, as its value (the size of the newly formed cluster) can only be determined after the swap is performed, making it impossible to anticipate when selecting a swap. In contrast, the deletion-recourse is simply the size of the cluster being deleted, which is always known in advance.

We next explain how Rule~\eqref{eq:techoverview:super-effective} ensures near-linear clustering recourse, by separately bounding the total insertion-recourse and total deletion-recourse.
Suppose the $1$-swap local search performs a series $(\cinsert_1, \cdelete_1)$, 
$(\cinsert_2, \cdelete_2)$, $\ldots$ 
of super-effective swaps 
(such that each
$(\cinsert_i, \cdelete_i)$ satisfies Rule~\eqref{eq:techoverview:super-effective}).
Let $T_i$ denote the deletion-recourse w.r.t. $\cdelete_i$; by Property~\ref{behavior:deletion}, this equals the current cluster size $|\calX(\cdelete_i)|$.
The super-effectiven Rule~\eqref{eq:techoverview:super-effective} ensures that
\begin{equation*}
    \frac{1}{\poly(n)} ~\le~\frac{\cost(X, \Cterminal)}{\cost(X, \Cinitial)}
    ~\le~
    \exp\Big(-\sum_{i} T_{i} / n\Big).
\end{equation*}
Therefore, the total deletion-recourse is bounded by $\sum_{i} T_{i} = \Tilde{O}(n)$.
As for the total insertion-recourse, let us consider a single point $x\in X$. The ``lazy'' update strategy (Property~\ref{behavior:insert}) ensures that the point $x$ only experiences $O(\log_{1+\epsilon}\Delta) = O(\epsilon^{-1}\log\Delta)$ insertion-reassignments between every two deletion-reassignments, where $\Delta$ is the aspect ratio. As a result, up to an overhead factor of only $O(\epsilon^{-1} \log \Delta)$, the total insertion-recourse can be ``charged'' for the total deletion-recourse.

\paragraph{Comparison with other selection rules.}
Many previous local search algorithms \cite{AryaGKMMP04,KanungoMNPSW04,GuptaT08,Cohen-Addad18, Cohen-AddadKM19, FriggstadRS19,LattanziS19, ChooGPR20} adopt a simple selection rule that considers a swap $(\cinsert,\cdelete)$ selectable if it improves the objective by a {\em fixed} factor, specifically satisfying:
\begin{equation}
    \label{eq:techoverview:simple-rule}
    \frac{\cost(X, C \cup \{\cinsert\} \setminus \{\cdelete\})}{\cost(X, C)}
    ~\le~ 1 - \frac{1}{\poly(n)}.
\end{equation}
In contrast, our Rule~\eqref{eq:techoverview:super-effective} is more {\em adaptive}, as it incorporates the deletion-recourse into the selection criterion.
Under Rule~\eqref{eq:techoverview:simple-rule}, the clustering recourse of local search can be bounded only by the product of a) the total number of swaps and b) the worst-case number of reassignments per swap, leading to a quadratic bound.
Another line of work \cite{Alimonti96, KhannaMSV98, FilmusW12, FilmusW14, 0001G0MS0V18, Cohen-AddadGHOS22} investigates {\em non-oblivious} local search algorithms, which employ selection rules that are based on well-crafted variants of the clustering objective, rather than the naive $\cost(X, C)$, and can lead to better approximation ratios for {\kMedian} \cite{Cohen-AddadGHOS22}.
Still, these approaches provide no guarantee on clustering recourse.
In sum, our super-effectiveness Rule~\eqref{eq:techoverview:super-effective} is the first in the local search literature that achieves a near-linear clustering recourse.
We hope that the insight on our selection rule can benefit future local search algorithms in general.

\paragraph{Existence of super-effective swaps.} 
The final component of our framework is to ensure the existence of super-effective swaps, and we establish the following claim:

\begin{claim}[Informal; see \Cref{lem:sample}]
    \label{claim:techoverview:sample}
    If the current solution $C$ does not achieve some target constant approximation ratio, then for a random non-center $\cinsert \in X$ sampled via $D^2$-sampling \cite{ArthurV07}, there exists, with constant probability, a center $\cdelete \in C$ such that $(\cinsert, \cdelete)$ satisfies the super-effective Rule~\eqref{eq:techoverview:super-effective}.
\end{claim}

We actually proved a ``robust'' version of \Cref{claim:techoverview:sample}, such that it holds even when the {\em exact} clustering objective in Rule~\eqref{eq:techoverview:super-effective} is replaced by an approximation over $\calX$, provided that $\calX$ is, say, $1.01$-accurate.
This robustness is important for {\em implementing} our framework because it allows local search to use approximate objective values provided by the clustering oracle and thereby avoid the need for $\Omega(n)$-time exact computations. On the other hand, stringent accuracy is necessary, since large errors may mislead the selection of swaps and cause the search process to behave arbitrarily. This necessity also highlights a technical difficulty, since even a $(1 + 1 / n)$-error in the objective appears large enough to mislead the swap selection.

The proof for \Cref{claim:techoverview:sample} borrows ideas from \cite{LattanziS19, ChooGPR20, 0001CLP23}, which also analyze the $D^2$-sampling scheme of a swap-in non-center $\cinsert \notin C$, but under the simple ``fixed-progress'' Rule~\eqref{eq:techoverview:simple-rule} (rather than our ``adaptive-progress'' Rule~\eqref{eq:techoverview:super-effective}) with an {\em exact} clustering objective on the LHS. 
Our improved probabilistic analysis addresses both limitations.
In this way, we demonstrate that the $D^{2}$-sampling scheme is in fact more ``adaptable'' and ``robust'', and thus also works for our settings.
We have also carefully optimized all parameters involved in the our proof, and derived a concrete ratio (instead of the generic $O(1)$) for \Cref{claim:techoverview:sample} (see \Cref{remark:general-graph}), which is precisely the approximation ratio for our algorithms in \Cref{thm:main-informal} for graphs.
Beyond showing the existence, this claim also suggests a concrete way to find a feasible non-center $\cinsert$ via $D^2$-sampling. 
However, this $D^2$-sampling still requires an efficient implementation  which we will discuss in \Cref{sec:techoverview:implementation-graph}.

Finally, we emphasize that the above discussion concerns only clustering recourse. To extend our framework to running time, one key step is to account for the update time of a concrete clustering oracle. 
This involves replacing the cluster size in the RHS of Rule~\eqref{eq:techoverview:super-effective} with a time-related quantity. With this adjustment, a more complex but similar argument would also imply a bound on the overall running time.

\subsubsection{Implementing our framework with an ANN data structure}
\label{sec:techoverview:implementation-graph}

As discussed in \Cref{sec:techoverview:super-effective-swap}, our framework assumes a clustering oracle, which we implement in this section via a concrete data structure.
Besides maintaining the clustering efficiently, this data structure should also facilitate local search by enabling fast swap selection. In particular, it should support $D^2$-sampling to identify a swap-in candidate $\cinsert$ and provide sufficient information to determine a swap-out candidate $\cdelete$.
A technical observation here is that all these necessary functionalities reduce to, or can be built upon, the fundamental task of maintaining {\em approximate nearest neighbor (ANN)} information; henceforth, we focus on this task.

Specifically, we need a  data structure that {\em explicitly} maintains an {\em accurate} $(1 + \epsilon)$-approximate nearest neighbor ($(1 + \eps)$-ANN) $\DSc[x] \in C$ in the maintained center set $C$ of every point $x \in X$, such that
\begin{align}
    \label{eq:techoverview:ANN}
    \dist(x, \DSc[x])
    ~\le~ (1 + \eps) \cdot \dist(x, C),
    \qquad
    \forall x \in X.
\end{align}
Here, $0 < \epsilon$ is chosen to be sufficiently small (e.g., $\epsilon < 0.01$) to meet the accuracy requirement for the existence of a super-effective swap (as discussed in \Cref{sec:techoverview:super-effective-swap}). 
Moreover, the data structure should additionally support center insertions and deletions in $C$ while maintaining Condition~\eqref{eq:techoverview:ANN}. 

\paragraph{Bypassing worst-case guarantees.} However, it is difficult to achieve a worst-case time guarantee for the mentioned task: inserting or deleting a center may trigger as many as $\Omega(n)$ reassignments in the clustering, i.e., changes to the ANN array $(\DSc[x])_{x \in X}$, even for a large constant approximation (let alone $(1 + \epsilon)$).
Luckily, our framework already provides a way to bypass the need for worst-case guarantees, as long as the data structure satisfies the well-behaved properties~\ref{behavior:insert} and \ref{behavior:deletion} and its update time scales with the clustering recourse. Specifically:
\begin{itemize}
    \item For a deletion, the update must run in time {\em almost-linear} in $|\calX(\cdelete)|$, where $\cdelete$ is the to-delete center (recall that $\calX(\cdelete)$ is the cluster centered at $\cdelete$). If so, the total running time for all deletions is bounded by the total clustering recourse $\sum_{i} T_{i}$ (see \Cref{sec:techoverview:super-effective-swap}), which is near-linear, i.e., $\sum_{i} T_{i} = \tilde{O}(n)$ under the super-effective Rule~\eqref{eq:techoverview:super-effective}.
    
    \item For an insertion, the data structure must implement the ``lazy update'' strategy described in Property~\ref{behavior:insert}, with a running time proportional to the number of reassignments.
    In that case, the total running time for all insertions will again match the total clustering recourse $\sum_{i} T_{i}$ (see also \Cref{sec:techoverview:super-effective-swap}).
\end{itemize}

\paragraph{From Euclidean (\Cref{thm:main-informal-Euclidean}) to graph settings (\Cref{thm:main-informal}).} 
Unfortunately, in Euclidean space, even this amortized update time analysis is still difficult to realize, due to the stringent requirement of a very small $\epsilon$ in Condition~\eqref{eq:techoverview:ANN}, which renders well-studied ANN techniques/tools ineffective.
For instance, the LSH-based ANN scheme (e.g., \cite{AndoniI06}) achieves a query time of $n^{1/(1+\epsilon)^2} \approx n^{0.98}$, which is only marginally sublinear. As a result, the aforementioned amortized update would still result in a total running time of $n^{1.98}$, since each of the $\tilde{O}(n)$ reassignments requires at least one ANN query.

Interestingly, we observe that the shortest-path metric of a graph, despite lacking efficient ANN algorithms, turns out to be well-suited for our amortized analysis.
In graphs, a simple yet powerful fact 
is that each cluster (under the nearest-neighbor assignment with respect to the shortest-path metric) forms a connected component.
To see why this is useful, suppose we wish to insert a center $\cinsert$. 
Then a crucial goal is to identify the newly formed cluster and assign $\DSc[x] = \cinsert$ for all its points, with running time proportional to the cluster size.
Intuitively, this can be implemented  by a local exploration of a connected component starting from $\cinsert$,
similar to breadth-first search/Dijkstra's algorithm.

Luckily, it is possible to turn Euclidean spaces into shortest-path metrics by constructing a Euclidean spanner \cite{HIS13}.
Hence, from now on, we shift our focus from clustering in Euclidean space to clustering in the shortest-path metric of a general weighted undirected graph $G = (V, E, w)$, and we simply consider the dataset to be the vertex set itself, i.e., $X = V$.
We also assume that the graph has constant maximum degree. We note that this assumption is made solely for ease of presentation, and our formal proof does not rely on it.

\paragraph{Leveraging bounded hop-diameter.} 
We now turn to the description of our data structure, which builds upon another technical advantage of graphs -- one can reduce the {\em hop-diameter} of a graph with little overhead.
Roughly speaking, the hop-diameter $\beta = \beta(G)$ of the graph $G$ is the maximum number of edges on any shortest path between two vertices, and without loss of generality, we may assume $\beta = m^{o(1)}$ since we can always preprocess the graph using \cite[Theorem~3.8]{ElkinN19} to ensure this. We note that this is exactly where the $m^{o(1)}$ factor in \Cref{thm:main-informal} originates.

In our design of the ANN data structure, the hop-diameter is leveraged to establish the following stronger Condition~\eqref{eq:techoverview:relaxation-property}, which we call {\em edge relaxation property}, in place of Condition~\eqref{eq:techoverview:ANN}. 
The main advantage of considering Condition~\eqref{eq:techoverview:relaxation-property} is that it is easier for the algorithm to check.
\begin{align}
    \label{eq:techoverview:relaxation-property}
    \dist(u, \DSc[u])
    ~\le~ (1 + \tfrac{\eps}{2\beta})\cdot (\dist(v, \DSc[v]) + w(u, v)),
    \qquad \forall (u, v) \in E.
\end{align}
Condition~\eqref{eq:techoverview:relaxation-property} does imply Condition~\eqref{eq:techoverview:ANN}; given any shortest path from a vertex $v \in V$ to its {\em exact} nearest center $c_v \in C$ (of length at most $\beta$), an induction of Condition~\eqref{eq:techoverview:relaxation-property} along this path implies that
$\dist(v, \DSc[v]) \le (1 + \tfrac{\eps}{2\beta})^{\beta} \cdot \dist(v, c_{v}) \le (1 + \eps) \cdot \dist(v, C)$.
On the other hand, the stronger Condition~\eqref{eq:techoverview:relaxation-property} concerns only a local edge-level property, and is therefore technically more tractable than Condition~\eqref{eq:techoverview:ANN}.

We also note that this idea of designing algorithms parameterized by the hop-diameter $\beta$ had been also studied for many other graph problems and settings \cite{Bernstein09, Nanongkai14, MillerPVX15, HenzingerKN16, HenzingerKN18, ElkinN19, AndoniSZ20, ElkinS23}.

\paragraph{Handling center updates.} 
We are ready to describe how our ANN data structure handles insertions and deletions, while maintaining Condition~\eqref{eq:techoverview:relaxation-property} with running time almost-linear in the number of reassignments.
Still, let us first consider a deletion of center $\cdelete$. We run an adapted {\em Dijkstra's algorithm}:
\begin{itemize}
    \item In the initialization, we ``erase'' the maintained ANN $\DSc[v]$ for every vertex $v \in \calX(\cdelete)$ by setting $\DSc[v] \gets \perp$; that is, we mark all vertices in $\calX(\cdelete)$ as ``unassigned''. After this step, all edges $(u, v) \in E$ that violate Condition~\eqref{eq:techoverview:relaxation-property} must be adjacent to $\calX(\cdelete)$.
    \item To retain Condition~\eqref{eq:techoverview:relaxation-property} for these violating edges $(u, v)$, we regard the boundary $\partial \calX(\cdelete)$ -- namely all vertices adjacent to but outside the cluster of $\cdelete$ -- as the sources, and computes the distance $\Hat{d}[u]$ to the sources for every vertex $u\in \calX(\cdelete)$. 
    Whenever the distance $\Hat{d}[u]$ is updated from a neighbor $v$ (i.e., $\Hat{d}[u] \gets \Hat{d}[v] + w(u, v)$), we also propagate the ANN assignment by setting $\DSc[u] \gets \DSc[v]$, which progressively retains Condition~\eqref{eq:techoverview:relaxation-property}.

\end{itemize}
An important adaptation here is to replace Dijkstra’s update rule with the negation of Condition~\eqref{eq:techoverview:relaxation-property}, specifically $\Hat{d}[u] > (1 + \frac{\epsilon}{2\beta}) \cdot (\Hat{d}[v] + w(u, v))$.    
This adaptation is crucial, as it guarantees that Dijkstra's algorithm will only explore the region $\calX(\cdelete) \cup \partial \calX(\cdelete)$ and can therefore be run in time $\tilde{O}(|\calX(\cdelete)|)$, assuming a bounded maximum degree (again, this degree assumption is not necessary for the formal proof).

Also, we emphasize that the above discussion only delivers the high-level algorithmic idea but has omitted many technical details.
For example, the rigorous version of Condition~\eqref{eq:techoverview:relaxation-property} replaces the distances $\dist(u, \DSc[u])$ and $\dist(v, \DSc[v])$ with certain approximate surrogates $\DSd[u]$ and $\DSd[v]$, so as to make sure that our {\em work-space-restricted} algorithm indeed maintains Condition~\eqref{eq:techoverview:relaxation-property}.
(Those surrogates $(\DSd[v])_{v \in V}$ will also be maintained by our ANN data structure.)

Finally, for an insertion operation, we could similarly apply an adapted Dijkstra's algorithm (or even simplify the procedure to a depth-first search) to perform the required ``lazy'' update, and we omit the detailed discussion here.

\subsection{Further related works}
\label{subsec:related-works}

Polynomial (not necessarily sub-quadratic) time approximation algorithms for Euclidean \kMeans
have been extensively studied.
Due to the Euclidean structure, the status of upper and lower bounds differ significantly from that in the general metrics.

On the upper bound side, the state-of-the-art for \kMeans in general metrics is a $\approx 9$-approximation via primal-dual~\cite{ANSW20}. However, by leveraging the properties of Euclidean space, the same work \cite{ANSW20} obtains a better $\approx 6.36$-approximation for Euclidean $k$-Means, and the followup works \cite{GrandoniORSV22, Cohen-AddadEMN22} further improve this ratio to the state-of-the-art $\approx 5.91$-approximation.
Moreover, Euclidean \kMeans admits $\mathrm{PTAS}$ (or $\mathrm{EPTAS}$) when either the number of dimensions $d$ is fixed 
\cite{Cohen-Addad18, FriggstadRS19, Cohen-AddadKM19, Cohen-AddadFS21},
or the number of centers $k$ is fixed \cite{KumarSS04,FeldmanMS07,DBLP:journals/siamcomp/Chen09,DBLP:conf/focs/AbbasiBBCGKMSS23}.
There is also a line of research that focuses on near-linear $\tilde{O}(nd)$ running time,
however currently they come at a cost of super-constant $\Omega(\polylog(k))$ approximation ratio~\cite{CLNSS20, CharikarHHVW23}.

On the lower bound side, in general metrics, \kMeans is
$\mathrm{NP}$-hard to approximate within a factor better than $4$ \cite{Cohen-AddadSL21}. 
For Euclidean \kMeans, it has been shown to be $\mathrm{APX}$-hard when the dimension $d = \Omega(\log n)$ is large \cite{AwasthiCKS15}. 
After a series of improvements \cite{LeeSW17, Cohen-AddadS19, Cohen-AddadSL21, Cohen-AddadSL22}, the state-of-the-art lower bounds are $1.06$ assuming $\mathrm{P} \ne \mathrm{NP}$ \cite{Cohen-AddadSL22}, $1.07$ assuming the Unique Games Conjecture \cite{Cohen-AddadS19}, and $1.36$ assuming the Johnson Coverage Hypothesis \cite{Cohen-AddadSL22}.
For more hardness results, the interested reader can refer to \cite{Cohen-AddadSL22} and its references.

\endgroup

\section{Preliminaries}
\label{sec:prelim}

{\bf The metric {\kzC} problem.}
Given an underlying metric space $(V,\dist)$, and a point set $X\subseteq V$,
the metric {\kzC} problem,
with parameters $k \ge 1$ and $z \ge 1$, 
aims to find a size-$k$ {\em center set} $C \in V^{k}$ to minimize the following {\em clustering objective} $\cost(X, C)$.
\begin{align*}
    \cost(X, C)
    ~\eqdef~ \sum_{v \in X} \dist^{z}(v, C).
\end{align*}
Here, $\dist(v, C) \eqdef \min_{c \in C} \dist(v, c)$ is the distance from a vertex $v \in V$ to its nearest center $c \in C$, and $\dist^{z}(v, C) \eqdef (\dist(v, C))^{z}$ denotes its $z$-th power.
Likewise, given a real-valued array $(\DSd[v])_{v \in V}$ (say), we will denote by $(\DSd^{z}[v])_{v \in V}$ the array of entry-wise $z$-th powers.

For ease of presentation, we assume without loss of generality that all pairwise distances between points are distinct;\footnote{This can be ensured by adding a small noise to every distance or by using a secondary key, such as the index of every point or vertex, for tie-breaking.}
given a (generic) solution $C \in V^{k}$, we can thus identify the {\em unique} nearest center $\argmin_{c \in C} \dist(v, c) \in C$ of every vertex $v \in V$ in this solution $C \in V^{k}$. Likewise, we assume that there is a unique {\em optimal solution} $C^{*} \eqdef \argmin_{C \in V^{k}} \cost(X, C)$, and denote by $\OPT(X) \eqdef \cost(X, C^{*})$ the \emph{optimal objective}.

\vspace{.1in}
\noindent
{\bf Clustering on graphs.}
Without loss of generality, every graph $G = (V, E, w)$ to be considered is {\em undirected}, {\em connected}, and {\em non-singleton}, and its edges $(u, v) \in E$ have {\em nonnegative} weights $w(u, v) \ge 0$.
Such a graph $G$ induces a metric space $(V, \dist)$ based on the pairwise {\em shortest-path} distances $\dist(u, v) \ge 0$ for all $(u, v) \in V \times V$, and we consider the metric \kzC problem on this shortest-path metric, with the dataset to be clustered, $X = V$, being the entire vertex set. Without ambiguity, in this graph clustering context, we will simplify the notation by letting $\OPT \eqdef \min_{C \in V^k} \cost(V, C)$ denote the optimal objective.

We assume that the given graph $G$ is represented using an {\em adjacency list};
so for every vertex $v \in V$, we can access its {\em degree} $\deg(v)$ in time $O(1)$ and its {\em set of adjacent vertices} $N(v)$ in time $O(\deg(v))$.
As usual, we denote by $n = |V|$ the number of vertices and by $m = |E| \ge n - 1$ the number of edges.
In addition, we introduce the {\em hop-boundedness} (\Cref{def:hop_bounded}) of graphs, a concept crucial to all later materials.

\begin{definition}[Hop-boundedness]
\label{def:hop_bounded}
\begin{flushleft}
A graph $G = (V, E, w)$ is called {\em $(\beta, \eps)$-hop-bounded}, where parameters $\beta \ge 1$ and $0 < \eps < 1$, when every pair of vertices $u, v \in V$ admits a path that visits at most $\beta$ edges and has length $\le 2^{\eps} \cdot \dist(u, v)$.
\end{flushleft}
\end{definition}

The following \Cref{lem:triangle} (which restates \cite[Lemma~A.1]{MMR19}) provides a generalization of {\em triangle inequalities} and will be useful in many places.

\begin{proposition}[{Generalized triangle inequalities \cite{MMR19}}]
\label{lem:triangle}
\begin{flushleft}
Given any $a, b \ge 0$ and any $z \ge 1$, $(a + b)^{z} \le (1 + \lambda)^{z - 1} \cdot a^{z} + (1 + 1 / \lambda)^{z - 1} \cdot b^{z}$, for any parameter $\lambda > 0$.
\end{flushleft}
\end{proposition}

\noindent
{\bf Preprocessing.}
The following \Cref{prop:naive_solution} (whose proof is deferred to \Cref{subsec:naive_solution}) shows that it is easy to find a coarse $n^{z + 1}$-approximate solution $\Cinitial$ to the {\kzC} problem.
This naive solution is useful to many steps in our algorithm, such as to initialize our local search.

\begin{proposition}[Coarse approximation]
\label{prop:naive_solution}
\begin{flushleft}
An $n^{z + 1}$-approximate feasible solution $\Cinitial \in V^{k}$ to the {\kzC} problem can be found in time $O(m \log(n))$.
\end{flushleft}
\end{proposition}

\noindent

Also, for ease of presentation, we would impose \Cref{assumption:edge-weight} throughout \Cref{sec:DS,sec:LS}.
Rather, the following \Cref{prop:edge-weight} (whose proof is deferred to \Cref{subsec:edge-weight}) shows how to avoid the reliance on \Cref{assumption:edge-weight}.

\begin{assumption}[Edge weights]
\label{assumption:edge-weight}
\begin{flushleft}
Every edge $(u, v) \in E$ has a bounded weight $w(u, v) \in [\wmin, \wmax]$, where (up to scale) the parameters $\wmin = 1$ and $\wmax \le n^{O(z)}$.
\end{flushleft}
\end{assumption}

\begin{proposition}[Removing \Cref{assumption:edge-weight}]
\label{prop:edge-weight}
\begin{flushleft}
For any $0 < \eps < 1$, a graph $G = (V, E, w)$
can be converted into a new graph $G' = (V, E, w')$
in time $O(m \log(n))$, such that:
\begin{enumerate}[font = {\em\bfseries}]
    \item \label{prop:edge-weight:bound}
    $G'$ satisfies \Cref{assumption:edge-weight}
    with parameters
$\wmin = 1$ and $\wmax \le 32z^{2} \eps^{-2} n^{z + 5}$.
    
    \item \label{prop:edge-weight:solution}
    Any $\alpha \in [1, n^{z + 1}]$-approximate solution $C$ to {\kzC} for $G'$
is a $2^{\eps}\alpha$-approximate solution
    to that for $G$.
\end{enumerate}
\end{flushleft}
\end{proposition}

\newcommand{\BST}{\calT}

\section{Data Structure for Local Search}
\label{sec:DS}

In this section, we present our data structures for local search, which will be denoted by $\calD$ or variant notations.
We will elaborate in \Cref{subsec:content} the contents of such a data structure $\calD$ and, for ease of reference, show in \Cref{fig:DS} a list of these contents.
Such a data structure $\calD$ can interact with our local search algorithm later in \Cref{sec:LS} only through the following four operations:\footnote{For ease of notation, we simply write $C + c' = C \cup \{c'\}$ and $C - c'' = C \setminus \{c''\}$ in the remainder of this paper.}
\begin{flushleft}
\begin{itemize}
    \item (\Cref{subsec:initialize}) $\Initialize(\Cinitial)$:
    On input an initial feasible solution $\Cinitial \in V^{k}$ (promised), this operation will initialize the data structure $\calD$; among other contents, the {\em center set} $C$ maintained by $\calD$ will be initialized to the input solution $C \gets \Cinitial$.

    \item (\Cref{subsec:insert}) $\Insert(\cinsert)$:
    On input a {\em current noncenter} $\cinsert \notin C$ (promised), this operation will insert $\cinsert \notin C$ into the maintained center set $C$, namely $C \gets C + \cinsert$, and modify other contents maintained by $\calD$ accordingly.
    
    \item (\Cref{subsec:delete}) $\Delete(\cdelete)$:
    On input a {\em current center} $\cdelete \in C$ (promised), this operation will delete $\cdelete \in C$ from the maintained center set $C$, namely $C \gets C - \cdelete$, and modify other contents maintained by $\calD$ accordingly.
    
    \item (\Cref{subsec:sample-noncenter}) $r \gets \SampleNoncenter()$:
    This operation will sample a random vertex $r \in V$ and return it (without modifying the data structure $\calD$); the distribution of $r \in V$ relies on the contents of $\calD$ but, most importantly, ensures that it is a noncenter $r \notin C$ almost surely.
\end{itemize}
\end{flushleft}
We emphasize that everything about our data structures $\calD$ is deterministic, except for the randomized operation $r \gets \SampleNoncenter()$.

\afterpage{
\begin{figure}[t]
\centering
\begin{mdframed}
{\bf Data Structure for Local Search.}

\vspace{.1in}
Our data structure, denoted by $\calD$, will be implemented based on an {\em isolation set cover} $\calJ \subseteq 2^{V}$, which is a certain collection of $|\calJ| = O(\log(n))$ vertex subsets $J \subseteq V$ that meets the conditions in \Cref{def:cover}; without ambiguity, a vertex subset $J \in \calJ$ will be called an {\em index}.

Then, our isolation-set-cover-based data structure $\calD$ will maintain the following contents.
\begin{flushleft}
\begin{itemize}
    \item $C \subseteq V$:
    a maintained {\em center set}.
    
    \item $C_{J} \eqdef C \cap J$:
    a maintained {\em center subset}, for every index $J \in \calJ$.
    
    \Comment{\Cref{def:cover,lem:cover} will ensure that $(\cup_{J \in \calJ} J = V) \implies (\cup_{J \in \calJ} C_{J} = C)$.}
    
    \item $(\DSc_{J}[v], \DSd_{J}[v])_{v \in V}$:
    a {\em subclustering} of all vertices $v \in V$, for every index $J \in \calJ$.
    
    \Comment{$\DSc_{J}[v] \in C_{J}$ identifies which {\em subcluster} contains vertex $v$, while $\DSd_{J}[v] \approx \dist(v, \DSc_{J}[v])$\\
    will approximate (imperfectly but well enough) vertex $v$'s distance to that center $\DSc_{J}[v]$.}
    
    \item $(\DSc[v], \DSd[v])_{v \in V}$: a {\em clustering} of all vertices $v \in V$.
    
    \Comment{$(\DSc[v], \DSd[v])$ will approximate (imperfectly but well enough) the {\em optimal/minimum subcluster}, among all the considered ones $(\DSc_{J}[v], \DSd_{J}[v])_{J \in \calJ}$.}
    
    \Comment{$\DSc[v] \in \cup_{J \in \calJ} C_{J} = C$ identifies which {\em cluster} contains vertex $v$, while $\DSd[v] \approx \dist(v, \DSc[v])$\\
    will approximate (imperfectly but well enough) vertex $v$'s distance to that center $\DSc[v]$.}
    
    \item $\BST$: a {\em binary search tree} \cite[Chapter~12]{CLRS22} of (say) size $|\BST| \le 2^{\lceil \log_{2}(n) \rceil + 1} - 1 = O(n)$ and height $= \lceil \log_{2}(n) \rceil = O(\log(n))$.
    
    \Comment{We will leverage this binary search tree $\BST$ to efficiently implement the randomized operation $\SampleNoncenter()$, akin to \cite[Lemma~4.2]{CLNSS20}.}
    
    \item $\DScost \eqdef \sum_{v \in V} \DSd^{z}[v]$: an {\em objective estimator} for the maintained center set $C \subseteq V$.
    
    \Comment{$\DScost$ will approximate (imperfectly but well enough) the actual objective $\cost(V, C)$ of the maintained center set $C \subseteq V$.}
    
    \item $(\DSloss[c], \DSvolume[c])_{c \in C}$: a {\em deletion estimator} of all centers $c \in C$.
    
    \Comment{Suppose that we would delete a current center $c \in C$ (promised) from the center set $C$ and reassign vertices in the current center-$c$ cluster to other survival centers $c' \in C - c$:\\
    $\DSloss[c] \geq 0$ will approximate the {\em change} of the objective $\cost(V, C)$ by this deletion.\\
    $\DSvolume[c] \geq 0$ will measure the {\em running time} to modify our data structure $\calD$ by this deletion.}
    
    \item $\{(G_{\tau}, \calG_{\tau})\}_{\tau \in [\CGnumber]}$: a {\em grouping} of all centers $c \in C$ into a number of $\CGnumber = O(\log(n))$ groups.
    
    \Comment{The disjoint {\em center groups} $\{G_{\tau}\}_{\tau \in [\CGnumber]}$ form a partition of the center set $C = \cup_{\tau \in [\CGnumber]} G_{\tau}$ (see \Cref{subsec:content} for the grouping criterion).
    Every $\calG_{\tau}$ for $\tau \in [\CGnumber]$ is a {\em red-black tree} of size $|\calG_{\tau}| = |G_{\tau}| \le |C|$ that keeps track of the index-$\tau$ center group $G_{\tau}$ \cite[Chapter~13]{CLRS22}}.
\end{itemize}
The {\em potential} $\Phi \eqdef \sum_{(J,\, v) \in \calJ \times V} \deg(v) \cdot \log_{2}(1 + \DSd_{J}[v])$ will help in establishing the performance guarantees of our data structure $\calD$ (although we need not maintain this {\em potential} $\Phi \ge 0$).
\end{flushleft}
\end{mdframed}
\caption{\label{fig:DS}A list of the contents of our data structure $\calD$.}
\end{figure}
\clearpage}

Regarding the underlying graph $G = (V, E, w)$, we would impose \Cref{assumption:DS:hop-bounded} (in addition to \Cref{assumption:edge-weight} about edge weights) throughout \Cref{sec:DS}.
Then, \Cref{prop:distance} follows directly.

\begin{assumption}[Hop-boundedness]
\label{assumption:DS:hop-bounded}
\begin{flushleft}
The underlying graph $G = (V, E, w)$ is $(\beta, \eps)$-hop-bounded, with known parameters $\beta \in [n]$ and $0 < \eps < 1$.
\end{flushleft}
\end{assumption}

\begin{proposition}[Bounded distances]
\label{prop:distance}
\begin{flushleft}
Every pair of vertices $(u, v \in V \colon u \ne v)$ has a bounded distance $\dist(u, v) \in [\dmin, \dmax]$, where the parameters $\dmin \eqdef \wmin = 1$ and $\dmax \eqdef 2^{\eps} \cdot \beta \cdot \wmax \leq n^{O(z)}$, provided \Cref{assumption:edge-weight,assumption:DS:hop-bounded}.
\end{flushleft}
\end{proposition}

\subsection{Contents of our data structure}
\label{subsec:content}

This subsection elaborates on the contents maintained by our data structure $\calD$, and we provide a list of these contents in \Cref{fig:DS}.
However, before all else, we shall introduce the notion of {\em isolation set cover} $\calJ$.
An isolation set cover $\calJ$ is not a content of $\calD$; rather, our data structure $\calD$ will be implemented based on it.

\subsection*{The isolation set cover $\calJ$}

An {\em isolation set cover} $\calJ \subseteq 2^{V}$ is a collection of at most $|\calJ| = O(\log(n))$ vertex subsets $J \subseteq V$ that satisfies the conditions in \Cref{def:cover}.

\begin{definition}[Isolation set cover]
\label{def:cover}
\begin{flushleft}
An {\em isolation set cover} $\calJ \subseteq 2^{V}$ is a collection of at most $|\calJ| = O(\log(n))$ vertex subsets $J \subseteq V$ that
(i)~every pair of vertices $v \ne u \in V$ can be {\em isolated}, $(v \in J) \wedge (u \notin J)$ for some vertex subset $J \in \calJ$, and
(ii)~all vertices can be {\em covered}, $\cup_{J \in \calJ} J = V$.
\end{flushleft}
\end{definition}

\noindent
The following \Cref{lem:cover} and its proof explicitly construct such an isolation set cover $\calJ \subseteq 2^{V}$.

\begin{lemma}[Isolation set cover]
\label{lem:cover}
\begin{flushleft}
An isolation set cover $\calJ$ can be found in time $O(n \log(n))$.
\end{flushleft}
\end{lemma}

\begin{proof}
We can identify every vertex $v \in V$ by a binary string $v[\cdot] \in \{0, 1\}^{\lceil \log_{2}(n) \rceil}$ and then construct $2\lceil \log_{2}(n) \rceil$ vertex subsets $J_{i,\, b} = \{v \in V \mid v[i] = b\}$, for $1 \leq i \leq \lceil \log_{2}(n) \rceil$ and $b \in \{0, 1\}$; let $\calJ$ be the collection of these $J_{i,\, b}$'s.
This construction clearly takes time $O(n \log(n))$ and $\calJ$ satisfies all desired properties, i.e., every pair of vertices $u \ne v \in V$ must differ $u[i] \ne v[i]$ in some bit $1 \leq i \leq \lceil \log_{2}(n) \rceil$, thus being isolated by either vertex subset $J_{i,\, u[i]}$ or $J_{i,\, v[i]}$. This finishes the proof.
\end{proof}

Our data structure $\calD$ is constructed based on this isolation set cover $\calJ$, i.e., its defining conditions in \Cref{def:cover} will help in efficiently maintaining the contents of data structure $\calD$, under the operations {\Initialize}, {\Insert}, and {\Delete}.\footnote{\label{footnote:cover}Can we build the clustering $(\DSc[v], \DSd[v])_{v \in V}$ directly, rather than indirectly, based on the isolation set cover $\calJ$ and the subclusterings $(\DSc_{J}[v], \DSd_{J}[v])_{(J,\, v) \in \calJ \times V}$?
The answer is yes, and that approach may shave a $\log(n)$ or $\log^{2}(n)$ factor from the running time of our algorithm. However, it will also significantly complicate the design and analysis of our data structure $\calD$, so we decide to adopt the current approach for ease of presentation.}
Further, without ambiguity, we would call every vertex subset $J \in \calJ$ an {\em index} in the remainder of \Cref{sec:DS}.

\subsection*{The center set $C$ and the center subsets $\{C_{J}\}_{J \in \calJ}$}

Firstly, our data structure $\calD$ maintains a {\em center set} $C \subseteq V$ and, for every index $J \in \calJ$, a {\em center subset} $C_{J} \eqdef C \cap J$; we observe that $\cup_{J \in \calJ} C_{J} = C$ and $\cup_{C_{J} \notni c} C_{J}
= C - c$, for every center $c \in C$.\footnote{\label{footnote:center-subset}We have $\cup_{J \in \calJ} C_{J} = \cup_{J \in \calJ} (C \cap J) = C \cap V = C \impliedby \cup_{J \in \calJ} J = V$ (\Cref{def:cover}) and, for every center $c \in C$, $\cup_{C_{J} \notni c} C_{J}
= C \cap (\cup_{J \notni c} J)
= C \cap (V - c)
= C - c \impliedby \cup_{J \notni c} J = V - c$ (\Cref{def:cover}).}
These {\em isolation-set-cover-based} center subsets $\{C_{J}\}_{J \in \calJ}$ will simplify the implementation of our data structure $\calD$. In more detail:

\begin{remark}[Center subsets]
By deleting a current center $c \in C$ from the maintained center set $C$, every vertex $v$ in the center-$c$ cluster shall move to another center-$c'$ cluster; in spirit, they are the nearest $c_{1}[v] = c$ and second-nearest $c_{2}[v] = c'$ centers of $v$.

Suppose that we have identified for a specific vertex $v \in V$ its nearest center $c_{J}[v]$ in every center subset $C_{J}$, $\forall J \in \calJ$, and would further identify its nearest $c_{1}[v]$ and second-nearest $c_{2}[v]$ centers in the whole center set $C$.
The underlying {\em isolation set cover} $\calJ$ (\Cref{def:cover}) ensures that:\\
(i)~$c_{1}[v]$ must be the the nearest one among $\{c_{J}[v]\}_{J \in \calJ} \impliedby \cup_{J \in \calJ} C_{J} = C$, and\\
(ii)~$c_{2}[v]$ must be the nearest one among $\{c_{J}\}_{J \in \calJ \colon C_{J} \notin c} \impliedby \cup_{C_{J} \notni c} C_{J}
= C - c$.\\
So we can easily identify the nearest $c_{1}[v]$ and second-nearest $c_{2}[v]$ centers by enumerating all indices $J \in \calJ$ and all indices $(J \in \calJ \colon C_{J} \notin c)$, respectively, in time $O(|\calJ|) = O(\log(n))$.

In sum, this approach gets rid of (the harder task of) finding second-nearest centers $c_{2}[v]$ and reduces to (the easier task of) finding nearest centers $\{c_{J}[v]\}_{J \in \calJ}$, which will simplify the design and analysis of our data structure $\calD$ (cf.\ \Cref{footnote:cover}).
Also, the actual simplification, as we will consider {\em approximate} nearest and second-nearest centers, will be more significant.
\end{remark}

\subsection*{The subclusterings $(\DSc_{J}[v], \DSd_{J}[v])_{(J,\, v) \in \calJ \times V}$}

Secondly, our data structure $\calD$ maintains a {\em subclustering} $(\DSc_{J}[v], \DSd_{J}[v])_{v \in V}$, for every index $J \in \calJ$.
I.e., regarding every index-$(J \in \calJ)$ center subset $C_{J} \subseteq C$, we maintain the vertex-wise {\em approximate nearest centers} $\DSc_{J}[v] \in C_{J}$, for which $\dist(v, \DSc_{J}[v]) \gtrsim \dist(v, C_{J})$, and the vertex-wise {\em approximate distances} $\DSd_{J}[v] \gtrsim \dist(v, \DSc_{J}[v])$ to those approximate nearest centers $\DSc_{J}[v]$.
(Here, although we have twofold approximations, both will be accurate enough; see \Cref{lem:subclusterings}.)

\begin{invariant*}
\begin{flushleft}
For every index-$(J \in \calJ)$ subclustering $(\DSc_{J}[v], \DSd_{J}[v])_{v \in V}$:\\
In case of a nonempty maintained center subset $C_{J} \ne \emptyset$:
\begin{enumerate}
    \item[{\bf \term[A1]{invar:subclusterings:center}.}]
    $(\DSc_{J}[c], \DSd_{J}[c]) = (c, 0)$, for every center $c \in C_{J}$.
    
    \item[{\bf \term[A2]{invar:subclusterings:edge}.}]
    $\DSd_{J}[u] \le 2^{\eps / \beta} \cdot (\DSd_{J}[v] + w(u, v))$, for every edge $(u, v) \in E$.
    
    \item[{\bf \term[A3]{invar:subclusterings:vertex}.}]
    $\DSc_{J}[v] \in C_{J}$ and $\dist(v, C_{J}) \le \dist(v, \DSc_{J}[v]) \le \DSd_{J}[v]$, for every vertex $v \in V$.
\end{enumerate}
In case of an empty maintained center subset $C_{J} = \emptyset$:
\begin{enumerate}
\setcounter{enumi}{3}
    \item[{\bf \term[A4]{invar:subclusterings:empty}.}]
    $(\DSc_{J}[v], \DSd_{J}[v]) = (\perp, \dmax)$, for every vertex $v \in V$.\footnote{The notation $\perp$ represents a badly-defined ``center'', and we let $\dist(v, \perp) = +\infty$.}
\end{enumerate}
\end{flushleft}
\end{invariant*}

The following \Cref{lem:subclusterings} shows that the subclusterings $(\DSc_{J}[v], \DSd_{J}[v])_{v \in V}$ can well approximate the ``groundtruth'', provided Invariants~\ref{invar:subclusterings:center} to \ref{invar:subclusterings:empty} (and \Cref{assumption:edge-weight,assumption:DS:hop-bounded}).

\begin{lemma}[Maintained subclusterings]
\label{lem:subclusterings}
$\DSd_{J}[v] \le \min(\dmax,\ 2^{2\eps} \cdot \dist(v, C_{J}))$, for every entry $(J, v) \in \calJ \times V$, provided Invariants~\ref{invar:subclusterings:center} to \ref{invar:subclusterings:empty} (as well as \Cref{assumption:edge-weight,assumption:DS:hop-bounded}).
\end{lemma}

\begin{proof}
The case of an empty maintained center subset $C_{J} = \emptyset$ is trivial (Invariant~\ref{invar:subclusterings:empty}); below we address the other case of a nonempty maintained center subset $C_{J} \ne \emptyset$.

Without loss of generality, let us consider a specific vertex $v \in V$ and its nearest center $c_{J,\, v}^{*} \eqdef \argmin_{c \in C_{J}} \dist(v, c) \in C_{J}$.
Since our graph $G = (V, E, w)$ is $(\beta, \eps)$-hop-bounded (\Cref{assumption:DS:hop-bounded}), there exists a $v$-to-$c_{J,\, v}^{*}$-path $(v \equiv x_{0}), x_{1}, \dots, (x_{\beta'} \equiv c_{J,\, v}^{*})$  with $\beta' \le \beta$ edges such that
\begin{align}
\label{eq:subclusterings}
    \sum_{i \in [\beta']} w(x_{i - 1}, x_{i})
    ~\le~ 2^{\eps} \cdot \dist(v, c_{J,\, v}^{*})
    ~=~ 2^{\eps} \cdot \dist(v, C_{J}).
\end{align}
Moreover, Invariant~\ref{invar:subclusterings:edge} ensures that, for every $i \in [\beta']$:
\begin{align*}
    2^{(\eps / \beta) \cdot (i - 1)} \cdot \DSd_{J}[x_{i - 1}]
    ~\le~ 2^{(\eps / \beta) \cdot i} \cdot (\DSd_{J}[x_{i}] + w(x_{i - 1}, x_{i}))
    ~\le~ 2^{(\eps / \beta) \cdot i} \cdot \DSd_{J}[x_{i}] + 2^{\eps} \cdot w(x_{i - 1}, x_{i}).
\end{align*}
Since $\DSd_{J}[x_{\beta'}] \equiv \DSd_{J}[c_{J,\, v}^{*}] = 0$ (Invariant~\ref{invar:subclusterings:center}), we infer from an induction over $i \in [\beta']$ that
\begin{align*}
    \DSd_{J}[v]
    ~\le~ 2^{(\eps / \beta) \cdot \beta'} \cdot \DSd_{J}[x_{\beta'}] + \sum_{i \in [\beta']} 2^{\eps} \cdot w(x_{i - 1}, x_{i})
    ~=~ \sum_{i \in [\beta']} 2^{\eps} \cdot w(x_{i - 1}, x_{i}).
\end{align*}
This equation, in combination with \Cref{eq:subclusterings,assumption:edge-weight,prop:distance}, implies the claimed bounds $\DSd_{J}[v] \le 2^{2\eps} \cdot \dist(v, C_{J})$ and $\DSd_{J}[v] \le 2^{2\eps} \cdot \beta \cdot \wmax = \dmax$.

This finishes the proof of \Cref{lem:subclusterings}.
\end{proof}

We also remark that all possible modifications of the subclusterings $(\DSc_{J}[v], \DSd_{J}[v])_{v \in V}$ (due to the operations in \Cref{subsec:initialize,subsec:insert,subsec:delete}) will always maintain Invariants~\ref{invar:subclusterings:center} to \ref{invar:subclusterings:empty}.

\begin{remark}[Subclusterings]
Remarkably, we are directly maintaining the stronger Invariant~\ref{invar:subclusterings:edge} rather than directly maintaining (\Cref{lem:subclusterings}) its weaker implication $\DSd_{J}[v] \leq 2^{2\eps} \cdot \dist(v, C_{J})$. This is because Invariant~\ref{invar:subclusterings:edge} is easy to maintain -- we only need to detect edge-wise violations. In contrast, when its weaker implication $\DSd_{J}[v] \leq 2^{2\eps} \cdot \dist(v, C_{J})$ is violated by an entry $(J, v) \in \calJ \times V$, we have to recompute $\DSd_{J}[v]$, which is far less efficient.
\end{remark}

\subsection*{The clustering $(\DSc[v], \DSd[v])_{v \in V}$}

Thirdly, our data structure $\calD$ maintains a {\em clustering} $(\DSc[v], \DSd[v])_{v \in V}$.
Formally:

\begin{invariant*}
\begin{flushleft}
For the clustering $(\DSc[v], \DSd[v])_{v \in V}$:
\begin{enumerate}
    \item[{\bf \term[B]{invar:clustering}.}]
    Every pair $(\DSc[v], \DSd[v])$, for $v \in V$, is the $\DSd_{J}[v]$-minimizer among the pairs $(\DSc_{J}[v], \DSd_{J}[v])_{J \in \calJ}$.\footnote{\label{footnote:break-tie}If two pairs $(\DSc_{J}[u], \DSd_{J}[u])$ and $(\DSc_{J}[v], \DSd_{J}[v])$ have the same $\DSd_{J}[u] = \DSd_{J}[v]$ values, we break ties in favor of the pair with $\DSc_{J}[u], \DSc_{J}[v] \ne \perp$ (if any) but otherwise arbitrarily.}
\end{enumerate}
\end{flushleft}
\end{invariant*}

\noindent
Thus, we are trying to maintain the vertex-wise {\em approximate nearest centers} $\DSc[v] \in C$ in the {\em entire} maintained center set $C$, for which $\dist(v, \DSc[v]) \gtrsim \dist(v, C)$, and the vertex-wise {\em approximate distances} $\DSd[v] \gtrsim \dist(v, \DSc[v])$ to those approximate nearest centers $\DSc[v]$.\textsuperscript{\ref{footnote:cover}}
These are formalized into the following \Cref{lem:clustering}, provided Invariant~\ref{invar:clustering}.

\begin{lemma}[Maintained clustering]
\label{lem:clustering}
\begin{flushleft}
For the clustering $(\DSc[v], \DSd[v])_{v \in V}$, provided Invariant~\ref{invar:clustering} (as well as \Cref{assumption:edge-weight,assumption:DS:hop-bounded} and Invariants~\ref{invar:subclusterings:center} to \ref{invar:subclusterings:empty}):
\begin{enumerate}[font = {\em\bfseries}]
    \item \label{lem:clustering:center}
    $(\DSc[c], \DSd[c]) = (c, 0)$, for every center $c \in C$.
    
    \item \label{lem:clustering:vertex}
    $\DSc[v] \in C$ and $\dist(v, C) \le \dist(v, \DSc[v]) \le \DSd[v] \le 2^{2\eps} \cdot \dist(v, C)$, for every vertex $v \in V$.
\end{enumerate}
\end{flushleft}
\end{lemma}

\begin{proof}
Consider a specific vertex $v \in V$ and its defining index $J_{v} \eqdef \argmin_{J \in \calJ} \DSd_{J}[v]$ for Invariant~\ref{invar:clustering}, namely $(\DSc[v], \DSd[v]) = (\DSc_{J_{v}}[v], \DSd_{J_{v}}[v])$.
The index-$J_{v}$ center subset must be nonempty $C_{J_{v}} \ne \emptyset$, since $\cup_{J \in \calJ} C_{J} = C \ne \emptyset$ (\Cref{def:cover}),\textsuperscript{\ref{footnote:center-subset}}
$\DSd_{J}[v] \le \dmax$ for every index $J \in \calJ$ (\Cref{lem:subclusterings}), where the equality holds whenever $C_{J} = \emptyset$ (Invariant~\ref{invar:subclusterings:empty}), and we break ties in favor of the pairs $(\DSc_{J}[v], \DSd_{J}[v])$ with $\DSc_{J}[v] \ne \perp$ (\Cref{footnote:break-tie}).
Then, we know from Invariant~\ref{invar:subclusterings:vertex} that $\DSc[v] = \DSc_{J_{v}}[v] \in C_{J_{v}} \subseteq C$ and
\begin{align*}
    \DSd[v]
    ~=~ \DSd_{J_{v}}[v]
    ~\ge~ \dist(v, \DSc_{J_{v}}[v])
    ~=~ \dist(v, \DSc[v])
    ~\ge~ \dist(v, C).
\end{align*}
Further, consider the actual nearest center $c_{v}^{*} = \argmin_{c \in C} \dist(v, c)$ in the maintained center set $C$ and a specific center subset $C_{J_{v}^{*}} \ni c_{v}^{*}$ containing it; once again, such a center subset $C_{J_{v}^{*}}$ must exist, given that $\cup_{J \in \calJ} C_{J} = C \ni c_{v}^{*}$.
Then, we can deduce that
\begin{align*}
    \DSd[v]
    ~=~ \DSd_{J_{v}}[v]
    ~\le~ \DSd_{J_{v}^{*}}[v]
    ~\le~ 2^{2\eps} \cdot \dist(v, C_{J_{v}^{*}})
    ~=~ 2^{2\eps} \cdot \dist(v, C).
\end{align*}
Here, the second step applies $J_{v} = \argmin_{J \in \calJ} \DSd_{J}[v]$, and the third step applies \Cref{lem:subclusterings} (notice that $C_{J_{v}^{*}} \ne \emptyset \impliedby C_{J_{v}^{*}} \ni c_{v}^{*}$).
Combining everything together gives \Cref{lem:clustering:vertex}.

Further, \Cref{lem:clustering:center} follows from the observation $\dist(c, C) = 0 \implies \DSd[c] = 0 \implies \DSc[c] \in c$; here, the first step applies \Cref{lem:clustering:vertex}, and the second step applies \Cref{prop:distance}, namely $\dist(c, v) \ge \dmin = 1$ for any other vertex $v \in V - c$.
This finishes the proof of \Cref{lem:clustering}.
\end{proof}

All possible modifications to the clustering $(\DSc[v], \DSd[v])_{v \in V}$ (due to the operations in \Cref{subsec:initialize,subsec:insert,subsec:delete}) will always maintain Invariant~\ref{invar:clustering}.
In this regard, the following \Cref{lem:maintain-clustering} shows that these two contents can be modified efficiently.
\Cref{lem:maintain-clustering} allows us to focus on maintaining the subclusterings $(\DSc_{J}[v], \DSd_{J}[v])_{(J,\, v) \in \calJ \times V}$. I.e., every time when they are modified, we can maintain -- or {\em synchronize} -- the clustering $(\DSc[v], \DSd[v])_{v \in V}$ efficiently.

\begin{lemma}[Synchronization]
\label{lem:maintain-clustering}
\begin{flushleft}
For the clustering $(\DSc[v], \DSd[v])_{v \in V}$:
\begin{enumerate}[font = {\em\bfseries}]
    \item \label{lem:maintain-clustering:initialize}
    It can be built on the subclusterings $(\DSc_{J}[v], \DSd_{J}[v])_{(J,\, v) \in \calJ \times V}$ in time $O(n \log(n))$.
    
    \item \label{lem:maintain-clustering:modify}
    It can be synchronized to maintain Invariant~\ref{invar:clustering} in time $O(\log(n))$, every time when the subclusterings $(\DSc_{J}[v], \DSd_{J}[v])_{(J,\, v) \in \calJ \times V}$ are modified at a single entry $(J, v) \in \calJ \times V$.
\end{enumerate}
\end{flushleft}
\end{lemma}

\begin{proof}
Obvious; every pair $(\DSc[v], \DSd[v])$, as the $\DSd_{J}[v]$-minimizer among the pairs $(\DSc_{J}[v], \DSd_{J}[v])_{J \in \calJ}$, can be built or synchronized by enumerating $(\DSc_{J}[v], \DSd_{J}[v])_{J \in \calJ}$ in time $O(|\calJ|) = O(\log(n))$.
This finishes the proof of \Cref{lem:maintain-clustering}.
\end{proof}

\subsection*{The binary search tree $\BST$ and the objective estimator $\DScost$}

Fourthly, our data structure $\calD$ maintains a {\em binary search tree} $\BST$ \cite[Chapter~12]{CLRS22}, (built from the clustering $(\DSc[v], \DSd[v])_{v \in V}$) and an {\em objective estimator} $\DScost$.
Formally:

\newcommand{\sfvalue}{{\sf value}}

\begin{invariant*}
\begin{flushleft}
For the binary search tree $\BST$ and the objective estimator $\DScost$:
\begin{enumerate}
    \item[{\bf \term[C]{invar:binary-search-tree}.}]
    $\BST = \BST((v, \DSd^{z}[v])_{v \in V})$ is a binary search tree storing abstract pairs $(U, \sfvalue[U])$. Concretely:\\
    (i)~Its every {\em leaf} takes the form $(U, \sfvalue[U]) = (\{v\}, \DSd^{z}[v])$, i.e., corresponding one-to-one to every vertex $v \in V$ and storing the $z$-th power of this vertex $v$'s approximate distance $\DSd^z[v]$.\\
    (ii)~Its every {\em nonleaf} takes the form $(U, \sfvalue[U]) = (U_{\sf left} \cup U_{\sf right}, \sfvalue[U_{\sf left}] + \sfvalue[U_{\sf right}])$, i.e., ``merging'' this nonleaf's left child $(U_{\sf left}, \sfvalue[U_{\sf left}])$ and right child $(U_{\sf right}, \sfvalue[U_{\sf right}])$.\\
    (iii)~We ensure the disjointness $U_{\sf left} \cap U_{\sf right} = \emptyset$, for every nonleaf $(U, \sfvalue[U])$.\footnote{The parent $(U, \sfvalue[U])$ of two leaves $(\{v_{\sf left}\}, \DSd^{z}[v_{\sf left}])$ and $(\{v_{\sf right}\}, \DSd^{z}[v_{\sf right}])$ -- at most one dummy leaf $(\emptyset, 0)$ -- satisfies this property $\{v_{\sf left}\} \cap \{v_{\sf right}\} = \emptyset$, since all leaves $(\{v\}, \DSd^{z}[v])$ correspond one-to-one to all vertices $v \in V$. Then, by induction, other nonleaves $(U, \sfvalue[U])$ can be implemented to maintain this property.}\\
    \Comment{Accordingly, this binary search tree $\BST$ ensures that $\BST.\texttt{root} = (V, \sum_{v \in V} \DSd^{z}[v])$ and can be implemented to have size $|\BST| \le 2^{\lceil \log_{2}(n) \rceil + 1} - 1 = O(n)$ and height $= \lceil \log_{2}(n) \rceil = O(\log(n))$.}
    
    \item[{\bf \term[D]{invar:DScost}.}]
    $\DScost = \BST.\texttt{root}.\texttt{value} = \sum_{v \in V} \DSd^{z}[v]$.
\end{enumerate}
\end{flushleft}
\end{invariant*}

The following \Cref{lem:deletion-estimator} gives useful upper and lower bounds on the objective estimator $\DScost$, provided Invariants~\ref{invar:binary-search-tree} and \ref{invar:DScost}.

\begin{lemma}[Maintained objective estimator]
\label{lem:DScost}
\begin{flushleft}
$\cost(V, C) \le \DScost \le 2^{2\eps z} \cdot \cost(V, C)$, provided Invariants~\ref{invar:binary-search-tree} and \ref{invar:DScost} (as well as \Cref{assumption:edge-weight,assumption:DS:hop-bounded}, Invariants~\ref{invar:subclusterings:center} to \ref{invar:subclusterings:empty}, and Invariant~\ref{invar:clustering}).
\end{flushleft}
\end{lemma}

\begin{proof}
This follows immediately from \Cref{lem:clustering:vertex} of \Cref{lem:clustering}, since $\cost(V, C) = \sum_{v \in V} \dist^{z}(v, C)$ and $\DScost = \sum_{v \in V} \DSd^{z}[v]$ (Invariant~\ref{invar:DScost}).
\end{proof}

All possible modifications to the binary search tree $\BST$ and the objective estimator $\DScost$ (due to the operations in \Cref{subsec:initialize,subsec:insert,subsec:delete}) will always maintain Invariants~\ref{invar:binary-search-tree} and \ref{invar:DScost}.
In this regard, the following \Cref{lem:maintain-DScost} shows that these two contents can be modified efficiently.

\begin{lemma}[Synchronization]
\label{lem:maintain-DScost}
\begin{flushleft}
For the binary search tree $\BST$ and the objective estimator $\DScost$:
\begin{enumerate}[font = {\em\bfseries}]
    \item \label{lem:maintain-DScost:initialize}
    Both can be built on the clusterings $(\DSc[v], \DSd[v])_{v \in V}$ in time $O(n)$.
    
    \item \label{lem:maintain-DScost:modify}
    Both can be synchronized to maintain Invariants~\ref{invar:binary-search-tree} and \ref{invar:DScost} in time $O(\log(n))$, every time when the subclusterings $(\DSc_{J}[v], \DSd_{J}[v])_{(J,\, v) \in \calJ \times V}$ are modified at a single entry $(J, v) \in \calJ \times V$.
\end{enumerate}
\end{flushleft}
\end{lemma}

\begin{proof}
This binary search tree $\BST$ has size $|\BST| \le 2^{\lceil \log_{2}(n) \rceil + 1} - 1 = O(n)$, so we can build it in time $O(|\BST|) = O(n)$ \cite[Chapter~12]{CLRS22}.
When the subclusterings $(\DSc_{J}[v], \DSd_{J}[v])_{(J,\, v) \in \calJ \times V}$ are modified at a single entry $(J, v) \in \calJ \times V$, the clustering $(\DSc[v], \DSd[v])_{v \in V}$ may therefore be modified but, more importantly, can only be modified at this single vertex $v \in V$ (Invariant~\ref{invar:clustering}).
Thus by construction, this binary search tree $\BST$ can be modified to maintain Invariant~\ref{invar:binary-search-tree} in time $O(\log(|\BST|)) = O(\log(n))$ \cite[Chapter~12]{CLRS22}.
Further, the objective estimator $\DScost = \BST.\texttt{root}.\texttt{value}$ will be synchronized to automatically maintain Invariant~\ref{invar:DScost}. This finishes the proof of \Cref{lem:maintain-DScost}.
\end{proof}

\subsection*{The deletion estimator $(\DSloss[c], \DSvolume[c])_{c \in C}$}

Fifthly, our data structure $\calD$ maintains a {\em deletion estimator} $(\DSloss[c], \DSvolume[c])_{c \in C}$ (built from the subclusterings $(\DSc_{J}[v], \DSd_{J}[v])_{(J,\, v) \in \calJ \times V}$).
Formally:

\begin{invariant*}
For the deletion estimator $(\DSloss[c], \DSvolume[c])_{c \in C}$:
\begin{enumerate}
    \item[{\bf \term[E]{invar:DSloss}.}]
    $\DSloss[c] = \sum_{v \in V \colon \DSc[v] = c} ((\min_{J \in \calJ \colon J \notni c} \DSd_{J}^{z}[v]) - \DSd^{z}[v])$, for every center $c \in C$.
    
    \item[{\bf \term[F]{invar:DSvolume}.}]
    $\DSvolume[c] = \frac{1}{2m |\calJ|} \sum_{(J,\, v) \in \calJ \times V \colon \DSc_{J}[v] = c} \deg(v)$, for every center $c \in C$.
\end{enumerate}
\end{invariant*}

The following \Cref{lem:deletion-estimator} gives useful bounds on the deletion estimator $(\DSloss[c], \DSvolume[c])_{c \in C}$, provided Invariants~\ref{invar:DSloss} and \ref{invar:DSvolume}.

\begin{lemma}[Maintained deletion estimator]
\label{lem:deletion-estimator}
\begin{flushleft}
For the deletion estimator $(\DSloss[c], \DSvolume[c])_{c \in C}$, provided Invariants~\ref{invar:DSloss} and \ref{invar:DSvolume} (as well as \Cref{assumption:edge-weight,assumption:DS:hop-bounded}, Invariants~\ref{invar:subclusterings:center} to \ref{invar:subclusterings:empty}, and Invariant~\ref{invar:clustering}):
\begin{enumerate}[font = {\em\bfseries}]
    \item \label{lem:deletion-estimator:DSloss}
    $\DSloss[c] \le \sum_{v \in V \colon \DSc[v] = c} (2^{2\eps z} \cdot \dist^{z}(v, C - c) - \DSd^{z}[v])$, for every center $c \in C$.
    
    \item \label{lem:deletion-estimator:DSvolume}
    $\sum_{c \in C} \DSvolume[c] = 1$ and $\frac{1}{2m |\calJ|} \le \DSvolume[c] \le 1$, for every center $c \in C$.
\end{enumerate}
\end{flushleft}
\end{lemma}

\begin{proof}
Without loss of generality, let us consider a specific center $c \in C$.

\vspace{.1in}
\noindent
{\bf \Cref{lem:deletion-estimator:DSloss}.}
Provided Invariant~\ref{invar:DSloss}, it suffices to prove that $\min_{J \in \calJ \colon J \notni c} \DSd_{J}[v] \le 2^{2\eps} \cdot \dist(v, C - c)$, for every vertex $(v \in V \colon \DSc[v] = c)$, as follows:
\begin{align*}
    \min_{J \in \calJ \colon J \notni c} \DSd_{J}[v]
    ~\le~ \min_{J \in \calJ \colon J \notni c} 2^{2\eps} \cdot \dist(v, C_{J})
    ~=~ 2^{2\eps} \cdot \dist(v, C - c)
\end{align*}
Here, the first step applies \Cref{lem:subclusterings}, and the second step applies $\cup_{C_{J} \notni c} C_{J} = C - c$.\textsuperscript{\ref{footnote:center-subset}}

\vspace{.1in}
\noindent
{\bf \Cref{lem:deletion-estimator:DSvolume}.}
The first part $\sum_{c \in C} \DSvolume[c] = 1$ is a direct consequence of Invariant~\ref{invar:DSvolume}:
\begin{align*}
    \sum_{c \in C} \DSvolume[c]
    ~=~ \sum_{c \in C} \Big(\sum_{(J,\, v) \in \calJ \times V \colon \DSc_{J}[v] = c} \tfrac{\deg(v)}{2m |\calJ|}\Big)
    ~=~ |\calJ| \cdot \Big(\sum_{v \in V} \tfrac{\deg(v)}{2m |\calJ|}\Big) 
    ~=~ 1.
\end{align*}
Then the second part $\frac{1}{2m |\calJ|} \le \DSvolume[c] \le 1$ follows directly, since $\deg(v) \ge 1$ for every vertex $v \in V$ in the considered {\em connected} graph $G$ and $\{(J,\, v) \in \calJ \times V \mid \DSc_{J}[v] = c\} \supseteq \{J \in \calJ \mid C_{J} \ni c\} \times \{c\} \ne \emptyset$, where the first step applies Invariant~\ref{invar:subclusterings:center} and the second step applies $\cup_{J \in \calJ} C_{J} = C \ni c$.\textsuperscript{\ref{footnote:center-subset}}

This finishes the proof of \Cref{lem:deletion-estimator}.
\end{proof}

\begin{remark}[Deletion estimator]
\label{rem:deletion-estimator}
Recall \Cref{lem:subclusterings,lem:clustering} that the approximate distances $\DSd_{J}[v]$ and $\DSd[v]$ well approximate the ``groundtruth''.
Likewise, the deletion objective estimator $\DSloss[c]$ has the same spirit.
Namely, $\DSloss[c] \approx \cost(\cluster_{c}, C - c) - \cost(\cluster_{c}, C) \approx \cost(V, C - c) - \cost(V, C)$ well approximates the change of the objective $\cost(V, C) = \sum_{v \in V} \dist^{z}(v, C)$ by the deletion of a current center $c \in C$ and the reassignment of vertices in the center-$c$ cluster $\cluster_{c} = \{v \in V \mid \DSc[v] = c\}$ to other survival centers $c' \in C - c$.
\end{remark}

All our modifications to the deletion estimator $(\DSloss[c], \DSvolume[c])_{c \in C}$ (due to the operations in \Cref{subsec:initialize,subsec:insert,subsec:delete}) will always maintain Invariants~\ref{invar:binary-search-tree} and \ref{invar:DScost}.
In this regard, \Cref{lem:maintain-deletion-estimator} shows that this deletion estimator $(\DSloss[c], \DSvolume[c])_{c \in C}$ can be modified efficiently.

\begin{lemma}[Synchronization]
\label{lem:maintain-deletion-estimator}
\begin{flushleft}
For the deletion estimator $(\DSloss[c], \DSvolume[c])_{c \in C}$:
\begin{enumerate}[font = {\em\bfseries}]
    \item \label{lem:maintain-deletion-estimator:initialize}
    It can be built on the subclusterings $(\DSc_{J}[v], \DSd_{J}[v])_{(J,\, v) \in \calJ \times V}$ in time $O(n \log(n))$.
    
    \item \label{lem:maintain-deletion-estimator:modify}
    It can be synchronized to maintain Invariants~\ref{invar:DSloss} and \ref{invar:DSvolume} in time $O(\log(n))$, every time when the subclusterings $(\DSc_{J}[v], \DSd_{J}[v])_{(J,\, v) \in \calJ \times V}$ are modified at a single entry $(J, v) \in \calJ \times V$.
\end{enumerate}
\end{flushleft}
\end{lemma}

\begin{proof}
Both \Cref{lem:maintain-deletion-estimator:initialize,lem:maintain-deletion-estimator:modify} are rather obvious.

\vspace{.1in}
\noindent
{\bf \Cref{lem:maintain-deletion-estimator:initialize}.}
Starting from the {\em empty} deletion estimator $(\DSloss[c], \DSvolume[c])_{c \in C} \gets (0, 0)^{|C|}$, let us iterate the following for every vertex $v \in V$:
\begin{itemize}
    \item $\DSloss[\DSc[v]] \gets \DSloss[\DSc[v]] + ((\min_{J \in \calJ \colon J \notni \DSc[v]} \DSd_{J}[v]^{z}) - \DSd^{z}[v])$.\\
    \Comment{This requires finding the minimum $(\min_{J \in \calJ \colon J \notni \DSc[v]} \DSd_{J}[v]^{z})$ and takes time $O(|\calJ|)$.}
    
    \item $\DSvolume[\DSc_{J}[v]] \gets \DSvolume[\DSc_{J}[v]] + \frac{\deg(v)}{2m |\calJ|}$, for every index $J \in \calJ$.\\
    \Comment{This requires enumerating all indices $J \in \calJ$ and also takes time $O(|\calJ|)$.}
\end{itemize}
After the iteration terminates, we get Invariants~\ref{invar:DSloss} and \ref{invar:DSvolume} maintained by construction. Further, the total running time for this initialization $= O(|C|) + n \cdot O(|\calJ|) = O(n \log(n))$.

\vspace{.1in}
\noindent
{\bf \Cref{lem:maintain-deletion-estimator:modify}.}
Regarding every modification of the subclustering, from $(\DSc_{J}[v], \DSd_{J}[v])$ to $(\DSc'_{J}[v], \DSd'_{J}[v])$ (say), at a single entry $(J, v) \in \calJ \times V$, we would synchronize the deletion estimator as follows:
\begin{flushleft}
\begin{itemize}
    \item The clustering can only be modified at {\em one} vertex $v$, from $(\DSc[v], \DSd[v])$ to $(\DSc'[v], \DSd'[v])$ (say).\\
    \Comment{This synchronization takes time $O(\log(n))$, by \Cref{lem:maintain-clustering:modify} of \Cref{lem:maintain-clustering}.}

    \item $\DSloss[\DSc[v]] \gets \DSloss[\DSc[v]] - ((\min_{J \in \calJ \colon J \notni \DSc[v]} \DSd_{J}[v]^{z}) - \DSd^{z}[v])$.\\
    $\DSloss[\DSc'[v]] \gets \DSloss[\DSc'[v]] + ((\min_{J \in \calJ \colon J \notni \DSc'[v]} \DSd'_{J}{}^{z}[v]) + \DSd'{}^{z}[v])$.\\
    \Comment{This requires finding two minimums $(\min_{J \in \calJ \colon J \notni \DSc[v]} \DSd_{J}[v]^{z})$ and $(\min_{J \in \calJ \colon J \notni \DSc'[v]} \DSd'_{J}{}^{z}[v])$, thus taking time $O(|\calJ|) + O(|\calJ|) = O(|\calJ|)$.}
    
    \item $\DSvolume[\DSc_{J}[v]] \gets \DSvolume[\DSc_{J}[v]] - \frac{\deg(v)}{2m |\calJ|}$.\\
    $\DSvolume[\DSc'_{J}[v]] \gets \DSvolume[\DSc'_{J}[v]] + \frac{\deg(v)}{2m |\calJ|}$.\\
    \Comment{Clearly, this takes time $O(1)$.}
\end{itemize}
\end{flushleft}
Afterward, we get Invariants~\ref{invar:DSloss} and \ref{invar:DSvolume} maintained by construction. Also, the total running time for this synchronization $= O(\log(n)) + O(|\calJ|) + O(1) = O(\log(n))$.

This finishes the proof of \Cref{lem:maintain-deletion-estimator}.
\end{proof}

\subsection*{The grouping $\{(G_{\tau}, \calG_{\tau})\}_{\tau \in [\CGnumber]}$}

Sixthly, our data structure maintains a {\em grouping} $\{(G_{\tau}, \calG_{\tau})\}_{\tau \in [\CGnumber]}$ of the maintained centers $c \in C$, where the number of groups $\CGnumber \eqdef \lfloor \log_{2}(2m|\calJ|) + 1\rfloor = O(\log(n))$, based on the deletion estimator $(\DSloss[c], \DSvolume[c]){c \in C}$. Formally:

\begin{invariant*}
\begin{flushleft}
For the grouping $\{(G_{\tau}, \calG_{\tau})\}_{\tau \in [\CGnumber]}$:
\begin{enumerate}
    \item[{\bf \term[G]{invar:grouping:group}.}]
    Every {\em center group} $G_{\tau} = \big\{c \in C \bigmid \lfloor-\log_{2}(\DSvolume[c]) + 1\rfloor = \tau\big\}$, for $\tau \in [\CGnumber]$, includes all centers $c \in C$ whose deletion volume estimator are bounded between $\frac{1}{2^{\tau}} < \DSvolume[c] \le \frac{1}{2^{\tau - 1}}$.
    
    \item[{\bf \term[H]{invar:grouping:tree}.}]
    Every $\calG_{\tau} = \calG_{\tau}((c, \DSloss[c])_{c \in G_{\tau}})$, for $\tau \in [\CGnumber]$, is a {\em red-black tree} storing pairs $(c, \DSloss[c])_{c \in G_{\tau}}$, using the $\DSloss[c]$-minimizing priority \cite[Chapter~13]{CLRS22}.
    \hfill
    \Comment{$|\calG_{\tau}| = |G_{\tau}| \le |C| \le n$.}
\end{enumerate}
\end{flushleft}
\end{invariant*}

\begin{remark}[Grouping]
\label{remark:grouping}
As mentioned, suppose that we would delete a current center $c \in C$ from the maintained center set $C$, roughly speaking
(i)~its deletion loss estimator $\DSloss[c] \ge 0$ can measure the {\em progress} on minimizing the clustering objective $\cost(V, C)$ and
(ii)~its deletion volume estimator $\frac{1}{2m |\calJ|} \le \DSvolume[c] \le 1$ can measure the {\em running time} to modify our data structure $\calD$.

Accordingly, (Invariant~\ref{invar:grouping:group}) every center group $G_{\tau}$ for $\tau \in [\CGnumber]$ includes, up to a factor $\le 2$, those {\em almost equally time-consuming centers $c \in C$}, and (Invariant~\ref{invar:grouping:tree}) its associated red-black tree $\calG_{\tau}$ can help us {\em efficiently identify the most progressive center $c \in G_{\tau}$} therein.
\end{remark}

The following \Cref{lem:grouping} presents useful properties of the grouping $\{(G_{\tau}, \calG_{\tau})\}_{\tau \in [\CGnumber]}$, provided Invariants~\ref{invar:grouping:group} and \ref{invar:grouping:tree}.

\begin{lemma}[Maintained grouping]
\label{lem:grouping}
\begin{flushleft}
$C = \cup_{\tau \in [\CGnumber]} G_{\tau}$ and every red-black tree $\calG_{\tau}$, $\forall \tau \in [\CGnumber]$, supports searching in time $O(\log(n))$, provided Invariants~\ref{invar:grouping:group} and \ref{invar:grouping:tree} (as well as Invariants~\ref{invar:subclusterings:center} to \ref{invar:subclusterings:empty} and Invariant~\ref{invar:DSvolume}).
\end{flushleft}
\end{lemma}

\begin{proof}
The maintained center set $C$ can be covered $C = \cup_{\tau \in [\CGnumber]} G_{\tau}$, because every center $c \in C$ has a deletion volume estimator bounded between $\frac{1}{2m |\calJ|} \le \DSvolume[c] \le 1$ (\Cref{lem:deletion-estimator:DSvolume} of \Cref{lem:deletion-estimator}), hence belonging to the index-($\tau_{c} = \lfloor - \log_{2}(\DSvolume[c]) + 1\rfloor \in [\CGnumber]$) center group $G_{\tau_{c}}$.
Every red-black tree $\calG_{\tau}$ for $\tau \in [\CGnumber]$ has size $|\calG_{\tau}| \le n$ and, therefore, can support searching in time $O(\log(|\calG_{\tau}|)) = O(\log(n))$ \cite[Chapter~13]{CLRS22}. This finishes the proof of \Cref{lem:grouping}.
\end{proof}

All possible modifications to the grouping $\{(G_{\tau}, \calG_{\tau})\}_{\tau \in [\CGnumber]}$ (due to the operations in \Cref{subsec:initialize,subsec:insert,subsec:delete}) will always maintain Invariants~\ref{invar:grouping:group} and \ref{invar:grouping:tree}.
In this regard, the following \Cref{lem:maintain-grouping} shows that this grouping $\{(G_{\tau}, \calG_{\tau})\}_{\tau \in [\CGnumber]}$ can be modified efficiently.

\begin{lemma}[Synchronization]
\label{lem:maintain-grouping}
\begin{flushleft}
For the grouping $\{(G_{\tau}, \calG_{\tau})\}_{\tau \in [\CGnumber]}$:
\begin{enumerate}[font = {\em\bfseries}]
    \item \label{lem:maintain-grouping:initialize}
    It can be built on the deletion estimator $(\DSloss[c], \DSvolume[c])_{c \in C}$ in time $O(n \log(n))$.
    
    \item \label{lem:maintain-grouping:modify}
    It can be synchronized to maintain Invariants~\ref{invar:grouping:group} and \ref{invar:grouping:tree} in time $O(\log(n))$, every time when the subclusterings $(\DSc_{J}[v], \DSd_{J}[v])_{(J,\, v) \in \calJ \times V}$ are modified at a single entry $(J, v) \in \calJ \times V$.
\end{enumerate}
\end{flushleft}
\end{lemma}

\begin{proof}
Both \Cref{lem:maintain-grouping:initialize,lem:maintain-grouping:modify} are rather obvious.

\vspace{.1in}
\noindent
{\bf \Cref{lem:maintain-grouping:initialize}.}
Starting with the {\em empty} grouping $\{(G_{\tau}, \calG_{\tau})\}_{\tau \in [\CGnumber]} \gets (\emptyset,\ \emptyset)^{|\CGnumber|}$, we enumerate every center $c \in C$, identify which center group $\tau_{c} = \lfloor -\log_{2}(\DSvolume[c]) + 1 \rfloor \in [\CGnumber]$ it belongs to, and insert this center $c$ into the identified center group $G_{\tau_{c}}$ and the pair $(c, \DSloss[c])$ into the identified red-black tree $\calG_{\tau_{c}}$.
After the enumeration terminates, we get Invariants~\ref{invar:grouping:group} and \ref{invar:grouping:tree} maintained by construction. Further, the total running time for this initialization is dominated by the $|C| \le n$ insertions of $(c, \DSloss[c])$ into the size-($\le |C|$) red-black trees $\calG_{\tau}$, thus the total running time $O(|C| \log(|C|)) = O(n \log(n))$.

\vspace{.1in}
\noindent
{\bf \Cref{lem:maintain-grouping:modify}.}
Regarding every modification of the subclustering, from $(\DSc_{J}[v], \DSd_{J}[v])$ to, say, $(\DSc'_{J}[v], \DSd'_{J}[v])$, at a single entry $(J, v) \in \calJ \times V$, we would synchronize the grouping $\{(G_{\tau}, \calG_{\tau})\}_{\tau \in [\CGnumber]}$ as follows:
\begin{flushleft}
\begin{itemize}
    \item The clustering can only be modified at {\em one} vertex $v$, from $(\DSc[v], \DSd[v])$ to, say, $(\DSc'[v], \DSd'[v])$.\\
    The deletion loss estimator can only be modified at {\em two} centers $\DSc[v]$ and $\DSc'[v]$ (possibly identical), from $\DSloss[\DSc[v]]$ and $\DSloss[\DSc'[v]]$ to, say, $\DSloss'[\DSc[v]]$ and $\DSloss'[\DSc'[v]]$.\\
    The deletion volume estimator can only be modified at {\em two} centers $\DSc_{J}[v]$ and $\DSc'_{J}[v]$ (possibly identical), from $\DSvolume[\DSc_{J}[v]]$ and $\DSvolume[\DSc'_{J}[v]]$ to, say, $\DSvolume'[\DSc_{J}[v]]$ and $\DSvolume'[\DSc'_{J}[v]]$.\\
    \Comment{See the proof of \Cref{lem:maintain-deletion-estimator} for more details. This synchronization takes time $O(\log(n))$, by \Cref{lem:maintain-clustering:modify} of \Cref{lem:maintain-clustering} and \Cref{lem:maintain-deletion-estimator:modify} of \Cref{lem:maintain-deletion-estimator}.}
    
    \item Thus, the deletion estimator can only be modified at {\em four} centers $c \in \{\DSc[v], \DSc'[v], \DSc_{J}[v], \DSc'_{J}[v]\}$ (possibly a multiset), center-wise from $(\DSloss[c], \DSvolume[c])$ to, say, $(\DSloss'[c], \DSvolume'[c])$. Then for every such center $c \in \{\DSc[v], \DSc'[v], \DSc_{J}[v], \DSc'_{J}[v]\}$:
    \begin{itemize}
        \item We identify which {\em original} center group $\tau_{c} = \lfloor -\log_{2}(\DSvolume[c]) + 1 \rfloor \in [\CGnumber]$ it belongs to,\\
        delete it from this {\em original} center group $G_{\tau_{c}} \gets G_{\tau_{c}} - c$, and\\
        delete the {\em original} pair $(c, \DSloss[c])$ from the associated {\em original} red-black tree $\calG_{\tau_{c}}$.\\
        \Comment{This requires {\em one} deletion from a size-($\le |C|$) red-black tree $\calG_{\tau_{c}}$, thus taking time $O(1) + O(1) + O(\log(|C|)) = O(\log(n))$ \cite[Chapter~13]{CLRS22}.}
        
        \item We identify which {\em new} center group $\tau'_{c} = \lfloor -\log_{2}(\DSvolume'[c]) + 1 \rfloor \in [\CGnumber]$ it belongs to,\\
        insert it into this {\em new} center group $G_{\tau'_{c}} \gets G_{\tau'_{c}} + c$, and\\
        insert the {\em new} pair $(c, \DSloss'[c])$ into the associated {\em original} red-black tree $\calG_{\tau_{c}}$.\\
        \Comment{This requires {\em one} insertion into a size-($\le |C|$) red-black tree $\calG_{\tau'_{c}}$, thus taking time $O(1) + O(1) + O(\log(|C|)) = O(\log(n))$ \cite[Chapter~13]{CLRS22}.}
    \end{itemize}
\end{itemize}
\end{flushleft}
Afterward, both Invariants~\ref{invar:grouping:group} and \ref{invar:grouping:tree} are maintained by construction. Moreover, the total running time is $O(\log(n)) + 4 \cdot (O(\log(n)) + O(\log(n))) = O(\log(n))$.

This finishes the proof of \Cref{lem:maintain-grouping}.
\end{proof}

\subsection*{The potential $\Phi$}

Finally, as mentioned, the potential $\Phi = \sum_{(J,\, v) \in \calJ \times V} \deg(v) \cdot \log_{2}(1 + \DSd_{J}[v])$ defined in \Cref{fig:DS} will helps with our analysis (although our data structure $\calD$ need not maintain it).

\begin{lemma}[Potential]
\label{lem:potential}
\begin{flushleft}
$0 \leq \Phi \le \Phi_{\max}$, for the parameter $\Phi_{\max} \eqdef 2m |\calJ| \cdot \log_{2}(1 + \dmax)$, provided \Cref{assumption:edge-weight,assumption:DS:hop-bounded}.

\Comment{$\Phi_{\max} = O(z m \log^{2}(n)) \impliedby \dmax \le n^{O(z)},\ |\calJ| = O(\log(n))$ (\Cref{prop:distance,def:cover}).}
\end{flushleft}
\end{lemma}

\begin{proof}
Obvious; $\sum_{v \in V} \deg(v) = 2m$ and $\DSd_{J}[v] \leq \dmax$, for every entry $(J, v) \in \calJ \times V$ (\Cref{lem:subclusterings}).
\end{proof}

\subsection{The operation {\Initialize}}
\label{subsec:initialize}

This subsection shows the operations {\Initialize} and {\InitializeSubClusterings} -- see \Cref{fig:initialize} for their implementation -- which initialize our data structure $\calD$, based on an initial feasible solution $\Cinitial \in V^{k}$ (promised).
Essentially, we will utilize Dijkstra's algorithm \cite[Chapter~22.3]{CLRS22}.

\begin{figure}[t]
\centering
\begin{mdframed}
\begin{flushleft}
Operation $\term[\Initialize]{alg:initialize}(\Cinitial)$

\vspace{.1in}
{\bf Input:}
An initial feasible solution $\Cinitial \in V^{k}$ (promised).
\begin{enumerate}
    \item \label{alg:initialize:center}
    $C \gets \Cinitial$ and $C_{J} \gets C \cap J$, for every index $J \in \calJ$.
    
    \item \label{alg:initialize:subclusterings}
    $\InitializeSubClusterings()$.
    
    \item \label{alg:initialize:rest}
    Initialize the rest of our data structure $\calD$, based on the subclusterings from Line~\ref{alg:initialize:subclusterings}, using \Cref{lem:maintain-clustering,lem:maintain-DScost,lem:maintain-deletion-estimator,lem:maintain-grouping}.
\end{enumerate}
\end{flushleft}
\end{mdframed}

\begin{mdframed}
\begin{flushleft}
Suboperation $\term[\InitializeSubClusterings]{alg:initialize-subclusterings}()$
\begin{enumerate}
\setcounter{enumi}{3}
    \item \label{alg:initialize-subclusterings:for-nonempty}
    For every nonempty maintained center subset $C_{J} \ne \emptyset$:
    
    \item \label{alg:initialize-subclusterings:modify-nonempty}
    \qquad $(\DSc_{J}[v], \DSd_{J}[v])_{v \in V} \gets (\argmin_{c \in C_{J}} \dist(v, c),\ \dist(v, C_{J}))_{v \in V}$.
    \hfill
    \Comment{Dijkstra's algorithm.}
    
    \item \label{alg:initialize-subclusterings:for-empty}
    For every empty maintained center subset $C_{J} = \emptyset$:
    
    \item \label{alg:initialize-subclusterings:modify-empty}
    \qquad $(\DSc_{J}[v], \DSd_{J}[v])_{v \in V} \gets (\perp, \dmax)^{|V|}$.
\end{enumerate}
\end{flushleft}
\end{mdframed}
\caption{\label{fig:initialize}The operations {\Initialize} and {\InitializeSubClusterings}.}
\end{figure}

The following \Cref{lem:initialize} shows the performance guarantees of the operation {\Initialize}.

\begin{lemma}[{\Initialize}]
\label{lem:initialize}
\begin{flushleft}
Given as input an (initial) feasible solution $\Cinitial \in V^{k}$ (promised), after the operation $\Initialize(\Cinitial)$:
\begin{enumerate}[font = {\em\bfseries}]
    \item \label{lem:initialize:center}
    $C = \Cinitial$ and $C_{J} = C \cap J$, for every index $J \in \calJ$.
    
    \item \label{lem:initialize:DScost}
    The objective estimator $\DScost = \cost(V, \Cinitial)$.
    
    \item \label{lem:initialize:invar}
    The initialized data structure $\calD$ maintains Invariants~\ref{invar:subclusterings:center} to \ref{invar:subclusterings:empty} and thus Invariants~\ref{invar:clustering} to \ref{invar:DSvolume}.
    
    \item \label{lem:initialize:runtime}
    The worst-case running time $\Tinitialize = O(m \log(n) + n \log^{2}(n))$.
\end{enumerate}
\end{flushleft}
\end{lemma}

\begin{proof}
Let us go through the operation {\Initialize} step by step.

\noindent
First, we initialize the center set $C \gets \Cinitial$ and the center subsets $C_{J} \gets C \cap J$ for $J \in \calJ$, based on the isolation set cover $\calJ$ (\Cref{def:cover,lem:cover}).
\hfill
(Line~\ref{alg:initialize:center})

\noindent
\Comment{This ensures $C = \Cinitial$ by construction.}

\noindent
\Comment{The running time $= O(n \log(n)) + (1 + |\calJ|) \cdot O(n) = O(n \log(n))$.}

\noindent
Afterward, we invoke the suboperation {\InitializeSubClusterings} (Lines~\ref{alg:initialize-subclusterings:for-nonempty} to \ref{alg:initialize-subclusterings:modify-empty}) to  initialize the subclusterings $(\DSc_{J}[v], \DSd_{J}[v])_{(J,\, v) \in \calJ \times V}$.
\hfill
(Line~\ref{alg:initialize:subclusterings})

\noindent
Specifically, every index-$(J \in \calJ)$ subclustering falls into either {\bf Case~1} or {\bf Case~2}:

\vspace{.1in}
\noindent
{\bf Case~1:}
The index-$J$ center subset is nonempty $C_{J} \ne \emptyset$.
\hfill
(Line~\ref{alg:initialize-subclusterings:for-nonempty})

\noindent
In this case, we assign $(\DSc_{J}[v], \DSd_{J}[v]) \gets (\argmin_{c \in C_{J}} \dist(v, c),\ \dist(v, C_{J}))$ the {\em (exact) nearest center} and the {\em (exact) distance}, for every vertex $v \in V$.
\hfill
(Line~\ref{alg:initialize-subclusterings:modify-nonempty})

\noindent
\Comment{This maintains Invariants~\ref{invar:subclusterings:center} to \ref{invar:subclusterings:vertex} by construction (since $(\DSc_{J}[v], \DSd_{J}[v])_{v \in V}$ is the ``groundtruth'').}

\noindent
\Comment{The running time $= O(m + n \log(n))$, using Dijkstra's algorithm \cite[Chapter~22.3]{CLRS22}.}

\vspace{.1in}
\noindent
{\bf Case~2:}
The index-$J$ center subset is empty $C_{J} = \emptyset$.
\hfill
(Line~\ref{alg:initialize-subclusterings:for-empty})

\noindent
In this case, we simply assign $(\DSc_{J}[v], \DSd_{J}[v]) \gets (\perp, \dmax)$, for every vertex $v \in V$.
\hfill
(Line~\ref{alg:initialize-subclusterings:modify-empty})

\noindent
\Comment{This maintains Invariant~\ref{invar:subclusterings:empty} by construction.}

\noindent
\Comment{The running time $= O(n)$.}

\vspace{.1in}
\noindent
Eventually, we initialize the rest of our data structure $\calD$, based on the subclusterings from Line~\ref{alg:initialize:subclusterings}, using the respective ``initialization'' parts of \Cref{lem:maintain-clustering,lem:maintain-DScost,lem:maintain-deletion-estimator,lem:maintain-grouping}.
\hfill
(Line~\ref{alg:initialize:rest})

\noindent
\Comment{This maintains Invariant~\ref{invar:clustering} to \ref{invar:DSvolume} by construction (\Cref{lem:maintain-clustering,lem:maintain-DScost,lem:maintain-deletion-estimator,lem:maintain-grouping}).}

\noindent
\Comment{The running time $= O(n \log(n)) + O(n) + O(n \log(n)) = O(n \log(n))$.}

\noindent
In sum, the operation {\Initialize} satisfies both \Cref{lem:initialize:center,lem:initialize:DScost,lem:initialize:invar}; for \Cref{lem:initialize:DScost} in particularly, we have
\begin{align*}
    \DScost
    & ~=~ \sum_{v \in V} \min_{J \in \calJ} \DSd_{J}^{z}[v] \\
    & ~=~ \sum_{v \in V} \min_{J \in \calJ} \Big(\dist^{z}(v, C_{J}) \cdot \mathbb{I}(C_{J} \ne \emptyset) + \dmax^{z} \cdot \mathbb{I}(C_{J} = \emptyset)\Big) \\
    & ~=~ \sum_{v \in V} \dist^{z}(v, C) \\
    & ~=~ \cost(V, C).
\end{align*}
Here, the first step applies Invariants~\ref{invar:clustering} and \ref{invar:DScost}.
The second step applies the constructions in Lines~\ref{alg:initialize-subclusterings:modify-nonempty} and \ref{alg:initialize-subclusterings:modify-empty}.
And the third step applies $\cup_{J \in \calJ} C_{J} = C$ (\Cref{footnote:center-subset}) and \Cref{prop:distance}.
Further, (\Cref{lem:initialize:runtime}) the total running time
\begin{align*}
    \Tinitialize
    ~=~ \underbrace{O(n \log(n))}_{\text{Line~\ref{alg:initialize:center}}}
    + \underbrace{|\calJ| \cdot O(m + n \log(n))}_{\text{Line~\ref{alg:initialize:subclusterings}}}
    + \underbrace{O(n \log(n))}_{\text{Line~\ref{alg:initialize:rest}}}
    ~=~ O(m \log(n) + n \log^{2}(n)).
\end{align*}
This finishes the proof of \Cref{lem:initialize}.
\end{proof}

\subsection{The operation {\Insert}}
\label{subsec:insert}

This subsection presents the operations {\Insert} and {\InsertSubclustering} -- see \Cref{fig:insert} for their implementation -- which address the insertion of a {\em current noncenter} $\cinsert \notin C$ (promised) into the maintained center set $C$.
Without ambiguity, throughout \Cref{subsec:insert} we denote by $\calD$ and $\calD'$ the data structures before and after the operation $\Insert(\cinsert)$, respectively; likewise for their contents, the maintained center sets $C$ versus $C'$, the maintained center subsets $C_{J}$ versus $C'_{J}$, etc.
Essentially, we will utilize an adaptation of depth-first search \cite[Chapter~20.3]{CLRS22}.

\begin{figure}[t]
\centering
\begin{mdframed}
\begin{flushleft}
Operation $\term[\Insert]{alg:insert}(\cinsert)$

\vspace{.1in}
{\bf Input:}
A current noncenter $\cinsert \notin C$ (promised).
\begin{enumerate}
    \item \label{alg:insert:for}
    For every index $(J \in \calJ \colon J \ni \cinsert)$:
    
    \item \label{alg:insert:subclustering}
    \qquad $\InsertSubclustering(\cinsert, 0, J, \cinsert)$.
    \hfill
    \Comment{Adapted from depth-first search.}
    
    \item \label{alg:insert:center-subset}
    \qquad $C_{J} \gets C_{J} + \cinsert$.
    
    \item \label{alg:insert:center-set}
    $C \gets C + \cinsert$.
\end{enumerate}
\end{flushleft}
\end{mdframed}

\begin{mdframed}
\begin{flushleft}
Suboperation $\term[\InsertSubclustering]{alg:insert-subclustering}(\cinsert, d, J, v)$

\vspace{.1in}
{\bf Input:}
A pair $(\cinsert, d) \in (V \setminus C) \times [0,\ \dmax]$ and an entry $(J, v) \in \calJ \times V \colon J \ni \cinsert$ (promised).

\Comment{Synchronize the rest of our data structure $\calD$ accordingly (\Cref{lem:maintain-clustering,lem:maintain-DScost,lem:maintain-deletion-estimator,lem:maintain-grouping}), every time when the index-$(J \in \calJ)$ subclustering $(\DSc_{J}[v], \DSd_{J}[v])_{v \in V}$ is modified in Line~\ref{alg:insert-subclustering:modify}.}
\begin{enumerate}
\setcounter{enumi}{4}
    \item \label{alg:insert-subclustering:modify}
    $(\DSc_{J}[v], \DSd_{J}[v]) \gets (\cinsert, d)$.
    
    \item \label{alg:insert-subclustering:for}
    For every neighbor $\big(u \in N(v) \colon \DSd_{J}[u] > 2^{\eps / \beta} \cdot (\DSd_{J}[v] + w(u, v))\big)$:
    
    \item \label{alg:insert-subclustering:recursion}
    \qquad $\InsertSubclustering(\cinsert, \DSd_{J}[v] + w(u, v), J, u)$.
\end{enumerate}
\end{flushleft}
\end{mdframed}
\caption{\label{fig:insert}The (sub)operations {\Insert} and {\InsertSubclustering}.}
\end{figure}

The following \Cref{lem:insert} shows the performance guarantees of the operation {\Insert}.

\begin{lemma}[{\Insert}]
\label{lem:insert}
\begin{flushleft}
Invoke the operation $\Insert(\cinsert)$, on input a current noncenter $\cinsert \notin C$ (promised), for a data structure $\calD$ that maintains Invariants~\ref{invar:subclusterings:center} to \ref{invar:subclusterings:empty} and Invariants~\ref{invar:clustering} to \ref{invar:DSvolume}:
\begin{enumerate}[font = {\em\bfseries}]
    \item \label{lem:insert:center}
    $C' = C + \cinsert$ and $C'_{J} = C' \cap J$, for every index $J \in \calJ$.
    
    \item \label{lem:insert:DSd}
    $\DSd'[v] \le \DSd[v]$, for every vertex $v \in V$.
    
    \item \label{lem:insert:DSc}
    $(\DSc'[v], \DSd'[v]) = (\DSc[v], \DSd[v])$, for every vertex $(v \in V \colon \DSc'[v] \ne \cinsert)$.
    
    \item \label{lem:insert:DSvolume}
    $\DSvolume'[c] \le \DSvolume[c]$, for every center $c \in C = C' - \cinsert$.
    
    \item \label{lem:insert:invar}
    The modified data structure $\calD'$ maintains Invariants~\ref{invar:subclusterings:center} to \ref{invar:subclusterings:empty} and thus Invariants~\ref{invar:clustering} to \ref{invar:DSvolume}.
    
    \item \label{lem:insert:runtime}
    The worst-case running time $\Tinsert = (\Phi - \Phi') \cdot O(\eps^{-1} \beta \log(n)) \le \Phi_{\max} \cdot O(\eps^{-1} \beta \log(n))$.
\end{enumerate}
\end{flushleft}
\end{lemma}

\begin{proof}
The operation {\Insert} iterates the following for every index $(J \in \calJ \colon J \ni \cinsert)$:
\hfill
(Line~\ref{alg:insert:for})

\noindent
We invoke the suboperation {\InsertSubclustering} (Lines~\ref{alg:insert-subclustering:modify} to \ref{alg:insert-subclustering:recursion}), which is an adaptation of depth-first search \cite[Chapter~20.3]{CLRS22}, to modify the index-$J$ subclustering $(\DSc_{J}[v], \DSd_{J}[v])_{v \in V}$.
\hfill
(Line~\ref{alg:insert:subclustering})

\vspace{.1in}
\noindent
Specifically, starting with the input noncenter $\cinsert \in J \setminus C_{J}$ itself,\footnote{We have $\cinsert \in J \setminus C_{J} = J \setminus C \impliedby (\cinsert \notin C) \wedge (\cinsert \in J)$.} we move this $\cinsert$ to the {\em inserted} index-$J$ center-$\cinsert$ subcluster, namely $(\DSc'_{J}[\cinsert], \DSd'_{J}[\cinsert]) \gets (\cinsert, 0)$.
\hfill
(Line~\ref{alg:insert-subclustering:modify})

\noindent
Afterward, we test Invariant~\ref{invar:subclusterings:edge} for every neighbor $u \in N(\cinsert)$:
\hfill
(Line~\ref{alg:insert-subclustering:for})

\noindent
A violating neighbor $\big(u \in N(\cinsert) \colon \DSd_{J}[u] > 2^{\eps / \beta} \cdot (\DSd'_{J}[\cinsert] + w(u, \cinsert)) = 2^{\eps / \beta} \cdot w(u, \cinsert)\big)$ shall also move to the {\em inserted} index-$J$ center-$\cinsert$ cluster, i.e., its violation to Invariant~\ref{invar:subclusterings:edge} (roughly speaking) shows that the input noncenter $\cinsert \in J \setminus C_{J}$ is much nearer than its current center $\DSc_{J}[u] \in C_{J}$.
We thus recursively invoke the suboperation {\InsertSubclustering} to further modify the entry $(J, u)$, namely $(\DSc'_{J}[u], \DSd'_{J}[u]) \gets (\cinsert, \DSd'_{J}[\cinsert] + w(u, \cinsert)) = (\cinsert, w(u, \cinsert))$.
\hfill
(Line~\ref{alg:insert-subclustering:recursion})

\noindent
\Comment{Throughout the recursion of {\InsertSubclustering}, every time when Line~\ref{alg:insert-subclustering:modify} modifies the index-$J$ subclustering $(\DSc_{J}[v], \DSd_{J}[v])_{v \in V}$, in time $O(1)$, we would synchronize the rest of our data structure $\calD$ (\Cref{lem:maintain-clustering,lem:maintain-DScost,lem:maintain-deletion-estimator,lem:maintain-grouping}), in time $O(\log(n)) + O(\log(n)) + O(\log(n)) = O(\log(n))$.}

\vspace{.1in}
\noindent
After the above recursion of {\InsertSubclustering} terminates -- which we assume for the moment but will prove later in \Cref{lem:insert:runtime} -- we complete our modifications for this index $(J \in \calJ \colon J \ni \cinsert)$ by maintaining the index-$J$ center subset $C'_{J} \gets C_{J} + \cinsert$.
\hfill
(Line~\ref{alg:insert:center-subset})

\noindent
After the iteration of all indices $(J \in \calJ \colon J \ni \cinsert)$ terminates, we complete our modifications to the whole data structure $\calD$ by maintaining the center set $C' \gets C + \cinsert$.
\hfill
(Line~\ref{alg:insert:center-set})

\vspace{.1in}
As the whole process of the operation {\Insert} now are clear, we are ready to present the proof of \Cref{lem:insert}.
For ease of notation, we denote by $\ell \ge 1$ the total number of invocations (due to either Line~\ref{alg:insert:subclustering} or \ref{alg:insert-subclustering:recursion}) of the suboperation {\InsertSubclustering} and by $(\cinsert, d_{i}, J_{i}, v_{i})$ the input to every invocation $i \in [\ell]$; we observe that the first argument $\cinsert$ is always the same.

We begin with the following \Cref{claim:insert-subclustering}, which will be useful in several places. (Recall \Cref{prop:distance} for the parameters $\dmin = 1$ and $\dmax \le n^{O(z)}$.)

\begin{claim}
\label{claim:insert-subclustering}
\begin{flushleft}
For the input $(\cinsert, d_{i}, J_{i}, v_{i})$ to every invocation $i \in [\ell]$ of {\InsertSubclustering}:
\begin{enumerate}[font = {\em\bfseries}]
    \item \label{claim:insert-subclustering:di}
    $d_{i} \in \{0\} \cup [\dmin, +\infty)$.
    
    \item \label{claim:insert-subclustering:dist}
    $d_{i}\ge \dist(v_{i}, \cinsert)$.
    
    \item \label{claim:insert-subclustering:dJ}
    $\DSd_{J_{i}}[v_{i}] \ge \max(\dmin,\ 2^{\eps / \beta} \cdot d_{i})$.
    \hfill
    \Comment{Here $\DSd_{J_{i}}[v_{i}]$ denotes the one just before this invocation $i \in [\ell]$.}
\end{enumerate}
\end{flushleft}
\end{claim}

\begin{proof}
We prove \Cref{claim:insert-subclustering} by induction on the order of all invocations $i \in [\ell]$.

\vspace{.1in}
\noindent
{\bf Base Case: an invocation $i \in [\ell]$ by Line~\ref{alg:insert:subclustering}.}
Such an invocation $i \in [\ell]$ is the first invocation for some index $(J \in \calJ \colon J \ni \cinsert)$ and has input $(\cinsert, d_{i}, J_{i}, v_{i}) = (\cinsert, 0, J, \cinsert)$; both \Cref{claim:insert-subclustering:di,claim:insert-subclustering:dist} are trivial $d_{i} = 0 = \dist(\cinsert, \cinsert)$.
Also, we can conclude with \Cref{claim:insert-subclustering:dJ}, as follows:
\begin{align*}
    \DSd_{J_{i}}[v_{i}]
    ~=~ \DSd_{J}[\cinsert]
    ~\ge~ \dist(\cinsert, C_{J})
    ~\ge~ \dmin
    ~=~ \max(\dmin,\ 0)
    ~=~ \max(\dmin,\ 2^{\eps / \beta} \cdot d_{i}).
\end{align*}
Here, the second step applies Invariant~\ref{invar:subclusterings:vertex} for the original data structure $\calD$ (i.e., this invocation $i \in [\ell]$ is the first invocation for the index $J$), and the third step applies \Cref{prop:distance} to the input {\em noncenter} $\cinsert \notin C_{J} = C \cap J \impliedby \cinsert \notin C$.

\vspace{.1in}
\noindent
{\bf Induction Step: an invocation $i \in [\ell]$ by Line~\ref{alg:insert-subclustering:recursion}.}
Such a {\em recursive} invocation $i \in [\ell]$ is invoked by an earlier invocation $i' \in [i - 1]$, which considers an adjacent vertex $(v_{i'} \in V \colon (v_{i}, v_{i'}) \in E)$ (Line~\ref{alg:insert-subclustering:for}) for the same index $(J_{i'} = J_{i} = J \colon J \ni \cinsert)$ (Line~\ref{alg:insert:for}) and makes the modification $(\DSc_{J}[v_{i'}], \DSd_{J}[v_{i'}]) \gets (\cinsert, d_{i'})$.
Without loss of generality ({\bf induction hypothesis}), we have $\DSd_{J}[v_{i'}] = d_{i'} \ge \dist(v_{i'}, \cinsert)$.
And to truly invoke the {\em recursive} invocation $i \in [\ell]$, we must have $\DSd_{J}[v_{i}] > 2^{\eps / \beta} \cdot (\DSd_{J}[v_{i'}] + w(v_{i}, v_{i'}))$ (Line~\ref{alg:insert-subclustering:for}) and $d_{i} = \DSd_{J}[v_{i'}] + w(v_{i}, v_{i'})$ (Line~\ref{alg:insert-subclustering:recursion}).
Given these, we can deduce \Cref{claim:insert-subclustering:di,claim:insert-subclustering:dJ,claim:insert-subclustering:dist} as follows.

\noindent
\Cref{claim:insert-subclustering:di}:
$d_{i} = \DSd_{J}[v_{i'}] + w(v_{i}, v_{i'}) \ge w(v_{i}, v_{i'}) \ge \wmin = \dmin = 1$.
\hfill
\Comment{Cf.\ \Cref{prop:distance}.}

\noindent
\Cref{claim:insert-subclustering:dist}:
$d_{i} \ge \dist(v_{i'}, \cinsert) + w(v_{i'}, v_{i}) \ge \dist(v_{i}, \cinsert)$.
\hfill
\Comment{Triangle inequality.}

\noindent
\Cref{claim:insert-subclustering:dJ}:
$\DSd_{J}[v_{i}] > 2^{\eps / \beta} \cdot (\DSd_{J}[v_{i'}] + w(v_{i}, v_{i'})) = 2^{\eps / \beta} \cdot d_{i} = \max(\dmin,\ 2^{\eps / \beta} \cdot d_{i})$.
\hfill
\Comment{$d_{i} \ge \dmin = 1$ (\Cref{claim:insert-subclustering:di})}

This finishes the proof of \Cref{claim:insert-subclustering}.
\end{proof}

Now we move back to the proof of \Cref{lem:insert}.

\noindent
{\bf \Cref{lem:insert:center}.}
Lines~\ref{alg:insert:center-subset} and \ref{alg:insert:center-set} trivially imply that $C' = C + \cinsert$ and $C'_{J} = C' \cap J$, for every index $J \in \calJ$.

\vspace{.1in}
\noindent
{\bf \Cref{lem:insert:DSd}.}
By \Cref{claim:insert-subclustering:dJ} of \Cref{claim:insert-subclustering}, we have $\DSd'_{J}[v] \le \DSd_{J}[v]$, for every entry $(J, v) \in \calJ \times V$, which implies that $\DSd'[v] = \min_{J \in \calJ} \DSd'_{J}[v] \le \min_{J \in \calJ} \DSd_{J}[v] = \DSd[v]$, for every vertex $v \in V$.
(More rigorously, the step $\DSd'[v] = \min_{J \in \calJ} \DSd'_{J}[v]$ assumes Invariant~\ref{invar:clustering} for the modified data structure $\calD$, which we will prove later in \Cref{lem:insert:invar}.)

\vspace{.1in}
\noindent
{\bf \Cref{lem:insert:DSc,lem:insert:DSvolume}.}
Only Line~\ref{alg:insert-subclustering:modify} can modify ``$(\DSc'_{J}[v], \DSd'_{J}[v]) \gets (\cinsert, d)$'' an entry $(J, v) \in \calJ \times V$ of the subclusterings, after which the considered vertex $v$ must locate in the inserted index-$J$ center-$\cinsert$ subcluster. This observation implies that:
\begin{align*}
    (\DSc'_{J}[v], \DSd'_{J}[v])
    & ~=~ (\DSc_{J}[v], \DSd_{J}[v]),
    && \forall (J, v) \in \calJ \times V \colon \DSc'_{J}[v] \ne \cinsert. \\
    \{(J,\, v) \in \calJ \times V \mid \DSc'_{J}[v] = c\}
    & ~\subseteq~ \{(J,\, v) \in \calJ \times V \mid \DSc_{J}[v] = c\},
    && \forall c \in C = C' - \cinsert.
\end{align*}
\Cref{lem:insert:DSc} follows directly from the first equation above, and
\Cref{lem:insert:DSvolume} follows directly from the second equation above:
\begin{align*}
    \DSvolume'[c]
    ~=~ \sum_{(J,\, v) \in \calJ \times V \colon \DSc'_{J}[v] = c} \tfrac{\deg(v)}{2m |\calJ|}
    ~\le~ \sum_{(J,\, v) \in \calJ \times V \colon \DSc_{J}[v] = c} \tfrac{\deg(v)}{2m |\calJ|}
    ~=~ \DSvolume[c].
\end{align*}
Here, the first step holds if the modified data structure $\calD'$ maintains Invariant~\ref{invar:DSvolume}, which we assume for the moment but will prove later in \Cref{lem:insert:invar}.

\vspace{.1in}
\noindent
{\bf \Cref{lem:insert:invar}.}
We would prove that the modified subclusterings $(\DSc'_{J}[v], \DSd'_{J}[v])_{(J,\, v) \in \calJ \times V}$ maintain Invariants~\ref{invar:subclusterings:center} to \ref{invar:subclusterings:empty}; suppose so, the rest of the modified data structure $\calD'$ maintains Invariants~\ref{invar:clustering} to \ref{invar:DSvolume} by construction, given the respective ``synchronization'' parts of \Cref{lem:maintain-clustering,lem:maintain-DScost,lem:maintain-deletion-estimator,lem:maintain-grouping}.
Invariant~\ref{invar:subclusterings:empty} is rather trivial.\footnote{Invariant~\ref{invar:subclusterings:empty} asserts that: If $C'_{J} = \emptyset$, then $(\DSc'_{J}[v], \DSd'_{J}[v])_{v \in V} = (\perp, \dmax)^{n}$. We observe that $C'_{J} = \emptyset$ means $J \notni \cinsert$, so this index-$J$ subclustering is unmodified $(\DSc'_{J}[v], \DSd'_{J}[v])_{v \in V} = (\DSc_{J}[v], \DSd_{J}[v])_{v \in V} = (\perp, \dmax)^{n}$.}
Below, we would establish Invariants~\ref{invar:subclusterings:center} to \ref{invar:subclusterings:vertex} for a specific index $(J \in \calJ \colon J \ni \cinsert)$; notice that $C'_{J} = C_{J} + \cinsert \ne \emptyset$.

\vspace{.1in}
\noindent
{\bf Invariant~\ref{invar:subclusterings:center}:}
$(\DSc'_{J}[c], \DSd'_{J}[c]) = (c, 0)$, for every center $c \in C'_{J} = C_{J} + \cinsert$.

\noindent
The first invocation (due to Line~\ref{alg:insert:subclustering}) of the suboperation {\InsertSubclustering} modifies (Line~\ref{alg:insert-subclustering:modify}) the entry $(J, \cinsert)$ from $(\DSc_{J}[\cinsert], \DSd_{J}[\cinsert])$ to $(\DSc'_{J}[\cinsert], \DSd'_{J}[\cinsert]) = (\cinsert, 0)$.
At this moment, Invariant~\ref{invar:subclusterings:center} holds for both this inserted center $\cinsert$ and every original center $(\DSc'_{J}[c], \DSd'_{J}[c]) = (\DSc_{J}[c], \DSd_{J}[c]) = (c, 0)$, $\forall c \in C_{J}$.
But thereafter, an inserted/original center $c \in C'_{J} = C_{J} + \cinsert$ can never pass the test in Line~\ref{alg:insert-subclustering:for}, i.e., $\DSd'_{J}[c] = 0 \ngtr 2^{\eps / \beta} \cdot (\DSd'_{J}[v] + w(c, v))$. Hence, the subsequent recursions (due to Line~\ref{alg:insert-subclustering:recursion}) of the suboperation {\InsertSubclustering} cannot modify those pairs $(\DSc'_{J}[c], \DSd'_{J}[c]) = (c, 0)$, $\forall c \in C'_{J}$.

\vspace{.1in}
\noindent
{\bf Invariant~\ref{invar:subclusterings:edge}:}
$\DSd'_{J}[u] \le 2^{\eps / \beta} \cdot (\DSd'_{J}[v] + w(u, v))$, for every edge $(u, v) \in E$.

\noindent
(i)~If the operation {\Insert} does not modify $(\DSc'_{J}[v], \DSd'_{J}[v]) = (\DSc_{J}[v], \DSd_{J}[v])$ the entry $(J, v)$, we have
\begin{align*}
    \DSd'_{J}[u]
~\le~ \DSd_{J}[u]
    ~\le~ 2^{\eps / \beta} \cdot (\DSd_{J}[v] + w(u, v))
    ~=~ 2^{\eps / \beta} \cdot (\DSd'_{J}[v] + w(u, v)).
\end{align*}
Here, the first step applies \Cref{lem:insert:DSd}, and the second step holds (the premise of \Cref{lem:insert}) since the original data structure $\calD$ maintains Invariant~\ref{invar:subclusterings:edge}.

\noindent
(ii)~Otherwise, consider the last modification ``$(\DSc'_{J}[v], \DSd'_{J}[v]) \gets (\cinsert, d_{i})$'' to this entry $(J, v) = (J_{i}, v_{i})$, by some invocation $i \in [\ell]$ of the suboperation {\InsertSubclustering}.
As long as the considered edge $(u, v)$ violates Invariant~\ref{invar:subclusterings:edge} at this moment, i.e., $\DSd'_{J}[u] > 2^{\eps / \beta} \cdot (\DSd'_{J}[v] + w(u, v))$, Lines~\ref{alg:insert-subclustering:for} and \ref{alg:insert-subclustering:recursion} will detect this violation and recursively invoke the suboperation {\InsertSubclustering} to modify $(\DSc'_{J}[v], \DSd'_{J}[v])$, after which the considered edge $(u, v)$ will maintain Invariant~\ref{invar:subclusterings:edge}.

Combining both cases gives Invariant~\ref{invar:subclusterings:edge}. (Rigorously, here we assume that the operation {\Insert} will terminate -- this will be proved later in \Cref{lem:insert:runtime}.)

\vspace{.1in}
\noindent
{\bf Invariant~\ref{invar:subclusterings:vertex}:}
$\DSc'_{J}[v] \in C'_{J}$ and $\dist(v, C'_{J}) \le \dist(v, \DSc'_{J}[v]) \le \DSd'_{J}[v]$, for every vertex $v \in V$.

\noindent
(i)~If the operation {\Insert} does not modify $(\DSc'_{J}[v], \DSd'_{J}[v]) = (\DSc_{J}[v], \DSd_{J}[v])$ the considered entry $(J, v)$, it trivially maintains Invariant~\ref{invar:subclusterings:vertex}; notice that $C'_{J} = C_{J} + \cinsert \supseteq C_{J} \implies \dist(v, C'_{J}) \le \dist(v, C_{J})$.

\noindent
(ii)~Otherwise, consider the last modification ``$(\DSc'_{J}[v], \DSd'_{J}[v]) \gets (\cinsert, d_{i})$'' to this entry $(J, v) = (J_{i}, v_{i})$, by some invocation $i \in [\ell]$ of the suboperation {\InsertSubclustering}.
We thus have $\DSc'_{J}[v] = \cinsert \in C_{J} + \cinsert = C'_{J} \implies \dist(v, C'_{J}) \le \dist(v, \DSc'_{J}[v])$.
And it follows directly from \Cref{claim:insert-subclustering:dist} of \Cref{claim:insert-subclustering} that $\dist(v, \DSc'_{J}[v]) = \dist(v, \cinsert) \le d_{i} = \DSd'_{J}[v]$.

Combining both cases gives Invariant~\ref{invar:subclusterings:vertex}.

\vspace{.1in}
\noindent
{\bf \Cref{lem:insert:runtime}.}
The running time of the operation $\Insert(\cinsert)$ is dominated by the total running time of all invocations of the suboperation {\InsertSubclustering} (Lines~\ref{alg:insert:subclustering} and \ref{alg:insert-subclustering:recursion}).\footnote{In contrast, maintaining the center subsets $C_{J}$ (Line~\ref{alg:insert:center-subset}) and the center set $C$ (Line~\ref{alg:insert:center-set}) each takes time $O(1)$.}
Hence, we can infer \Cref{lem:insert:runtime} from a combination of two observations.

\vspace{.1in}
\noindent
(i)~Every invocation $i \in [\ell]$ takes time $O(\log(n)) + O(\deg(v_{i})) = O(\deg(v_{i}) \cdot \log(n))$. Specifically:

\noindent
First, we modify a single entry $(J_{i}, v_{i})$ of the subclusterings, in time $O(1)$, and synchronize the rest of our data structure $\calD$ using \Cref{lem:maintain-clustering,lem:maintain-DScost,lem:maintain-deletion-estimator,lem:maintain-grouping}, in time $O(\log(n))$.
\hfill
(Line~\ref{alg:insert-subclustering:modify})

\noindent
Then, we enumerate all neighbors of vertex $v_{i}$ in time $O(\deg(v_{i}))$.
\hfill
(Line~\ref{alg:insert-subclustering:for})

\noindent
The possible recursive invocations should not be counted to this invocation $i \in [\ell]$.
\hfill
(Line~\ref{alg:insert-subclustering:recursion})

\vspace{.1in}
\noindent
(ii)~Every invocation $i \in [\ell]$ changes the potential by $-\deg(v_{i}) \cdot \Omega(\eps / \beta)$.

\noindent
Recall that the potential formula $\Phi = \sum_{(J,\, v) \in \calJ \times V} \deg(v) \cdot \log_{2}(1 + \DSd_{J}[v])$; likewise for $\Phi'$ (\Cref{fig:DS}).
Regarding the modification in the considered invocation $i \in [\ell]$ (Line~\ref{alg:insert-subclustering:modify}), say from $(\DSc_{J_{i}}[v_{i}], \DSd_{J_{i}}[v_{i}])$ to $(\DSc'_{J_{i}}[v_{i}], \DSd'_{J_{i}}[v_{i}]) \gets (\cinsert_{i}, d_{i})$, we address either case $\{\DSc'_{J_{i}}[v_{i}] = d_{i} = 0\}$ or $\{\DSc'_{J_{i}}[v_{i}] = d_{i} > 0\}$ separately.

\vspace{.1in}
\noindent
{\bf Case~1: $\DSd'_{J_{i}}[v_{i}] = d_{i} = 0$.}
The potential change in this case is
\begin{align*}
    \deg(v_{i}) \cdot \log_{2}\Big(\tfrac{1 + \DSd'_{J_{i}}[v_{i}]}{1 + \DSd_{J_{i}}[v_{i}]}\Big)
    & ~\le~ -\deg(v_{i}) \cdot \log_{2}(1 + \dmin)
    \hspace{0.75cm}
    \tag{\Cref{claim:insert-subclustering:dJ} of \Cref{claim:insert-subclustering}} \\
    & ~\le~ -\deg(v_{i}) \cdot \Omega(\eps / \beta).
    \tag{$\dmin = 1$, $z \ge 1$, and $0 < \eps < 1$}
\end{align*}

\vspace{.1in}
\noindent
{\bf Case~2: $\DSd'_{J_{i}}[v_{i}] = d_{i} > 0$.}
The potential change in this case is
\begin{align*}
    \deg(v_{i}) \cdot \log_{2}\Big(\tfrac{1 + \DSd'_{J_{i}}[v_{i}]}{1 + \DSd_{J_{i}}[v_{i}]}\Big)
    & ~\le~ -\deg(v_{i}) \cdot \log_{2}\Big(\tfrac{1 + 2^{\eps / \beta} \cdot d_{i}}{1 + d_{i}}\Big)
    \tag{\Cref{claim:insert-subclustering:dJ} of \Cref{claim:insert-subclustering}} \\
    & ~\le~ -\deg(v_{i}) \cdot \log_{2}\Big(\tfrac{1 + 2^{\eps / \beta}}{2}\Big)
    \tag{\Cref{claim:insert-subclustering:di} of \Cref{claim:insert-subclustering}} \\
    & ~\le~ -\deg(v_{i}) \cdot \Omega(\eps / \beta).
\end{align*}
Combining both cases gives observation~(ii) and thus \Cref{lem:insert:runtime}.

This finishes the proof of \Cref{lem:insert}.
\end{proof}

\subsection{The operation {\Delete}}
\label{subsec:delete}

\newcommand{\vol}{\mathrm{vol}}

This subsection presents the operations {\Delete} and {\DeleteSubclustering} -- see \Cref{fig:delete} for their implementation -- which address the deletion of a {\em current center} $\cdelete \in C$ (promised) from the maintained center set $C$.
Without ambiguity, throughout \Cref{subsec:delete} we denote by $\calD$ and $\calD'$ the data structures before and after the operation $\Delete(\cdelete)$, respectively; likewise for their contents, the maintained center sets $C$ versus $C'$, the maintained center subsets $C_{J}$ versus $C'_{J}$, etc.
Essentially, we will utilize an adaptation of Dijkstra's algorithm \cite[Chapter~22.3]{CLRS22}.

The following \Cref{lem:delete} shows the performance guarantees of the operation {\Delete}.

\begin{lemma}[{\Delete}]
\label{lem:delete}
\begin{flushleft}
Invoke the operation $\Delete(\cdelete)$, on input a current center $\cdelete \in C$ (promised), for a data structure $\calD$ that maintains Invariants~\ref{invar:subclusterings:center} to \ref{invar:subclusterings:empty} and Invariants~\ref{invar:clustering} to \ref{invar:DSvolume}:
\begin{enumerate}[font = {\em\bfseries}]
    \item \label{lem:delete:center}
    $C' = C - \cdelete$ and $C'_{J} = C' \cap J$, for every index $J \in \calJ$.
    
    \item \label{lem:delete:loss}
    $\DScost' \le \DScost + \DSloss[\cdelete]$.
    
    \item \label{lem:delete:potential}
    $\Phi' - \Phi \le \DSvolume[\cdelete] \cdot \Phi_{\max}$.
    
    \item \label{lem:delete:invar}
    The modified data structure $\calD'$ maintains Invariants~\ref{invar:subclusterings:center} to \ref{invar:subclusterings:empty} and thus Invariants~\ref{invar:clustering} to \ref{invar:DSvolume}.
    
    \item \label{lem:delete:runtime}
    The worst-case running time $\Tdelete = \DSvolume[\cdelete] \cdot O(m \log^{3}(n))$.
\end{enumerate}
\end{flushleft}
\end{lemma}

\afterpage{
\begin{figure}[t]
\centering
\begin{mdframed}
\begin{flushleft}
Operation $\term[\Delete]{alg:delete}(\cdelete)$

\vspace{.1in}
{\bf Input:}
A current center $\cdelete \in C$ (promised).
\begin{enumerate}
    \item\label{alg:delete:for}
    For every index $(J \in \calJ \colon J \ni \cdelete)$:
    
    \item \label{alg:delete:subclustering}
    \qquad $\DeleteSubclustering(J, \cdelete)$.
    
    \item \label{alg:delete:center-subset}
    \qquad $C_{J} \gets C_{J} - \cdelete$.
    
    \item \label{alg:delete:center-set}
    $C \gets C - \cdelete$.
\end{enumerate}
\end{flushleft}
\end{mdframed}

\begin{mdframed}
\begin{flushleft}
Suboperation $\term[\DeleteSubclustering]{alg:delete-subclustering}(J, \cdelete)$
\hfill
\Comment{Adapted from Dijkstra's algorithm.}

\vspace{.1in}
{\bf Input:}
An index $J \in \calJ$ and a current center $\cdelete \in C_{J}$ (promised).

\Comment{Synchronize the rest of our data structure $\calD$ accordingly (\Cref{lem:maintain-clustering,lem:maintain-DScost,lem:maintain-deletion-estimator,lem:maintain-grouping}), every time when the index-$(J \in \calJ)$ subclustering $(\DSc_{J}[v], \DSd_{J}[v])_{v \in V}$ is modified in Line~\ref{alg:insert-subclustering:modify}.}
\begin{enumerate}
\setcounter{enumi}{4}
    \item \label{alg:delete-subclustering:subcluster}
    $U_{J} \gets \{u \in V \mid \DSc_{J}[u] = \cdelete\}$.
    \hfill
    \Comment{The index-$J$ center-$\cdelete$ subcluster.}
    
    \item \label{alg:delete-subclustering:non-universe}
    If $(U_{J} \ne V) \iff (C_{J} \supsetneq \{\cdelete\})$:
    
    \item \label{alg:delete-subclustering:erase}
    \qquad $(\DSc_{J}[v], \DSd_{J}[v])_{v \in U_{J}} \gets (\perp, +\infty)^{|U_{J}|}$.
    
    \item \label{alg:delete-subclustering:boundary}
    \qquad $\partial U_{J} \gets N(U_{J}) \setminus U_{J}$.
    \hfill
    \Comment{The outer boundary of $U_{J}$.}
    
    \item \label{alg:delete-subclustering:queue}
    \qquad Build a $\DSd_{J}[v]$-minimizing priority queue $\calQ \eqdef \calQ((\DSc_{J}[v], \DSd_{J}[v])_{v \in U_{J} \cup \partial U_{J}})$.\textsuperscript{\ref{footnote:break-tie}}
    
    \item \label{alg:delete-subclustering:while}
    \qquad While the priority queue $\calQ$ is nonempty:
    
    \item \label{alg:delete-subclustering:pop-min}
    \qquad $(\DSc_{J}[v^{*}], \DSd_{J}[v^{*}]) \gets \calQ.\texttt{pop-min}()$.
    
    \item \label{alg:delete-subclustering:for}
    \qquad\qquad For every neighbor $\big(u \in N(v^{*}) \cap U_{J} \colon \DSd_{J}[u] > 2^{\eps / \beta} \cdot (\DSd_{J}[v^{*}] + w(u, v^{*}))\big)$:
    
    \item \label{alg:delete-subclustering:modify}
    \qquad\qquad\qquad $(\DSc_{J}[u], \DSd_{J}[u]) \gets (\DSc_{J}[v^{*}],\ 2^{\eps / \beta} \cdot (\DSd_{J}[v^{*}] + w(u, v^{*})))$.
    
    \item \label{alg:delete-subclustering:priority}
    \qquad\qquad\qquad $\calQ.\texttt{update-priority}(u, (\DSc_{J}[u], \DSd_{J}[u]))$.
    
    \item \label{alg:delete-subclustering:universe}
    Otherwise, $(U_{J} = V) \iff (C_{J} = \{\cdelete\})$:
    
    \item \label{alg:delete-subclustering:reset}
    \qquad $(\DSc_{J}[v], \DSd_{J}[v])_{v \in V} \gets (\perp, \dmax)^{n}$.
\end{enumerate}
\end{flushleft}
\end{mdframed}
\caption{\label{fig:delete}The operations {\Delete} and {\DeleteSubclustering}.}
\end{figure}
\clearpage}

\begin{proof}
The operation {\Delete} iterates the following for every index $(J \in \calJ \colon J \ni \cdelete)$:
\hfill
(Line~\ref{alg:delete:for})

\noindent
\Comment{$C_{J} = C \cap J \supseteq \{\cdelete\}$, since $C \ni \cdelete$ and $J \ni \cdelete$ (promised).}

\noindent
We invoke the suboperation {\DeleteSubclustering} (Lines~\ref{alg:delete-subclustering:subcluster} to \ref{alg:delete-subclustering:priority}), which adapts Dijkstra's algorithm \cite[Chapter~22.3]{CLRS22}, to modify the index-$J$ subclustering $(\DSc_{J}[v], \DSd_{J}[v])_{v \in V}$.
\hfill
(Line~\ref{alg:delete:subclustering})

\vspace{.1in}
\noindent
We first identify the index-$J$ center-$\cdelete$ {\em subcluster} $U_{J} = \{v \in V \mid \DSc_{J}[v] = \cdelete\} \ne \emptyset$.
\hfill
(Line~\ref{alg:delete-subclustering:subcluster})

\noindent
(i)~If this subcluster $U_{J}$ is {\em not} the universe $(U_{J} \ne V) \iff (C_{J} \supsetneq \{\cdelete\})$:
\hfill
(Line~\ref{alg:delete-subclustering:non-universe})

\noindent
We would ``erase'' its maintenance $(\DSc_{J}[v], \DSd_{J}[v])_{v \in U_{J}} \gets (\perp, +\infty)^{|U_{J}|}$,
\hfill
(Line~\ref{alg:delete-subclustering:erase})\\
identify its {\em outer boundary} $\partial U_{J} = N(U_{J}) \setminus U_{J} \ne \emptyset$, and
\hfill
(Line~\ref{alg:delete-subclustering:boundary})\\
build a $\DSd_{J}[v]$-minimizing {\em priority queue} $\calQ = \calQ((\DSc_{J}[v], \DSd_{J}[v])_{v \in U_{J} \cup \partial U_{J}})$.
\hfill
(Line~\ref{alg:delete-subclustering:queue})

\noindent
\Comment{$\DSc_{J}[v] \in C_{J} - \cdelete$, for every vertex $v \in \partial U_{J}$ in the outer boundary. (This is vacuously true in case of $C_{J} = \{\cdelete\}$, by which $\partial U_{J} = \emptyset$.)}

\noindent
\Comment{If two pairs $(\DSc_{J}[u], \DSd_{J}[u]), (\DSc_{J}[v], \DSd_{J}[v])$ in the priority queue $\calQ$ have the same $\DSd_{J}[u] = \DSd_{J}[v]$ values, we break ties in favor of the pair with $\DSc_{J}[u], \DSc_{J}[v] \ne \perp$ (if any) but otherwise arbitrarily (\Cref{footnote:break-tie}).}

\noindent
Afterward, we move on to the iteration of Lines~\ref{alg:delete-subclustering:pop-min} to \ref{alg:delete-subclustering:priority}, until the priority queue gets empty $\calQ = \emptyset$ (\`{a} la Dijkstra's algorithm).
Specifically, a single iteration works as follows:
\hfill
(Line~\ref{alg:delete-subclustering:while})

\noindent
First, pop off the current $\DSd_{J}[v]$-minimizer, say the pair $(\DSc_{J}[v^{*}], \DSd_{J}[v^{*}])$.
\hfill
(Line~\ref{alg:delete-subclustering:pop-min})

\noindent
\Comment{$\DSc_{J}[v^{*}] \in C_{J} - \cdelete$, namely we must have $\DSc_{J}[v^{*}] \ne \perp$, since the considered pair $(\DSc_{J}[v^{*}], \DSd_{J}[v^{*}])$ is the $\DSd_{J}[v]$-minimizer of the priority queue $\calQ$ and we break ties in favor of the pairs $(\DSc_{J}[v], \DSd_{J}[v])$ with $\DSc_{J}[v] \ne \perp$ (\Cref{footnote:break-tie}).}

\noindent
Then, test Invariant~\ref{invar:subclusterings:edge} for every neighbor $u \in N(v^{*}) \cap U_{J}$:
\hfill
(Line~\ref{alg:delete-subclustering:for})

\noindent
A violating neighbor $\big(u \in N(v^{*}) \cap U_{J} \colon \DSd_{J}[u] > 2^{\eps / \beta} \cdot (\DSd_{J}[v^{*}] + w(u, v^{*}))\big)$ shall move to the index-$J$ center-$\DSc_{J}[v^{*}]$ cluster, i.e., given its violation to Invariant~\ref{invar:subclusterings:edge},
the center $\DSc_{J}[v^{*}] \in C_{J} - \cdelete$ is much nearer than its current center $\DSc_{J}[u] \in C_{J}$.
So we modify the entry $(J, u)$ of the subclustering, namely $(\DSc'_{J}[u], \DSd'_{J}[u]) \gets (\DSc_{J}[v^{*}],\ 2^{\eps / \beta} \cdot (d + w(u, v^{*})))$, and update this violating neighbor $u$'s priority, namely $\calQ.\texttt{update-priority}(u, (\DSc_{J}[u], \DSd_{J}[u]))$.
\hfill
(Lines~\ref{alg:delete-subclustering:modify} and \ref{alg:delete-subclustering:priority})

\noindent
(ii)~Otherwise, this subcluster $U_{J}$ {\em is} the universe $(U_{J} = V) \iff (C_{J} = \{\cdelete\})$:
\hfill
(Line~\ref{alg:delete-subclustering:universe})

\noindent
We would ``reset'' this index-$J$ subclustering $(\DSc_{J}[v], \DSd_{J}[v])_{v \in V} \gets (\perp, \dmax)^{n}$.
\hfill
(Line~\ref{alg:delete-subclustering:reset})

\noindent
\Comment{Throughout the suboperation {\DeleteSubclustering}, every time when Line~\ref{alg:delete-subclustering:modify} modifies the index-$J$ subclustering $(\DSc_{J}[v], \DSd_{J}[v])_{v \in V}$, we would synchronize the rest of our data structure $\calD$ using \Cref{lem:maintain-clustering,lem:maintain-DScost,lem:maintain-deletion-estimator,lem:maintain-grouping}, in time $O(\log(n)) + O(\log(n)) + O(\log(n)) = O(\log(n))$.}

\vspace{.1in}
\noindent
After the suboperation {\DeleteSubclustering} terminates -- which we assume for the moment but will establish later (cf.\ \Cref{lem:delete:runtime}) -- we complete our modifications for this index $(J \in \calJ \colon J \ni \cdelete)$ by maintaining the index-$J$ center subset $C'_{J} \gets C_{J} - \cdelete$.
\hfill
(Line~\ref{alg:delete:center-subset})

\noindent
After the iteration of all indices $(J \in \calJ \colon J \ni \cdelete)$ terminates, we complete our modifications to the whole data structure $\calD$ by maintaining the center set $C' \gets C - \cdelete$.
\hfill
(Line~\ref{alg:delete:center-set})

As the whole process of the operation {\Delete} now are clear, we are ready to show \Cref{lem:delete}.

\vspace{.1in}
\noindent
{\bf \Cref{lem:delete:center}.}
Lines~\ref{alg:delete:center-subset} and \ref{alg:delete:center-set} trivially imply that $C' = C - \cdelete$ and $C'_{J} = C' \cap J$, for every index $J \in \calJ$.

\vspace{.1in}
\noindent
{\bf \Cref{lem:delete:loss}.}
Recall that the objective estimator $\DScost = \sum_{v \in V} \DSd^{z}[v]$ (likewise for $\DScost'$) and the deletion objective estimator $\DSloss[\cdelete] = \sum_{v \in V \colon \DSc[v] = \cdelete} ((\min_{J \in \calJ \colon J \notni \cdelete} \DSd_{J}^{z}[v]) - \DSd^{z}[v])$ (Invariants~\ref{invar:DScost} and \ref{invar:DSloss}).
\begin{align*}
    \DScost + \DSloss[\cdelete] &~=~ \sum_{v \in V \colon \DSc[v] \neq \cdelete} \DSd^z[v] + \sum_{v \in V \colon \DSc[v] = \cdelete} \min_{J \in \calJ \colon J \notni \cdelete} \DSd_{J}^{z}[v].
\end{align*}
We compare the formulae of $\DScost'$ and $\DScost + \DSloss[\cdelete]$, for every vertex-$(v \in V)$, as follows.

\noindent
{\bf Case~1: $\DSc[v] \neq \cdelete$.}
Let $J_{v} \eqdef \argmin_{J \in \calJ} \DSd_{J}[v]$, for which $(\DSc[v], \DSd[v]) = (\DSc_{J_{v}}[v], \DSd_{J_{v}}[v])$ (Invariant~\ref{invar:clustering}).
The considered vertex $v$ is outside the index-$J_{v}$ center-$\cdelete$ subcluster $U_{J_{v}} = \{u \in V \mid \DSc_{J_{v}}[u] = \cdelete\}$.
Thus, the index-$J_{v}$ suboperation $\DeleteSubclustering(J_{v}, \cdelete)$ does not modify the entry $(J_{v}, v)$, namely  $(\DSc_{J_{v}}'[v], \DSd'_{J_{v}}[v]) = (\DSc_{J_{v}}[v], \DSd_{J_{v}}[v])$, which implies that 
\begin{align*}
    \DSd'[v]
    ~=~ \min_{J \in \calJ} \DSd'_{J}[v]
    ~\le~ \DSd'_{J_{v}}[v]
    ~=~ \DSd_{J_{v}}[v]
    ~=~ \DSd[v].
\end{align*}
Here the first step applies Invariant~\ref{invar:clustering} to the modified data structure $\calD'$, which we assume for the moment but will prove later in \Cref{lem:delete:invar}.

\noindent
{\bf Case~2: $\DSc[v] = \cdelete$.}
Let $K_{v} \eqdef \argmin_{J \in \calJ \colon J \notni \cdelete} \DSd_{J}[v]$,\footnote{This $K_{v}$ is well-defined, since $\cup_{J \in \calJ \colon J \notni \cdelete} J = V - \cdelete$ (\Cref{def:cover}).} for which $\DSc_{K_{v}}[v] \in C_{K_{v}} = C \cap K_{v} \implies \DSc_{K_{v}}[v] \neq \cdelete$ (Invariant~\ref{invar:subclusterings:vertex}).
Hence, the index-$K_{v}$ suboperation $\DeleteSubclustering(K_{v}, \cdelete)$ does not modify the entry $(K_{v}, v)$, namely $(\DSc_{K_{v}}'[v], \DSd_{K_{v}}'[v]) = (\DSc_{K_{v}}[v], \DSd_{K_{v}}[v])$, which implies that 
\begin{align*}
    \DSd'[v] ~=~ \min_{J \in \calJ} \DSd'_{J}[v] ~\le~ \DSd_{K_{v}}'[v] ~=~ \DSd_{K_{v}}[v] ~=~ \min_{J \in \calJ \colon J \notni \cdelete} \DSd_{J}[v],
\end{align*}
Here the first step also applies Invariant~\ref{invar:clustering} to the modified data structure $\calD'$, which we assume for the moment but will prove later in \Cref{lem:delete:invar}.

Combining both cases proves \Cref{lem:delete:loss}.

\vspace{.1in}
\noindent
{\bf \Cref{lem:delete:potential}.}
Recall that the potential formula $\Phi = \sum_{(J,\, v) \in \calJ \times V} \deg(v) \cdot \log_{2}(1 + \DSd_{J}[v])$; likewise for $\Phi'$ (\Cref{fig:DS}).
Moreover, the operation {\Delete} only modifies the entries $\big((J, v) \in \calJ \times V \colon \DSc_{J}[v] = \cdelete\big)$, from $(\DSc_{J}[v] ,\DSd_{J}[v])$ to $(\DSc'_{J}[v] ,\DSd'_{J}[v])$ (say).
Thus, the operation {\Delete} induces a potential change of
\begin{align*}
    \Phi' - \Phi
    & ~=~ \sum_{(J,\, v) \in \calJ \times V} \deg(v) \cdot \log_{2}\Big(\tfrac{1 + \DSd'_{J}[v]}{1 + \DSd_{J}[v]}\Big) \\
    & ~=~ \sum_{(J,\, v) \in \calJ \times V \colon \DSc_{J}[v] = \cinsert} \deg(v) \cdot \log_{2}\Big(\tfrac{1 + \DSd'_{J}[v]}{1 + \DSd_{J}[v]}\Big) \\
    & ~\le~ \sum_{(J,\, v) \in \calJ \times V \colon \DSc_{J}[v] = \cinsert} \deg(v) \cdot \log_{2}(1 + \dmax)
    \tag{\Cref{lem:subclusterings}} \\
    & ~=~ \DSvolume[\cdelete] \cdot \Phi_{\max}.
    \tag{Invariant~\ref{invar:DSvolume} and \Cref{lem:potential}}
\end{align*}

\noindent
{\bf \Cref{lem:delete:invar}.}
We would prove that the modified subclusterings $(\DSc'_{J}[v], \DSd'_{J}[v])_{(J,\, v) \in \calJ \times V}$ maintain Invariants~\ref{invar:subclusterings:center} to \ref{invar:subclusterings:empty}; suppose so, the rest of the modified data structure $\calD'$ maintains Invariants~\ref{invar:clustering} to \ref{invar:DSvolume} by construction, given the respective ``synchronization'' parts of \Cref{lem:maintain-clustering,lem:maintain-DScost,lem:maintain-deletion-estimator,lem:maintain-grouping}.
Invariant~\ref{invar:subclusterings:empty} is rather trivial.\footnote{Invariant~\ref{invar:subclusterings:empty} asserts that: If $(C'_{J} = \emptyset) \iff (C_{J} = \emptyset) \vee (C_{J} = \{\cdelete\})$, then $(\DSc'_{J}[v], \DSd'_{J}[v])_{v \in V} = (\perp, \dmax)^{n}$.\\
(i)~If $C_{J} = \emptyset$, then this index-$J$ subclustering is unmodified $(\DSc'_{J}[v], \DSd'_{J}[v])_{v \in V} = (\DSc_{J}[v], \DSd_{J}[v])_{v \in V} = (\perp, \dmax)^{n}$.\\
(ii)~If $C_{J} = \{\cdelete\}$, then this index-$J$ subclustering (Line~\ref{alg:delete-subclustering:reset}) is ``reset'' to $(\DSc'_{J}[v], \DSd'_{J}[v])_{v \in V} \gets (\perp, \dmax)^{n}$.}
Below, we would establish Invariants~\ref{invar:subclusterings:center} to \ref{invar:subclusterings:vertex} for a specific index $(J \in \calJ \colon C'_{J} = C_{J} - \cdelete \neq \emptyset)$. We safely assume that $C_{J} \ni \cdelete$; otherwise, Invariants~\ref{invar:subclusterings:center} to \ref{invar:subclusterings:vertex} also become trivial, as the index-$J$ subclustering is unmodified $(\DSc'_{J}[v], \DSd'_{J}[v])_{v \in V} = (\DSc_{J}[v], \DSd_{J}[v])_{v \in V}$.

We first establish the following \Cref{claim:delete-subclustering} for this index $J \in \calJ$.

\begin{claim}
\label{claim:delete-subclustering}
\begin{flushleft}
Every vertex $u \in U_{J}$ admits a $\partial U_{J}$-to-$u$-path $(\partial U_{J} \ni x_{0}), x_{1}, \dots, (x_{\ell} \equiv u)$ such that $\DSc'_{J}[u] = \DSc'_{J}[x_{0}]$ and $\DSd'_{J}[x_{i}] = 2^{\eps / \beta} \cdot(\DSd'_{J}[x_{i - 1}] + w(x_{i - 1}, x_{i}))$, $\forall i \in [\ell]$.
\end{flushleft}
\end{claim}

\begin{proof}
We claim that every vertex $u \in U_{J}$ admits a neighbor $v \in N(u)$ such that $\DSc'_{J}[u] = \DSc'_{J}[v]$ and $\DSd'_{J}[u] = 2^{\eps / \beta} \cdot (\DSd'_{J}[v] + w(u, v))$, which we call a {\em predecessor} $v \in N(u)$ of $u \in U_{J}$.\footnote{Notice that if a neighbor $v \in N(u)$ is a predecessor of a vertex $u \in U_{J}$, then $u$ cannot conversely be a predecessor of $v$, namely $\DSd'_{J}[u] = 2^{\eps / \beta} \cdot (\DSd'_{J}[v] + w(u, v)) \implies \DSd'_{J}[v] \ne 2^{\eps / \beta} \cdot (\DSd'_{J}[u] + w(u, v))$.}
Suppose so, then we can obtain a desirable $\partial U_{J}$-to-$u$-path $(\partial U_{J} \ni x_{0}), x_{1}, \dots, (x_{\ell} \equiv u)$, by starting from this vertex $v \in U_{J}$ and finding the predecessors recursively.

During the iteration of Lines~\ref{alg:delete-subclustering:pop-min} to \ref{alg:delete-subclustering:priority}, every pair $(\DSc'_{J}[u], \DSd'_{J}[u])$ for $u \in U_{J}$ must be modified at least once.
Namely, those pairs $(\DSc_{J}[v], \DSd_{J}[v])_{v \in U_{J}} \gets (\perp, +\infty)^{|U_{J}|}$ are ``erased'' initially (Line~\ref{alg:delete-subclustering:erase}), and the underlying graph $G = (V, E, w)$ is connected.
Now, let us consider a specific vertex $u \in U_{J}$ and the last modification to the pair $(\DSc'_{J}[u], \DSd'_{J}[u])$ (Line~\ref{alg:delete-subclustering:modify}), say
\begin{align*}
    (\DSc'_{J}[u], \DSd'_{J}[u])
    ~\gets~ (\DSc'_{J}[v^{*}],\ 2^{\eps / \beta} \cdot (\DSd'_{J}[v^{*}] + w(u, v^{*}))).
\end{align*}
After this last modification, the pair $(\DSc'_{J}[u], \DSd'_{J}[u])$ always keeps the same, and so does the other pair $(\DSc'_{J}[v^{*}], \DSd'_{J}[v^{*}])$, since it had already been popped off from the priority queue $\calQ$ (Line~\ref{alg:delete-subclustering:pop-min}).
Hence, this neighbor $v^{*} \in N(u)$ is a predecessor of the considered vertex $u \in U_{J}$.

This finishes the proof of \Cref{claim:delete-subclustering}.
\end{proof}

Now we move back to verifying Invariants~\ref{invar:subclusterings:center} to \ref{invar:subclusterings:vertex}.

\vspace{.1in}
\noindent
{\bf Invariant~\ref{invar:subclusterings:center}:}
$(\DSc'_{J}[c], \DSd'_{J}[c]) = (c, 0)$, for every center $c \in C'_{J} = C_{J} - \cdelete$.

\noindent
Every survival center $c \in C'_{J} = C_{J} - \cdelete$, for which $(\DSc_{J}[c], \DSd_{J}[c]) = (c, 0)$, is outside the index-$J$ center-$\cdelete$ subcluster $c \notin U_{J} = \{v \in V \mid \DSc_{J}[v] = \cdelete\}$, so the entry $(\DSc'_{J}[c], \DSd'_{J}[c]) = (\DSc_{J}[c], \DSd_{J}[c]) = (c, 0)$ is unmodified, thus trivially maintaining Invariant~\ref{invar:subclusterings:center}.

\vspace{.1in}
\noindent
{\bf Invariant~\ref{invar:subclusterings:edge}:}
$\DSd'_{J}[u] \le 2^{\eps / \beta} \cdot (\DSd'_{J}[v] + w(u, v))$, for every edge $(u, v) \in E$.

\vspace{.1in}
\noindent
{\bf Case~1: $(u \notin U_{J}) \wedge (v \notin U_{J})$.}
Since both vertices $u$ and $v$ are outside to the subcluster $U_{J}$, both approximate distances $\DSd'_{J}[v] = \DSd_{J}[v]$ and $\DSd'_{J}[u] = \DSd_{J}[u]$ are unmodified, thus trivially maintaining Invariant~\ref{invar:subclusterings:edge}.

\vspace{.1in}
\noindent
{\bf Case~2: $(u \notin U_{J}) \wedge (v \in U_{J})$.}
There exists a $\partial U_{J}$-to-$v$-path $(\partial U_{J} \ni x_{0}),x_{1}, \dots, (x_{\ell} \equiv v)$ such that $\DSd'_{J}[x_{i}] = 2^{\eps / \beta} \cdot (\DSd'_{J}[x_{i - 1}] + w(x_{i - 1}, x_{i}))$, $\forall i \in [\ell]$ (\Cref{claim:delete-subclustering}).
The original data structure $\calD$ satisfies $\DSd_{J}[x_{i}] \le 2^{\eps / \beta} \cdot (\DSd_{J}[x_{i - 1}] + w(x_{i - 1}, x_{i}))$, $\forall i \in [\ell]$ (Invariant~\ref{invar:subclusterings:edge}).
The ``source'' $x_{0} \in \partial U_{J}$ is outside the subcluster $U_{J}$, so its approximate distance $\DSd'_{J}[x_{0}] = \DSd_{J}[x_{0}]$ is unmodified.
By induction over the path, we have $\DSd_{J}[x_{\ell}] \le \DSd'_{J}[x_{\ell}] \iff \DSd_{J}[v] \le \DSd'_{J}[v]$ and can deduce Invariant~\ref{invar:subclusterings:edge} as follows:
\begin{align*}
    \DSd'_{J}[u]
    ~=~ \DSd_{J}[u]
    ~\le~ 2^{\eps / \beta} \cdot (\DSd_{J}[v] + w(u, v))
    ~\le~ 2^{\eps / \beta} \cdot (\DSd'_{J}[v] + w(u, v)).
\end{align*}
Here, the first step holds since the vertex $u$ is outside the subcluster $U_{J}$, and the second step applies Invariant~\ref{invar:subclusterings:edge} to the original data structure $\calD$.

\vspace{.1in}
\noindent
{\bf Case~3: $u \in U_{J}$.}
When the priority queue $\calQ$ is built (Line~\ref{alg:delete-subclustering:queue}), it includes the pair $(\DSc_{J}[v], \DSd_{J}[v])$ since $v \in N(U_{J}) \subseteq U_{J} \cup \partial U_{J} \impliedby u \in U_{J}$.
In the moment that this pair $(\DSc_{J}[v], \DSd_{J}[v])$ is popped off (Line~\ref{alg:delete-subclustering:pop-min}), if the edge $(u, v)$ violates Invariant~\ref{invar:subclusterings:edge}, namely $\DSd'_{J}[u] > 2^{\eps / \beta} \cdot (\DSd'_{J}[v] + w(u, v))$, we will detect this violation (Line~\ref{alg:delete-subclustering:for}) and modify this approximate distance ``$\DSd'_{J}[u] \gets 2^{\eps / \beta} \cdot (\DSd'_{J}[v] + w(u, v))$'' to retain Invariant~\ref{invar:subclusterings:edge} (Line~\ref{alg:delete-subclustering:modify}).
Hereafter, $\DSd'_{J}[u]$ can only decrease (Lines~\ref{alg:delete-subclustering:for} to \ref{alg:delete-subclustering:priority}) and $\DSd'_{J}[v]$ always keeps the same (since the pair $(\DSc_{J}[v], \DSd_{J}[v])$ had already been popped off); thus this edge $(u, v)$ will keep maintaining Invariant~\ref{invar:subclusterings:edge}.

Combining all three cases gives Invariant~\ref{invar:subclusterings:edge}.

\vspace{.1in}
\noindent
{\bf Invariant~\ref{invar:subclusterings:vertex}:}
$\DSc'_{J}[v] \in C'_{J}$ and $\dist(v, C'_{J}) \le \dist(v, \DSc'_{J}[v]) \le \DSd'_{J}[v]$, for every vertex $v \in V$.

\vspace{.1in}
\noindent
{\bf Case~1: $v \notin U_{J}$.}
The vertex $v$ is outside the subcluster $U_{J}$, so the entry $(\DSc'_{J}[v], \DSd'_{J}[v]) = (\DSc_{J}[v], \DSd_{J}[v])$ is unmodified, thus trivially maintaining Invariant~\ref{invar:subclusterings:vertex}, namely $\DSc'_{J}[v] = \DSc_{J}[v] \in C_{J} - \cdelete = C'_{J} \implies \dist(v, C'_{J}) \le \dist(v, \DSc'_{J}[v]) = \dist(v, \DSc_{J}[v]) \le \DSd_{J}[v] = \DSd'_{J}[v]$.

\vspace{.1in}
\noindent
{\bf Case~2: $v \in U_{J}$.}
There exists a $\partial U_{J}$-to-$v$-path $(\partial U_{J} \ni x_{0}), x_{1}, \dots, (x_{\ell} \equiv v)$ such that $\DSc'_{J}[v] = \DSc'_{J}[x_{0}]$ (\Cref{claim:delete-subclustering}).
The ``source'' $x_{0} \in \partial U_{J}$ is outside the subcluster $U_{J}$; as we have established in {\bf Case~1}, $(\DSc'_{J}[x_{0}] \in C'_{J}) \wedge (\dist(x_{0}, \DSc'_{J}[x_{0}]) \le \DSd'_{J}[x_{0}])$.

Therefore, we have $\DSc'_{J}[v] = \DSc'_{J}[x_{0}] \in C'_{J} \implies \dist(v, C'_{J}) \le \dist(v, \DSc'_{J}[v])$.
\Cref{claim:delete-subclustering} also asserts that $\DSd'_{J}[x_{i}] = 2^{\eps / \beta} \cdot (\DSd'_{J}[x_{i - 1}] + w(x_{i - 1}, x_{i})) \ge \DSd'_{J}[x_{i - 1}] + w(x_{i - 1}, x_{i})$, $\forall i \in [\ell]$.
By induction over the path $(\partial U_{J} \ni x_{0}), x_{1}, \dots, (x_{\ell} \equiv v)$, we can deduce Invariant~\ref{invar:subclusterings:vertex} as follows:
\begin{align*}
    \DSd'_{J}[v]
    & ~\ge~ \DSd'_{J}[x_{0}] + \sum_{i\in [\ell]} w(x_{i - 1}, x_{i}) \\
    & ~\ge~ \dist(x_{0}, \DSc'_{J}[x_{0}]) + \sum_{i\in [\ell]} w(x_{i - 1}, x_{i})
    \tag{$\DSd'_{J}[x_{0}] \ge \dist(x_{0}, \DSc'_{J}[x_{0}])$} \\
    &~\ge~ \dist(x_{\ell}, \DSc'_{J}[x_{0}])
    \tag{Triangle inequality} \\
    &~=~ \dist(v, \DSc'_{J}[v]).
    \tag{$x_{\ell} = v$ and $\DSc'_{J}[x_{0}] = \DSc'_{J}[v]$}
\end{align*}

Combining both cases gives Invariant~\ref{invar:subclusterings:vertex}.

\vspace{.1in}
\noindent
{\bf \Cref{lem:delete:runtime}.}
Dijkstra's algorithm has running time $O(m + n \log(n))$ on a $n$-vertex $m$-edge connected graph \cite[Chapter~22.3]{CLRS22}.
Regarding our adaptation, instead, every index-$(J \in \calJ \colon J \ni \cdelete)$ invocation of the suboperation {\DeleteSubclustering} involves $n_{J} \eqdef |U_{J} \cup \partial U_{J}| \le 2\vol(U_{J})$ vertices and $m_{J} \eqdef |E \cap (U_{J} \times (U_{J} \cup \partial U_{J}))| \le \vol(U_{J})$ edges.
These facts in combination with \Cref{lem:maintain-clustering,lem:maintain-DScost,lem:maintain-deletion-estimator,lem:maintain-grouping} imply the running time of the operation {\Delete}, as follows.
\begin{align*}
    \Tdelete
    & ~=~ \sum_{J \in \calJ \colon J \ni \cdelete} (m_{J} + n_{J} \log(n_{J})) \cdot O(\log(n)) \\
    & ~=~ \Big(\sum_{J \in \calJ \colon J \ni \cdelete} \vol(U_{J})\Big) \cdot O(\log^{2}(n)) \\
    & ~=~ \DSvolume[\cdelete] \cdot 2m |\calJ| \cdot O(\log^{2}(n))
    \tag{Invariant~\ref{invar:DSvolume}} \\
    & ~=~ \DSvolume[\cdelete] \cdot O(m \log^{3}(n)).
    \tag{\Cref{def:cover,lem:cover}}
\end{align*}
This finishes the proof of \Cref{lem:delete}.
\end{proof}

\subsection{The operation {\SampleNoncenter}}
\label{subsec:sample-noncenter}

This subsection presents the operation {\SampleNoncenter} -- see \Cref{fig:sample-noncenter} for its implementation -- which leverages the binary search tree $\BST$ (Invariant~\ref{invar:binary-search-tree}) to sample a random vertex $r \in V$ that is ensured to be a noncenter $r \notin C$, almost surely.
The following \Cref{lem:sample-noncenter} provides the performance guarantees of this operation.
(Herein, everything is naively adapted from \cite[Algorithm~2 and Lemma~4.2]{CLNSS20} and is included just for completeness.)

\begin{figure}[t]
\centering
\begin{mdframed}
\begin{flushleft}
Operation $\term[\SampleNoncenter]{alg:sample-noncenter}()$

\vspace{.1in}
{\bf Output:}
A random vertex $r \in V$.
\begin{enumerate}
    \item \label{alg:sample-noncenter:root} $(U, \sfvalue(U))\gets \BST.\texttt{root} = (V, \sum_{v \in V} \DSd^{z}[v])$.
    
    \item \label{alg:sample-noncenter:while} While $(U, \sfvalue(U))$ is not a leaf:
    
    \item \label{alg:sample-noncenter:go-down} 
    \qquad $(U, \sfvalue(U)) \gets 
    \begin{cases}
        (U_{\sf left}, \sfvalue(U_{\sf left})) & \text{w.p. } \frac{\sfvalue(U_{\sf left})}{\sfvalue(U)}\\
        (U_{\sf right}, \sfvalue(U_{\sf right})) & \text{w.p. } \frac{\sfvalue(U_{\sf right})}{\sfvalue(U)}
    \end{cases}$.
    \hfill
    \Comment{The left/right child.}

    \item \label{alg:sample-noncenter:return}
    Return $r \gets$ ``the unique vertex in the singleton $U$''.
    \hfill
    \Comment{$(U, \sfvalue(U))$ is a leaf.}
\end{enumerate}
\end{flushleft}
\end{mdframed}
\caption{\label{fig:sample-noncenter}The operation {\SampleNoncenter}.}
\end{figure}

\begin{lemma}[{\SampleNoncenter}]
\label{lem:sample-noncenter}
\begin{flushleft}
The randomized operation $\SampleNoncenter()$ has worst- case running time $O(\log(n))$ and returns a specific vertex $v \in V$ with probability $\frac{\DSd^{z}[v]}{\DScost} = \frac{\DSd^{z}[v]}{\sum_{v' \in V} \DSd^{z}[v']}$.\\
\Comment{It never returns a current center $c \in C$, by which $\DSd^{z}[c] = 0$ (\Cref{lem:clustering:vertex} of \Cref{lem:clustering}), almost surely.}
\end{flushleft}
\end{lemma}

\begin{proof}
Obvious. By construction (Invariant~\ref{invar:binary-search-tree}), the binary search tree $\BST$ has height $O(\log(n))$ and a specific vertex $v \in V$ will be returned with probability $\Pr_{r}[r = v] = \frac{\DSd^{z}[v]}{\sum_{v' \in V} \DSd^{z}[v']} = \frac{\DSd^{z}[v]}{\DScost}$.
\end{proof}

\section{Local Search for \texorpdfstring{the {\kzC} Problem}{}}
\label{sec:LS}

In this section, we will show our algorithm {\LocalSearch} and the subroutine {\TestEffectiveness} -- see \Cref{fig:LocalSearch} for their implementation.
Regarding the underlying graph $G = (V, E, w)$, we would impose \Cref{assumption:LS:hop-bounded} (in addition to \Cref{assumption:edge-weight} about edge weights) throughout \Cref{sec:LS}.

\afterpage{
\begin{figure}[t]
\centering
\begin{mdframed}
\begin{flushleft}
Procedure $\term[\LocalSearch]{alg:LS}(G, k)$

\vspace{.1in}
{\bf Input:}
A $(\beta, \eps)$-hop-bounded graph $G$, with $0 < \eps < \frac{1}{10z}$, and a variable $k \in [n]$.

{\bf Setup:}
A large enough integer parameter $\DSnumber = \eps^{-\Theta(z)} \log(n)$.

\white{\bf Setup:}
A number of $(\DSnumber + 1)$ data structures $\calD$ and $\{\calD'_{\sigma}\}_{\sigma \in [\DSnumber]}$.

\Comment{Only during (Line~\ref{alg:LS:schedule}) the schedule of the subroutines, $\TestEffectiveness(\calD, \calD'_{\sigma}, \cinsert_{\sigma})$ \\
for $\sigma \in [\DSnumber]$, these data structures $\calD$ and $\{\calD'_{\sigma}\}_{\sigma \in [\DSnumber]}$ can be non-identical.}

{\bf Output:}
A terminal solution $\Cterminal \in V^{k}$.
\begin{enumerate}
    \item \label{alg:LS:naive}
    Find an initial $n^{z + 1}$-approximate feasible solution $\Cinitial \in V^{k}$, based on \Cref{prop:naive_solution}.
    
    \item \label{alg:LS:initialize}
    $\calD.\Initialize(G, \Cinitial)$; likewise for $\{\calD'_{\sigma}\}_{\sigma \in [\DSnumber]}$.
    \hfill
    \Comment{Then $\calD'_{\sigma} \equiv \calD$.}

    \item \label{alg:LS:repeat}
    Repeat Lines~\ref{alg:LS:sample} to \ref{alg:LS:return}, until Line~\ref{alg:LS:return} returns a solution ``$\Cterminal \gets \calD.C$'':
    
    \item \label{alg:LS:sample}
    \qquad Sample a number of $\DSnumber$ noncenters, $\cinsert_{\sigma} \gets \calD.\SampleNoncenter()$ for $\sigma \in [\DSnumber]$.
    
    \item \label{alg:LS:schedule}
    \qquad Schedule a number of $\DSnumber$ subroutines, $\TestEffectiveness(\calD, \calD'_{\sigma}, \cinsert_{\sigma})$ for $\sigma \in [\DSnumber]$, \\
    \qquad step by step using the round-robin algorithm; during this process:
    
    \item \label{alg:LS:positive}
    \qquad\qquad If any one of the $\DSnumber$ subroutines, say $\sigma^{*} \in [\DSnumber]$, first returns a pair $(\cinsert_{\sigma^{*}}, \cdelete_{\sigma^{*}})$:
    
    \item \label{alg:LS:terminate}
    \qquad\qquad\qquad Terminate all (ongoing) subroutines.
    \hfill
    \Comment{Then $\calD'_{\sigma} \equiv \calD$.}
    
    \item \label{alg:LS:swap}
    \qquad\qquad\qquad $\calD.\Insert(\cinsert_{\sigma^{*}})$ and $\calD.\Delete(\cdelete_{\sigma^{*}})$; likewise for $\{\calD'_{\sigma}\}_{\sigma \in [\DSnumber]}$.
    \hfill
    \Comment{Then $\calD'_{\sigma} \equiv \calD$.}
    
    \item \label{alg:LS:negative}
    \qquad\qquad Else, all of the $\DSnumber$ subroutines (terminate and) return {\failure}'s:

    \item \label{alg:LS:return}
    \qquad\qquad\qquad Return $\Cterminal \gets \calD.C$.
\end{enumerate}
\end{flushleft}
\end{mdframed}

\begin{mdframed}
Subroutine $\term[\TestEffectiveness]{alg:pair}(\calD, \calD', \cinsert)$

\begin{flushleft}
{\bf Input:}
Two identical data structures $\calD \equiv \calD'$ and a noncenter $\cinsert \notin \calD.C = \calD'.C$ (promised).

\Comment{Only the second data structure $\calD'$ will be modified (in Line~\ref{alg:test:insert}); when this subroutine terminates (because of Line~\ref{alg:LS:terminate}, Line~\ref{alg:test:return-pair}, or Line~\ref{alg:test:return-failure}), $\calD'$ backtracks such that $\calD \equiv \calD'$ again.}

{\bf Output:}
A pair $(\cinsert, \cdelete)$ or a ``{\failure}''.
\begin{enumerate}
\setcounter{enumi}{10}
    \item \label{alg:test:insert}
    $\calD'.\Insert(\cinsert)$.
    \hfill
    \Comment{Then $\calD' \equiv \calD.\Insert(\cinsert)$.}
    
    \item \label{alg:test:for}
    For every center group $\tau \in [\CGnumber]$:
    \hfill
    \Comment{$\CGnumber = \lceil \log_{2}(2m|\calJ|) + 1\rceil$.}
    
    \item \label{alg:test:cdelete}
    \qquad $\cdelete_{\tau} \gets \argmin_{c \in \calD'.G_{\tau} - \cinsert} \calD'.\DSloss[c]$.
    \hfill
    \Comment{Recall Invariant~\ref{invar:grouping:group} for $\calD'.G_{\tau}$.}

    \item \label{alg:test:if}
    \qquad If $\frac{\calD'.\DScost + \calD'.\DSloss[\cdelete_{\tau}]}{\calD.\DScost}
    \le 1 - (\eps / 2) \cdot \calD'.\DSvolume[\cdelete_{\tau}]$:
    
    \item \label{alg:test:return-pair}
    \qquad\qquad Return $(\cinsert, \cdelete) \gets (\cinsert, \cdelete_{\tau})$.
    
    \item \label{alg:test:return-failure}
    Return a {\failure}.
\end{enumerate}
\end{flushleft}
\end{mdframed}
\caption{\label{fig:LocalSearch}The algorithm {\LocalSearch} and the subroutine {\TestEffectiveness}.}
\end{figure}
\clearpage}

\begin{assumption}[Hop-boundedness]
\label{assumption:LS:hop-bounded}
\begin{flushleft}
The underlying graph $G = (V, E, w)$ is $(\beta, \eps)$-hop-bounded, 
with known parameters $\beta \in [n]$ and $0 < \eps < \frac{1}{10z}$.
\end{flushleft}
\end{assumption}

\noindent
Since \Cref{assumption:LS:hop-bounded} (with $0 < \eps < \frac{1}{10z}$) is more restricted than \Cref{assumption:DS:hop-bounded} (with $0 < \eps < 1$), everything about our data structure $\calD$ established in \Cref{sec:DS} still holds.

The following \Cref{thm:LS} gives the performance guarantees of our algorithm {\LocalSearch}.

\begin{theorem}[{\LocalSearch}]
\label{thm:LS}
\begin{flushleft}
Provided \Cref{assumption:LS:hop-bounded,assumption:edge-weight}, the randomized algorithm {\LocalSearch} has worst-case running time $\eps^{-O(z)} m \beta \log^{5}(n)$ and, with probability $\ge 1 - n^{-\Theta(1)}$,\footnote{\label{footnote:exponent}Namely, (the absolute value of) the exponent $\Theta(1)$ can be an arbitrarily large but given constant.} returns an $\alpha_{z}(\eps)$-approximate feasible solution $\Cterminal \in V^{k}$ to the {\kzC} problem.\footnote{\label{footnote:alpha}The only constraint on (the feasible space of) this minimization problem is that {\em the objective must be nonnegative}:
$3 - \eps - 2^{1 + 2\eps z} \cdot ((1 + 1 / \lambda)^{z - 1} + \eps) > 0
\iff \heartsuit \eqdef (3 - \eps) \cdot 2^{-(1 + 2\eps z)} - \eps > (1 + 1 / \lambda)^{z - 1}
\iff \lambda > (\heartsuit^{1 / (z - 1)} - 1)^{-1}$.
This feasible space must be nonempty since $\heartsuit \ge (3 - \frac{1}{10}) \cdot 2^{-6/5} - \frac{1}{10} \approx 1.1623 > 1 \impliedby (z \ge 1) \wedge (0 < \eps < \frac{1}{10z})$.}
\begin{align*}
    \alpha_{z}(\eps)
    ~\eqdef~ \min_{\lambda} \bigg\{\,\frac{2^{1 + z + 2\eps z} \cdot (1 + \lambda)^{z - 1} + 2^{(1 + 4\eps) z}}{3 - \eps - 2^{1 + 2\eps z} \cdot ((1 + 1 / \lambda)^{z - 1} + \eps)} \,\biggmid\, \lambda > \frac{1}{((3 - \eps) / 2^{1 + 2\eps z} - \eps)^{1 / (z - 1)} - 1}\,\bigg\}.
\end{align*}
\end{flushleft}
\end{theorem}

Generally, the algorithm {\LocalSearch} starts with an initial feasible solution $C = \Cinitial \in V^{k}$ (\Cref{prop:naive_solution}) and, based on our data structure $\calD$, iteratively conducts {\em $(\cinsert, \cdelete)$-swaps}. Namely, a single $(\cinsert, \cdelete)$-swap modifies the center set $C \gets C + \cinsert - \cdelete$ by replacing a current {\em noncenter} $\cinsert \notin C$ with a current {\em center} $\cdelete \in C$.
In this regard, we will introduce in \Cref{subsec:effective} two concepts, {\em super-effectiveness} of such a pair $(\cinsert, \cdelete)$ (\Cref{def:effective}) and {\em super-effectiveness} of a noncenter $\cinsert \notin C$ (\Cref{def:super-effective}); both concepts are crucial to all materials in subsequent subsections.
Then, we will establish two technical lemmas:

\vspace{.1in}
\noindent
\Cref{lem:test} in \Cref{subsec:test} shows that, on input a {\em super-effective} noncenter $\cinsert \notin C$ (if promised), the subroutine {\TestEffectiveness} always can find an {\em super-effective} pair $(\cinsert, \cdelete)$.

\noindent
\Cref{lem:sample} in \Cref{subsec:sample} shows that, by utilizing the operation {\SampleNoncenter} of our data structure $\calD$, we can sample a {\em super-effective} noncenter $\cinsert \notin C$ with ``sufficiently high'' probability.

\vspace{.1in}
\noindent
By combining both lemmas with additional arguments, we will finally accomplish the above \Cref{thm:LS} in \Cref{subsec:local-search}.

\subsection{The concepts of super-effective pairs \texorpdfstring{$(\cinsert, \cdelete)$}{} and noncenters \texorpdfstring{$\cinsert$}{}}
\label{subsec:effective}

This subsection introduces the concepts of {\em super-effective pairs} (\Cref{def:effective}) and {\em super-effective noncenters} (\Cref{def:super-effective}); all materials in subsequent subsections crucially rely on these concepts.

First of all, let us introduce the concept of {\em super-effective pairs}, as follows (\Cref{def:effective}).

\begin{definition}[Super-effective pairs]
\label{def:effective}
\begin{flushleft}
For a data structure $\calD$ given in \Cref{fig:DS}, a pair of noncenter and center $(\cinsert, \cdelete) \colon (\cinsert \notin \calD.C) \wedge (\cdelete \in \calD.C)$ determines:
\begin{itemize}
    \item The $\cinsert$-inserted data structure $\calD' \gets \calD.\Insert(\cinsert)$.
    \hfill
    \Comment{Cf.\ \Cref{fig:insert}.}
    
    \item The $(\cinsert, \cdelete)$-swapped data structure $\calD'' \gets \calD'.\Delete(\cdelete)$.
    \hfill
    \Comment{Cf.\ \Cref{fig:delete}.}
\end{itemize}
We call this pair $(\cinsert, \cdelete)$ an {\em super-effective pair} when
\begin{align*}
    \frac{\calD''.\DScost}{\calD.\DScost}
    ~\le~ 1 - (\eps / 2) \cdot \calD'.\DSvolume[\cdelete].
\end{align*}
\end{flushleft}
\end{definition}

\begin{remark}[Super-effective pairs]
Here is the intuition behind a super-effective pair $(\cinsert, \cdelete)$:
\begin{flushleft}
\begin{itemize}
    \item $\frac{\calD''.\DScost}{\calD.\DScost}$ well estimates ``the progress on minimizing the clustering objective $\cost(V, C)$ by the $(\cinsert, \cdelete)$-swap''.\\
    \Comment{$\cost(V, \calD.C) \le \calD.\DScost \le 2^{2\eps z} \cdot \cost(V, \calD.C)$; likewise for $\calD''.\DScost$ (\Cref{lem:DScost}).}
    
    \item $\calD'.\DSvolume[\cdelete]$ (up to scale) well estimates ``the running time of the $(\cinsert, \cdelete)$-swap''.\\
    \Comment{Rigorously, (\Cref{lem:delete:runtime} of \Cref{lem:delete}) it well estimates ``the running time of the $\cdelete$-deletion'', i.e., $\Tdelete = \calD'.\DSvolume[\cdelete] \cdot O(m \log^{2}(n))$. Nonetheless, for the moment, do not worry about the running time of the $\cinsert$-insertion.}
\end{itemize}
\end{flushleft}
Hence, this particular super-effective pair $(\cinsert, \cdelete)$ is ``$\ge \eps / 2$-competitive'' with other pairs, making $\ge \eps / 2$ unit of progress per unit of running time.
(The bound $\ge \eps / 2$ {\em is} sufficient for our purpose.)
\end{remark}

However, the super-effectiveness condition for swap pairs is technically difficult to use directly -- it involves three data structures, the given one $\calD$, the $\cinsert$-inserted one $\calD'$, and the $(\cinsert, \cdelete)$-swapped one $\calD''$.
Instead, we would introduce the concept of {\em super-effective noncenters} (\Cref{def:super-effective}), which will be technically more tractable surrogate of {\em super-effective pair}.

\begin{definition}[Super-effective noncenters]
\label{def:super-effective}
\begin{flushleft}
For a data structure $\calD$ given in \Cref{fig:DS}, every pair of noncenter and center $\forall (p, q) \colon (p \notin \calD.C) \wedge (q \in \calD.C)$ determines the following {\em $(p, q)$-swap distoid} $\calD.\DSd_{p,\, q} \colon V \mapsto [0, +\infty)$.\footnote{\label{footnote:no-need-to-maintain}This $\calD.\DSd_{p,\, q}$ just helps with our analysis and need not be actually maintained. Namely, (in spirit) it upper-bounds the approximate distance $\calD''.\DSd[v]$, $\forall v \in V$, maintained by the modified data structure $\calD''$.}
\begin{align*}
    \calD.\DSd_{p,\, q}(v) ~\eqdef~
    \begin{cases}
        2^{2\eps} \cdot \dist(v, \calD.C + p - q)
        & \calD.\DSc[v] = q \\
        \min\big(\calD.\DSd[v],\ 2^{2\eps} \cdot \dist(v, p)\big)
        & \calD.\DSc[v] \ne q
    \end{cases},
    && \forall v \in V.
\end{align*}
A noncenter $\cinsert \notin \calD.C$ is called {\em super-effective} noncenter when
\begin{align*}
    \exists q \in \calD.C
    \colon\qquad
    \frac{\sum_{v \in V} \calD.\DSd_{\cinsert,\, q}^{z}(v)}{\calD.\DScost}
    ~\le~ 1 - \eps \cdot \calD.\DSvolume[q],
\end{align*}
\end{flushleft}
\end{definition}

We emphasize that the super-effectiveness condition for noncenters only involves one data structure $\calD$ (thus independent of its $\cinsert$-inserted and $(\cinsert, \cdelete)$-swapped counterparts $\calD'$ and $\calD''$).
This accounts for why this condition is technically more tractable than super-effectiveness for swap pairs.

Regarding the subroutine {\TestEffectiveness} (\Cref{fig:LocalSearch}) and the concepts of super-effective pairs and super-effective noncenters (\Cref{def:effective,def:super-effective}), we will establish the following \Cref{lem:test,lem:sample} in \Cref{subsec:test,subsec:sample}, respectively.
(Recall that  of the {\kzC} problem.)

\begin{lemma}[{\TestEffectiveness}]
\label{lem:test}
\begin{flushleft}
For the subroutine $\TestEffectiveness(\calD, \calD', \cinsert)$:
\begin{enumerate}[font = {\em\bfseries}]
    \item \label{lem:test:1}
    Any possible pair $(\cinsert, \cdelete)$ returned in Line~\ref{alg:test:return-pair} is super-effective.
    
    \item \label{lem:test:2}
    The subroutine will return an super-effective pair $(\cinsert, \cdelete)$, provided that the input noncenter $\cinsert \notin \calD.C$ (promised) is super-effective.
\end{enumerate}
\end{flushleft}
\end{lemma}

\begin{lemma}[Bounding the probability of super-effective sampling]
\label{lem:sample}
\begin{flushleft}
A single noncenter sampled in Line \ref{alg:LS:sample},  $\cinsert_{\sigma} \gets \calD.\SampleNoncenter()$ for $\sigma \in [\DSnumber]$, is super-effective with probability $\ge \eps^{4z}$, provided $\cost(V, \calD.C) \ge \alpha_{z}(\eps) \cdot \OPT$.
\hfill
\Comment{The optimal objective $\OPT = \min_{C \in V^{k}} \cost(V, C)$.}
\end{flushleft}
\end{lemma}

\subsection{A super-effective noncenter \texorpdfstring{$\cinsert$}{} ensures a super-effective pair \texorpdfstring{$(\cinsert, \cdelete)$}{}}
\label{subsec:test} 

This subsection completes the proof of \Cref{lem:test} (which is restated below for ease of reference).
The reader can reference \Cref{fig:LocalSearch} for the implementation of the subroutine {\TestEffectiveness}.

\begin{restate}[\Cref{lem:test}]
\begin{flushleft}
For the subroutine $\TestEffectiveness(\calD, \calD', \cinsert)$:
\begin{enumerate}[font = {\em\bfseries}]
    \item
    Any possible pair $(\cinsert, \cdelete)$ returned in Line~\ref{alg:test:return-pair} is super-effective.
    
    \item
    The subroutine will return an super-effective pair $(\cinsert, \cdelete)$, provided that the input noncenter $\cinsert \notin \calD.C$ (promised) is super-effective.
\end{enumerate}
\end{flushleft}
\end{restate}

\begin{proof}
The subroutine modifies the second input data structure $\calD'.\Insert(\cinsert)$ at first (Line~\ref{alg:test:insert}); given this, throughout this proof, we would use the notation $\calD'$ to actually denote the $\cinsert$-inserted data structure $\calD' \gets \calD'.\Insert(\cinsert)$.
We thus have $\calD' \equiv \calD.\Insert(\cinsert)$ and $\calD'.C = \calD.C + \cinsert$, since the two input data structures are identical (as Lines~\ref{alg:LS:initialize}, \ref{alg:LS:terminate}, and \ref{alg:LS:swap} promise).
These notations are more consistent with the notations in \Cref{sec:DS}, facilitating references to our results therein.

\vspace{.1in}
\noindent
{\bf \Cref{lem:test:1}.}
That any possible pair $(\cinsert, \cdelete) \gets (\cinsert, \cdelete_{\tau})$ returned in Line~\ref{alg:test:return-pair} is an super-effective pair is a direct consequence of a combination of \Cref{lem:delete:loss} of \Cref{lem:delete} and the condition in Line~\ref{alg:test:if}:
\begin{align*}
    \calD''.\DScost
    & ~\le~ \calD'.\DScost + \calD'.\DSloss[\cdelete_{\tau}]
    \tag{\Cref{lem:delete:loss} of \Cref{lem:delete}} \\
    & ~\le~ \big(1 - (\eps / 2) \cdot \calD'.\DSvolume[\cdelete_{\tau}]\big) \cdot \calD.\DScost.
    \tag{Condition in Line~\ref{alg:test:if}}
\end{align*}
This is exactly the defining condition of an super-effective pair (\Cref{def:effective}).

\vspace{.1in}
\noindent
{\bf \Cref{lem:test:2}.}
It suffices to show that the subroutine will return a pair $(\cinsert, \cdelete)$ in Line~\ref{alg:test:return-pair} (which must be an super-effective pair (\Cref{lem:test:1})), provided that the input noncenter $\cinsert \notin \calD.C$ is super-effective.
Consider a specific center $q \in \calD.C = \calD'.C - \cinsert$ that satisfies this super-effectiveness (\Cref{def:super-effective}):
\begin{align}
\label{eq:lem:test:1}
    \sum_{v \in V} \calD.\DSd_{\cinsert,\, q}^{z}(v)
    ~\le~ \big(1 - \eps \cdot \calD.\DSvolume[q]\big) \cdot \calD.\DScost.
\end{align}
Also, assume the following \Cref{eq:lem:test:2} for the moment, which we will prove later.
\begin{align}
    \label{eq:lem:test:2}
    \calD'.\DScost + \calD'.\DSloss[q]
    ~\le~ \sum_{v \in V} \calD.\DSd_{\cinsert,\, q}^{z}(v).
\end{align}
We index by $\tau_{q} \eqdef \lfloor -\log_{2}(\calD'.\DSvolume[q]) + 1\rfloor \in [\CGnumber]$ the center group $\calD'.G_{\tau_{q}}$ that includes the considered center $q \in \calD.C = \calD'.C - \cinsert$ (Invariant~\ref{invar:grouping:group}).\footnote{Recall that $\frac{1}{2m |\calJ|} \le \calD'.\DSvolume[q] \le 1$ (\Cref{lem:deletion-estimator:DSvolume} of \Cref{lem:deletion-estimator}) and $\CGnumber = \lceil \log_2(2m|\calJ|) + 1\rceil$ (Invariant~\ref{invar:grouping:group}). Also, all {\em disjoint} center groups $\{\calD'.G_{\tau}\}_{\tau \in [\CGnumber]}$ together cover the center set $\calD'.C = \cup_{\tau \in [\CGnumber]} \calD'.G_{\tau}$ (Invariant~\ref{invar:grouping:group} and \Cref{lem:grouping}).}
Without loss of generality, we can assume that the subroutine will arrive this center group $\tau_{q} \in [\CGnumber]$ (Line~\ref{alg:test:for});
otherwise, the subroutine must have already returned an (super-effective) pair $(\cinsert, \cdelete) \gets (\cinsert, \cdelete_{\tau})$ for an earlier group $\tau < \tau_{q}$ (Line~\ref{alg:test:return-pair}).

For this center group $\tau_{q} \in [\CGnumber]$, the subroutine will identify the ``$\cinsert$-excluded $\calD'.\DSloss[c]$-minimizer'' $\cdelete_{\tau_{q}} \eqdef \argmin_{c \in \calD'.G_{\tau_{q}} - \cinsert} \calD'.\DSloss[c]$ (Line~\ref{alg:test:cdelete}),\footnote{This center $\cdelete_{\tau_{q}} \in \calD'.G_{\tau_{q}} - \cinsert$ is well defined, since $\calD'.G_{\tau_{q}} - \cinsert \supseteq \{q\} \ne \emptyset \impliedby (\calD'.C - \cinsert \ni q) \wedge (\calD'.G_{\tau_{q}} \ni q)$.}
then test the condition in Line~\ref{alg:test:if} for this center $\cdelete_{\tau_{q}} \in \calD'.G_{\tau_{q}} - \cinsert$, and
finally return the pair $(\cinsert, \cdelete) \gets (\cinsert, \cdelete_{\tau_{q}})$ if the test passes (Line~\ref{alg:test:return-pair}).
Thus, it suffices to check the condition in Line~\ref{alg:test:if} for this center $\cdelete_{\tau_{q}} \in \calD'.G_{\tau_{q}} - \cinsert$:
\begin{align*}
    \calD'.\DScost + \calD'.\DSloss[\cdelete_{\tau_{q}}]
    & ~\le~ \calD'.\DScost + \calD'.\DSloss[q]
    \tag{Definition of $\cdelete_{\tau_{q}}$} \\
    & ~\le~ \big(1 - \eps \cdot \calD.\DSvolume[q]\big) \cdot \calD.\DScost
    \tag{\Cref{eq:lem:test:1,eq:lem:test:2}} \\
    & ~\le~ \big(1 - \eps \cdot \calD'.\DSvolume[q]\big) \cdot \calD.\DScost
    \tag{\Cref{lem:insert:DSvolume} of \Cref{lem:insert}} \\
    & ~\le~ \big(1 - (\eps / 2) \cdot \calD'.\DSvolume[\cdelete_{\tau_{q}}]\big) \cdot \calD.\DScost.
\end{align*}
Here, the last step applies $\calD'.\DSvolume[\cdelete_{\tau_{q}}] \le 2\calD'.\DSvolume[q]$, i.e., both centers $\cdelete_{\tau_{q}}$ and $q$ belong to the center subset $\calD'.G_{\tau_{q}} - \cinsert$, and all centers $\in \calD'.G_{\tau_{q}} = \{c \in \calD'.C \mid \lfloor-\log_{2}(\calD'.\DSvolume[c]) + 1\rfloor = \tau_{q}\}$ (Invariant~\ref{invar:grouping:group}) differ in their $\calD'.\DSvolume[c]$ values by a factor of $\le 2$.

It remains to verify \Cref{eq:lem:test:2}.
By Invariant~\ref{invar:DScost} and \Cref{lem:deletion-estimator:DSloss} of \Cref{lem:deletion-estimator}, we can upper-bound $\LHS~\text{of}~\eqref{eq:lem:test:2}
= \calD'.\DScost + \calD'.\DSloss[q]$ as follows:
\begin{align*}
    \LHS~\text{of}~\eqref{eq:lem:test:2}
    & ~\le~ \sum_{v \in V} \calD'.\DSd^{z}[v]
    + \sum_{v \in V \colon \calD'.\DSc[v] = q} \big(2^{2\eps z} \cdot \dist^{z}(v, \calD'.C - q) - \calD'.\DSd^{z}[v]\big) \\
    & ~=~ \sum_{v \in V \colon \calD'.\DSc[v] \ne q} \calD'.\DSd^{z}[v] + \sum_{v \in V \colon \calD'.\DSc[v] = q} 2^{2\eps z} \cdot \dist^{z}(v, \calD'.C - q).
\end{align*}
It suffices to prove that this formula
$\le \RHS~\text{of}~\eqref{eq:lem:test:2}
= \sum_{v \in V} \calD.\DSd_{\cinsert,\, q}^{z}(v)$, and we would address every vertex-$(v \in V)$ term in either case $\{\calD'.\DSc[v] \ne q\}$ or $\{\calD'.\DSc[v] = q\}$ separately.

\vspace{.1in}
\noindent
{\bf Case~1: $\calD'.\DSc[v] \ne q$.}
Following a combination of \Cref{lem:clustering:vertex} of \Cref{lem:clustering} and \Cref{lem:insert:DSd} of \Cref{lem:insert},
\begin{align*}
    \calD'.\DSd[v]
    & ~\le~ \min\big(\calD.\DSd[v],\ 2^{2\eps} \cdot \dist(v, \calD'.C)\big) \\
    & ~=~ \min\big(\calD.\DSd[v],\ 2^{2\eps} \cdot \dist(v, \calD.C + \cinsert)\big) \\
    & ~\le~ \min\big(\calD.\DSd[v], 2^{2\eps} \cdot \dist(v, \calD.C + \cinsert - q)\big) \\
    & ~\le~ \DSd_{\cinsert,\, q}(v)
    ~=~
    \begin{cases}
        2^{2\eps} \cdot \dist(v, \calD.C + \cinsert - q)
        & \calD.\DSc[v] = q \\
        \min\big(\calD.\DSd[v],\ 2^{2\eps} \cdot \dist(v, \cinsert)\big)
        & \calD.\DSc[v] \ne q
    \end{cases}.
\end{align*}
Here, the second step applies $\calD'.C = \calD.C + \cinsert$, and the last step applies $\cinsert \notin \calD.C$ and $q \in \calD.C$ and restates (\Cref{def:effective}) the defining formula of the function $\DSd_{\cinsert,\, q}(v)$.

\vspace{.1in}
\noindent
{\bf Case~2: $\calD'.\DSc[v] = q$.}
In this case, we must have $\cinsert \ne \calD'.\DSc[v] = q \in \calD.C \impliedby \cinsert \notin \calD.C$, which makes \Cref{lem:insert:DSc} of \Cref{lem:insert} applicable $\implies \calD.\DSc[v] = \calD'.\DSc[v] = q$.
We thus have
\begin{align*}
    2^{2\eps} \cdot \dist(v, \calD'.C - q)
    ~=~ 2^{2\eps} \cdot \dist(v, \calD.C + \cinsert - q)
    ~=~ \calD.\DSd_{\cinsert,\, q}(v).
\end{align*}
Combining both cases gives \Cref{eq:lem:test:2}.

This finishes the proof of \Cref{lem:test}.
\end{proof}

\subsection{Bounding the probability of super-effective sampling}
\label{subsec:sample}

\newcommand{\rmloss}{\mathrm{loss}_{z}}
\newcommand{\rmgain}{\mathrm{gain}_{z}}

This subsection completes the proof of \Cref{lem:sample} (which is restated below for ease of reference).
Essentially, this is a direct consequence of the following \Cref{lem:noncenter-subset}.

\begin{restate}[{\Cref{lem:sample}}]
\begin{flushleft}
A single noncenter sampled in Line \ref{alg:LS:sample},  $\cinsert_{\sigma} \gets \calD.\SampleNoncenter()$ for $\sigma \in [\DSnumber]$, is super-effective with probability $\ge \eps^{4z}$, provided $\cost(V, \calD.C) \ge \alpha_{z}(\eps) \cdot \OPT$.
\end{flushleft}
\end{restate}

\newcommand{\BarC}{S}

\begin{lemma}[A ``costly'' super-effective noncenter subset]
\label{lem:noncenter-subset}
\begin{flushleft}
The premise of \Cref{lem:sample}, namely $\cost(V, \calD.C) \ge \alpha_{z}(\eps) \cdot \OPT$, ensures that there exists such a noncenter subset $\BarC_{\calD} \subseteq V \setminus \calD.C$:
\begin{enumerate}[font = {\em\bfseries}]
    \item \label{lem:noncenter-subset:1}
    Every noncenter $\cinsert \in \BarC_{\calD} \subseteq V \setminus \calD.C$ is super-effective.
    
    \item \label{lem:noncenter-subset:2}
    $\cost(\BarC_{\calD}, \calD.C) \ge \eps^{z + 2} \cdot \cost(V, \calD.C)$.
\end{enumerate}
\end{flushleft}
\end{lemma}

\begin{proof}[Proof of \Cref{lem:sample} (Assuming \Cref{lem:noncenter-subset})]
\Cref{lem:sample} stems from a combination of \Cref{lem:noncenter-subset} and the performance guarantees of our data structure $\calD$ shown in \Cref{sec:DS}, as follows:
\begin{align*}
    \Pr_{\cinsert_{i}}\big[\, \text{$\cinsert_{i}$ is super-effective} \,\big]
    & ~\ge~ \Pr_{\cinsert_{i}}\big[\, \cinsert_{i} \in \BarC_{\calD} \,\big]
    \tag{\Cref{lem:noncenter-subset:1} of \Cref{lem:noncenter-subset}} \\
    & ~=~ \frac{\sum_{v \in \BarC_{\calD}} \calD.\DSd^{z}[v]}{\sum_{v \in V} \calD.\DSd^{z}[v]}
    \tag{\Cref{lem:sample-noncenter}} \\
& ~\ge~ 2^{-2\eps z} \cdot \frac{\cost(\BarC_{\calD}, \calD.C)}{\cost(V, \calD.C)}
    \tag{\Cref{lem:clustering:vertex} of \Cref{lem:clustering}} \\
    & ~\ge~ 2^{-2\eps z} \cdot \eps^{z + 2}
    \tag{\Cref{lem:noncenter-subset:2} of \Cref{lem:noncenter-subset}} \\
    & ~\ge~ \eps^{4z}.
    \tag{$z \ge 1$ and $0 < \eps < \frac{1}{10z}$}
\end{align*}
This finishes the proof of \Cref{lem:sample}.
\end{proof}

In the remainder of this subsection, we explicitly construct a ``costly'' super-effective noncenter subset $\BarC_{\calD} \subseteq V \setminus \calD.C$ (\Cref{def:noncenter-subset}) and verify \Cref{lem:noncenter-subset:1,lem:noncenter-subset:2} of \Cref{lem:noncenter-subset} for this $\BarC_{\calD}$.\footnote{The construction of our noncenter subset $\BarC$ and the proof of \Cref{lem:noncenter-subset} are refinements of \cite{LattanziS19}.}
Since we are considering on a single (unmodified) data structure $\calD$, without ambiguity, we simplify the notations by writing $C = \calD.C$, $\BarC = \BarC_{\calD}$, etc.

\begin{definition}[A ``costly'' super-effective noncenter subset $\BarC$]
\label{def:noncenter-subset}
\begin{flushleft}
Our construction of $\BarC \subseteq V \setminus C$ relies on several additional concepts:
\begin{itemize}
    \item Enumerate and index the {\em maintained centers} $C = \{c_{i}\}_{i \in [k]}$ in the maintained solution and the corresponding {\em maintained clusters} $\cluster_{i} \eqdef \{v \in V \mid \calD.\DSc[v] = c_{i}\}$, for $i \in [k]$.
    
    \item Enumerate and index the {\em optimal centers} $C^{*} = \{c_{i}^{*}\}_{i \in [k]}$ in the optimal solution and the corresponding {\em optimal clusters} $\cluster_{i}^{*} \eqdef \{v \in V \mid \argmin_{c^{*} \in C^{*}} \dist(v, c^{*}) = c_{i}^{*}\}$, $\forall i \in [k]$.
    
    \item We say an optimal center $c_{i}^{*} \in C^{*}$ is {\em captured} by its nearest maintained center $c = \argmin_{c' \in C} \dist(c^{*}, c')$.
    We can classify all maintained centers $c_{i} \in C$ into three kinds:
    
    (i)~A maintained center $c_{i} \in C$ that captures a {\em unique} optimal center is called {\em matched}; \\
    we denote by $C_{M} \subseteq C$ all {\em matched} maintained centers and by $M \subseteq [k]$ their indices and (without loss of generality) reindex $C$ and $C^{*}$ such that every {\em matched} maintained center $c_{i} \in C_{M}$ and its {\em unique} captured optimal center $c_{i}^{*} \in C^{*}$ have the same index $i \in M$.
    
    All other maintained centers $C_{\Bar{M}} \eqdef C \setminus C_{M}$ and their indices $\Bar{M} \eqdef [k] \setminus M$ can further be classified into two kinds:
    
    (ii)~A maintained center $c_{i} \in C$ that captures {\em no} optimal center is called {\em lonely}; \\
    we denote by $C_{L} \subseteq C_{\Bar{M}} \subseteq C$ all {\em lonely} maintained centers and by $L \subseteq \Bar{M} \subseteq [k]$ their indices.
    
    (iii)~Every other maintained centers $c_{i} \notin C_{M} \cup C_{L}$, which is {\em neither matched nor lonely}, captures {\em two or more} optimal centers.
\end{itemize}
Then, based on these notations, we can define a function $\rmloss(c_{i})$ on all maintained center $c_{i} \in C$ and a function $\rmgain(c_{i}^{*})$ on all optimal center $c_{i}^{*} \in C^{*}$, as follows;\footnote{Our construction of $\BarC$ does not rely on the function values $\rmloss(c_{i})$ on the {\em neither-matched-nor-lonely} maintained centers $c_{i} \notin C_{M} \cup C_{L}$, which thus can be defined arbitrarily -- we simply reuse the definition on the {\em lonely} maintained centers $c_{i} \in C_{L}$.}
we observe that $\rmloss(c_{i}) \ge 0$, for every maintained center $c_{i} \in C$ (of any kind).\footnote{Namely, $\rmloss(c_{i}) \ge \cost(\cluster_{i} \setminus \cluster_{i}^{*}, C - c_{i}) - \cost(\cluster_{i} \setminus \cluster_{i}^{*}, C) \ge 0$ in case of a {\em matched} maintained center $c_{i} \in C_{M}$, and $\rmloss(c_{i}) \ge \cost(\cluster_{i}, C - c_{i}) - \cost(\cluster_{i}, C) \ge 0$ in the opposite case.}
\begin{align*}
    \rmloss(c_{i})
    & ~\eqdef~
    \begin{cases}
        2^{2\eps z} \cdot \cost(\cluster_{i} \setminus \cluster_{i}^{*}, C - c_{i}) - \cost(\cluster_{i} \setminus \cluster_{i}^{*}, C)
        &\qquad \forall i \in M \\
        2^{2\eps z} \cdot \cost(\cluster_{i}, C - c_{i}) - \cost(\cluster_{i}, C)
        &\qquad \forall i \in L \\
        2^{2\eps z} \cdot \cost(\cluster_{i}, C - c_{i}) - \cost(\cluster_{i}, C)
        &\qquad \forall i \notin M \cup L
    \end{cases}, \\
    \rmgain(c_{i}^{*})
    & ~\eqdef~ 2^{(1 + 4\eps) z} \cdot \cost(\cluster_{i}^{*}, \{c_{i}^{*}\}) - \cost(\cluster_{i}^{*}, C),
    \hspace{2.52cm} \forall i \in [k].
\end{align*}
Now we define the noncenter subset $\BarC \subseteq V \setminus C$, based on the {\em noncenter subclusters} $\BarC_{i} \subseteq \cluster_{i}^{*} \setminus C$, for $i \in [k]$, the {\em matched index subset} $M' \subseteq M$, and the {\em unmatched index subset} $\Bar{M}' \subseteq \Bar{M}$.
(Recall Invariant~\ref{invar:DSvolume} for the deletion volume estimator $\DSvolume[c_{i}]$.)
\begin{align*}
    \BarC
    & ~\eqdef~ \cup_{i \in M' \cup \Bar{M}'} \BarC_{i}, \\
    \BarC_{i}
    & ~\eqdef~ \big\{v \in \cluster_{i}^{*} \setminus C \bigmid \dist^{z}(v, c_{i}^{*}) \le 2^{\eps z} \cdot \tfrac{1}{|\cluster_{i}^{*}|} \cdot \cost(\cluster_{i}^{*}, \{c_{i}^{*}\})\big\},
    \qquad\qquad \forall i \in [k], \\
    M' & ~\eqdef~ \big\{i \in M \bigmid \rmloss(c_{i}) + \rmgain(c_{i}^{*}) \le -2^{2\eps z} \cdot \eps \cdot \DSvolume[c_{i}] \cdot \cost(V, C)\big\}, \\
    \Bar{M}' & ~\eqdef~ \big\{i \in \Bar{M} \bigmid \exists j \in L \colon \rmloss(c_{j}) + \rmgain(c_{i}^{*}) \le -2^{2\eps z} \cdot \eps \cdot \DSvolume[c_{j}] \cdot \cost(V, C)\big\}.
\end{align*}
\end{flushleft}
\end{definition}

The remainder of this subsection is devoted to establishing \Cref{lem:noncenter-subset} for our noncenter subset $\BarC \subseteq V \setminus C$ from \Cref{def:noncenter-subset}; we begin with its \Cref{lem:noncenter-subset:1} (which is rephrased for ease of reference).

\begin{restate}[\Cref{lem:noncenter-subset}, \Cref{lem:noncenter-subset:1}]
Every noncenter $\cinsert \in \BarC \subseteq V \setminus C$ is super-effective, namely
\begin{align}
\label{eq:P}
    \exists q \in C
    \colon\qquad
    \sum_{v \in V} \DSd_{\cinsert,\, q}^{z}(v)
    ~\le~ (1 - \eps \cdot \DSvolume[q]) \sum_{v \in V} \DSd^{z}[v].
\end{align}
\end{restate}

\begin{proof}
Regarding our construction of $\BarC \subseteq V \setminus C$ (\Cref{def:noncenter-subset}), the considered noncenter $\cinsert \in \BarC$ locates in a unique index-$(i \in M' \cup \Bar{M}')$ noncenter subcluster $\BarC_{i} \ni \cinsert$.
We investigate either case, \{the index is matched $M' \ni i$\} versus \{the index is unmatched $\Bar{M}' \ni i$\}, separately.

\vspace{.1in}
\noindent
{\bf Case~1: the index is matched $M' \ni i$.}
The definition of $M'$ (\Cref{def:noncenter-subset}) ensures that
\begin{align*}
    \rmloss(c_{i}) + \rmgain(c_{i}^{*})
    ~\le~ -2^{2\eps z} \cdot \eps \cdot \DSvolume[c_{i}] \cdot \cost(V, C).
\end{align*}
\Cref{eq:P} turns out to hold simply for the index-$i$ maintained center $q = c_{i}$; we would decompose the $\LHS$ of \Cref{eq:P} into three parts and upper-bound them one by one, as follows.

\noindent
Firstly, the vertices $v \notin \cluster_{i} \cup \cluster_{i}^{*}$ (\Cref{def:super-effective}) each satisfy $\DSd_{\cinsert,\, c_{i}}^{z}(v) \le \DSd^{z}[v]$ and thus
\begin{align*}
    \sum_{v \notin \cluster_{i} \cup \cluster_{i}^{*}} \DSd_{\cinsert,\, c_{i}}^{z}(v)
    ~\le~ \sum_{v \notin \cluster_{i} \cup \cluster_{i}^{*}} \DSd^{z}[v].
    \hspace{7.01cm}
\end{align*}

\noindent
Secondly, the vertices $v \in \cluster_{i} \setminus \cluster_{i}^{*}$ (\Cref{def:super-effective}) each satisfy $\DSd_{\cinsert,\, c_{i}}(v) \le 2^{2\eps} \cdot \dist(v, C - c_{i})$ and
\begin{align*}
    \sum_{v \in \cluster_{i} \setminus \cluster_{i}^{*}} \DSd_{\cinsert,\, c_{i}}^{z}(v)
    ~\le~ 2^{2\eps z} \cdot \cost(\cluster_{i} \setminus \cluster_{i}^{*}, C - c_{i}).
    \hspace{4.68cm}
\end{align*}

\noindent
Thirdly, the vertices $v \in \cluster_{i}^{*}$ (\Cref{def:super-effective}) each satisfy $\DSd_{\cinsert,\, c_{i}}(v) \le 2^{2\eps} \cdot \dist(v, \cinsert)$ and thus
\begin{align*}
    \sum_{v \in \cluster_{i}^{*}} \DSd_{\cinsert,\, c_{i}}^{z}(v)
    & ~\le~ 2^{2\eps z} \cdot \cost(\cluster_{i}^{*}, \{\cinsert\}) \\
    & ~\le~ 2^{2\eps z} \cdot \big((1 + 2^{\eps})^{z - 1} \cdot \cost(\cluster_{i}^{*}, \{c_{i}^{*}\})\big. \\
    & \phantom{~=~}\qquad\qquad \big.+ (1 + 2^{-\eps})^{z - 1} \cdot |\cluster_{i}^{*}| \cdot \dist^{z}(c_{i}^{*}, \cinsert)\big) \\
    & ~\le~ 2^{2\eps z} \cdot \big((1 + 2^{\eps})^{z - 1}
    + (1 + 2^{-\eps})^{z - 1} \cdot 2^{\eps z}\big) \cdot \cost(\cluster_{i}^{*}, \{c_{i}^{*}\}) \\
    & ~=~ 2^{2\eps z} \cdot (1 + 2^{\eps})^{z} \cdot \cost(\cluster_{i}^{*}, \{c_{i}^{*}\}) \\
    & ~\le~ 2^{(1 + 4\eps) z} \cdot \cost(\cluster_{i}^{*}, \{c_{i}^{*}\}) \\
    & ~=~ \cost(\cluster_{i}^{*}, C)
    + \rmgain(c_{i}^{*}).
\end{align*}
Here, the second step applies \Cref{lem:triangle}, by setting the parameter $\lambda = 2^{\eps}$.
The third step uses the defining condition of the noncenter subcluster $\BarC_{i} \ni \cinsert$ (\Cref{def:noncenter-subset}), namely $\dist^{z}(\cinsert, c_{i}^{*}) \le 2^{\eps z} \cdot \tfrac{1}{|\cluster_{i}^{*}|} \cdot \cost(\cluster_{i}^{*}, \{c_{i}^{*}\})$.
The last step applies the formula of $\rmgain(c_{i}^{*})$ (\Cref{def:noncenter-subset}).

\noindent
A combination of the above three equations gives \Cref{eq:P}, for $q = c_{i}$, as follows.
\begin{align*}
    \LHS~\text{of}~\eqref{eq:P}
    & ~\le~ \sum_{v \notin \cluster_{i} \cup \cluster_{i}^{*}} \DSd^{z}[v]
    + 2^{2\eps z} \cdot \cost(\cluster_{i} \setminus \cluster_{i}^{*}, C - c_{i})
    + \cost(\cluster_{i}^{*}, C)
    + \rmgain(c_{i}^{*})
    \hspace{.09cm} \\
    & ~=~ \sum_{v \notin \cluster_{i} \cup \cluster_{i}^{*}} \DSd^{z}[v]
    + \cost(\cluster_{i} \cup \cluster_{i}^{*}, C)
    + \rmloss(c_{i}) + \rmgain(c_{i}^{*}) \\
    & ~\le~ \sum_{v \notin \cluster_{i} \cup \cluster_{i}^{*}} \DSd^{z}[v]
    + \cost(\cluster_{i} \cup \cluster_{i}^{*}, C)
    - 2^{2\eps z} \cdot \eps \cdot \DSvolume[c_{i}] \cdot \cost(V, C) \\
    & ~\le~ \sum_{v \notin \cluster_{i} \cup \cluster_{i}^{*}} \DSd^{z}[v]
    + \sum_{v \in \cluster_{i} \cup \cluster_{i}^{*}} \DSd^{z}[v]
    - \eps \cdot \DSvolume[c_{i}] \sum_{v \in V} \DSd^{z}[v] \\
    & ~=~ \RHS~\text{of}~\eqref{eq:P}.
\end{align*}
Here, the second step uses (\Cref{def:noncenter-subset}) $\rmloss(c_{i}) = 2^{2\eps z} \cdot \cost(\cluster_{i} \setminus \cluster_{i}^{*}, C - c_{i}) - \cost(\cluster_{i} \setminus \cluster_{i}^{*}, C)$,
the third step applies (\Cref{def:noncenter-subset}) the defining condition of $M' \ni i$, and
the fourth step applies (\Cref{lem:clustering:vertex} of \Cref{lem:clustering}) $\dist(v, C) \le \DSd[v] \le 2^{2\eps} \cdot \dist(v, C)$.

\vspace{.1in}
\noindent
{\bf Case~2: the index is unmatched $\Bar{M}' \ni i$.}
The definition of $\Bar{M}'$ (\Cref{def:noncenter-subset}) ensures that
\begin{align*}
    \exists j \in L \colon \rmloss(c_{j}) + \rmgain(c_{i}^{*}) \le -2^{2\eps z} \cdot \eps \cdot \DSvolume[c_{j}] \cdot \cost(V, C).
\end{align*}
\Cref{eq:P} turns out to hold for every such maintained center $q = c_{j}$; reusing the arguments for {\bf Case~1} (but for $c_{j}$ rather than $c_{i}$), we can obtain \Cref{eq:P}, for $q = c_{j}$, as follows.
\begin{align*}
    \LHS~\text{of}~\eqref{eq:P}
    & ~\le~ \sum_{v \notin \cluster_{j} \cup \cluster_{i}^{*}} \DSd^{z}[v]
    + 2^{2\eps z} \cdot \cost(\cluster_{j} \setminus \cluster_{i}^{*}, C - c_{j}) + \cost(V_{i}^{*}, C) + \rmgain(c_{i}^{*}) \\
    & ~\le~ \sum_{v \notin \cluster_{j} \cup \cluster_{i}^{*}} \DSd^{z}[v]
    + \cost(V_{j} \cup V_{i}^{*}, C) + \rmloss(c_{j}) + \rmgain(c_{i}^{*}) \\
& ~\le~ \sum_{v \notin \cluster_{j} \cup \cluster_{i}^{*}} \DSd^{z}[v]
    + \cost(\cluster_{j} \cup \cluster_{i}^{*}, C)
    - 2^{2\eps z} \cdot \eps \cdot \DSvolume[c_{j}] \cdot \cost(V, C) \\
    & ~\le~ \sum_{v \notin \cluster_{j} \cup \cluster_{i}^{*}} \DSd^{z}[v]
    + \sum_{v \in \cluster_{j} \cup \cluster_{i}^{*}} \DSd^{z}[v]
    - \eps \cdot \DSvolume[c_{j}] \sum_{v \in V} \DSd^{z}[v] \\
    & ~=~ \RHS~\text{of}~\eqref{eq:P}.
\end{align*}
Here, the second step applies (\Cref{def:noncenter-subset}) $\rmloss(c_{j}) = 2^{2\eps z} \cdot \cost(\cluster_{j}, C - c_{j}) - \cost(\cluster_{j}, C) \ge 2^{2\eps z} \cdot \cost(\cluster_{j}\setminus \cluster_{i}^{*}, C - c_{j}) - \cost(\cluster_{j}\setminus \cluster_{i}^{*}, C)$,
the third step applies the above equation for the considered index $j \in L$ (and drops the last term), and
the fourth step applies (\Cref{lem:clustering:vertex} of \Cref{lem:clustering}) $\dist(v, C) \le \DSd[v] \le 2^{2\eps} \cdot \dist(v, C)$.

Combining both cases accomplishes \Cref{lem:noncenter-subset:1} of \Cref{lem:noncenter-subset}.
\end{proof}

Now we move on to \Cref{lem:noncenter-subset:2} of \Cref{lem:noncenter-subset} (which is rephrased for ease of reference). This turns out to be a direct consequence of the following \Cref{claim:rmcost:1,claim:rmcost:2}.

\begin{restate}[\Cref{lem:noncenter-subset}, \Cref{lem:noncenter-subset:2}]
\begin{flushleft}
The premise of \Cref{lem:sample}, namely $\cost(V, \calD.C) \ge \alpha_{z}(\eps) \cdot \OPT$, ensures that $\cost(\BarC, C) \ge \eps^{z + 2} \cdot \cost(V, C)$.
\end{flushleft}
\end{restate}

\begin{claim}
\label{claim:rmcost:1}
\begin{flushleft}
$\cost(\BarC_{i}, C) \ge \eps^{z + 1} \cdot \cost(\cluster_{i}^{*}, C)$, for every index $i \in M' \cup \Bar{M}'$.
\end{flushleft}
\end{claim}

\begin{claim}
\label{claim:rmcost:2}
\begin{flushleft}
The premise of \Cref{lem:sample}, namely $\cost(V, \calD.C) \ge \alpha_{z}(\eps) \cdot \OPT$, ensures that $\sum_{i \in M' \cup \Bar{M}'} \cost(\cluster_{i}^{*}, C) \ge \eps \cdot \cost(V, C)$.
\end{flushleft}
\end{claim}

\begin{proof}[Proof of \Cref{lem:noncenter-subset}, \Cref{lem:noncenter-subset:2} (Assuming \Cref{claim:rmcost:1,claim:rmcost:2})]
Following \Cref{claim:rmcost:1,claim:rmcost:2}:
\begin{align*}
    \cost(\BarC, C)
    ~=~ \sum_{i \in M' \cup \Bar{M}'} \cost(\BarC_{i}, C)
    ~\ge~ \sum_{i \in M' \cup \Bar{M}'} \eps^{z + 1} \cdot \cost(\cluster_{i}^{*}, C)
    ~\ge~ \eps^{z + 2} \cdot \cost(V, C).
\end{align*}
This finishes the proof of \Cref{lem:noncenter-subset:2} of \Cref{lem:noncenter-subset}.
\end{proof}

Let us first establish \Cref{claim:rmcost:1}, as follows.

\begin{proof}[Proof of \Cref{claim:rmcost:1}]
(\Cref{def:noncenter-subset}) Every matched index $i \in M'$ ensures $\rmloss(c_{i}) + \rmgain(c_{i}^{*}) \le 0$,
while every unmatched index $i \in \Bar{M}'$ ensures $\exists j \in L \colon \rmloss(c_{j}) + \rmgain(c_{i}^{*}) \le 0$; as mentioned, we always have $\rmloss(c) \ge 0$, for every maintained center $c \in C$ (of any kind).
Therefore, in either case $i \in M' \cup \Bar{M}'$, we always have $0 \ge \rmgain(c_{i}^{*})
\ge (1 + 2^{4\eps})^{z} \cdot \cost(\cluster_{i}^{*}, \{c_{i}^{*}\}) - \cost(\cluster_{i}^{*}, C)$, which after being rearranged gives
\begin{align*}
    (1 + 2^{4\eps}) \cdot \cost(\cluster_{i}^{*}, \{c_{i}^{*}\})
    & ~\le~ \tfrac{1}{(1 + 2^{4\eps})^{z - 1}} \cdot \cost(\cluster_{i}^{*}, C) \\
    & ~\le~ \tfrac{1}{(1 + 2^{4\eps})^{z - 1}} \sum_{v \in \cluster_{i}^{*}} \big(\dist(v, c_{i}^{*}) + \dist(c_{i}^{*}, C)\big)^{z} \\
    & ~\le~ \cost(\cluster_{i}^{*}, \{c_{i}^{*}\})
    + \tfrac{1}{2^{4\eps(z - 1)}} \cdot |\cluster_{i}^{*}| \cdot \dist^{z}(c_{i}^{*}, C).
\end{align*}
Here, the second step applies the triangle inequality, and the last step applies \Cref{lem:triangle}, by setting the parameter $\lambda = 2^{4\eps}$.
Further rearranging this equation gives
\begin{align*}
    \dist(c_{i}^{*}, C)
    ~\ge~ 2^{4\eps} \cdot \tfrac{1}{|\cluster_{i}^{*}|^{1/z}} \cdot \big(\cost(\cluster_{i}^{*}, \{c_{i}^{*}\})\big)^{1/z}.
\end{align*}
Regarding every noncenter subcluster $\BarC_{i} = \big\{v \in \cluster_{i}^{*} \setminus C \bigmid \dist^{z}(v, c_{i}^{*}) \le 2^{\eps z} \cdot \tfrac{1}{|\cluster_{i}^{*}|} \cdot \cost(\cluster_{i}^{*}, \{c_{i}^{*}\})\big\}$ (\Cref{def:noncenter-subset}), every noncenter $v \in \BarC_{i} \subseteq \cluster_{i}^{*} \setminus C$ therein satisfies that
\begin{align*}
    \dist(v, c_{i}^{*})
~\le~ 2^{\eps} \cdot \tfrac{1}{|\cluster_{i}^{*}|^{1/z}} \cdot \big(\cost(\cluster_{i}^{*}, \{c_{i}^{*}\})\big)^{1/z}.
\end{align*}
Moreover, a simple counting argument implies that $\tfrac{|\cluster_{i}^{*} \setminus \BarC_{i}|}{|\cluster_{i}^{*}|} < 2^{-\eps z} \implies \tfrac{|\BarC_{i}|}{|\cluster_{i}^{*}|} > 1 - 2^{-\eps z}$.
Combining everything together, we can deduce that
\begin{align*}
    \cost(\BarC_{i}, C)
    ~=~ \sum_{v \in \BarC_{i}} \dist^{z}(v, C)
    & ~\ge~ \sum_{v \in \BarC_{i}} \big(\dist(c_{i}^{*}, C) - \dist(v, c_{i}^{*})\big)^{z} \\
    & ~\ge~ (2^{4\eps} - 2^{\eps})^{z} \cdot \tfrac{|\BarC_{i}|}{|\cluster_{i}^{*}|} \cdot \cost(\cluster_{i}^{*}, \{c_{i}^{*}\}) \\
    & ~\ge~ \big(2^{4\eps} - 2^{\eps}\big)^{z} \cdot (1 - 2^{-\eps z}) \cdot \cost(\cluster_{i}^{*}, \{c_{i}^{*}\}) \\
    & ~\ge~ \eps^{z + 1} \cdot \cost(\cluster_{i}^{*}, \{c_{i}^{*}\}).
\end{align*}
Here, the first step applies the triangle inequality, and the last step stems from elementary algebra.
This finishes the proof of \Cref{claim:rmcost:1}. 
\end{proof}

Before moving on to \Cref{claim:rmcost:2}, we first establish the following \Cref{claim:rmloss}, which upper-bounds the function values $\rmloss(c_{i})$ on all {\em matched/lonely} maintained centers $c_{i} \in C_{M} \cup C_{L}$ (\Cref{def:noncenter-subset}).

\begin{claim}
\label{claim:rmloss}
\begin{flushleft}
For every matched/lonely index $i \in M \cup L$ and any parameter $\lambda > 0$:
\begin{align*}
    \rmloss(c_{i})
    ~\le~ 2^{z + 2\eps z} \cdot (1 + \lambda)^{z - 1} \cdot \cost(\cluster_{i}, C^{*})
    + \big(2^{2\eps z} \cdot (1 + 1 / \lambda)^{z - 1} - 1\big) \cdot \cost(\cluster_{i}, C).
\end{align*}
\end{flushleft}
\end{claim}

\begin{proof}
Akin to the maintained clusters $\cluster_{i} = \{v \in V \mid \calD.\DSc[v] = c_{i}\}$ (\Cref{def:noncenter-subset}), we also consider the {\em (maintained) restricted subclusters} $\cluster'_{i} \eqdef \{v \in \cluster_{i} \mid \argmin_{c \in C} \dist(v, c) = c_{i}\} \subseteq \cluster_{i}$.\footnote{I.e., regarding this maintained cluster $\cluster_{i}$, its center $c_{i}$ (in comparison with other maintained centers $\in C$) is just a {\em $2^{2\eps}$-approximate nearest center} $\dist(v, c_{i}) \le 2^{2\eps} \cdot \dist(v, C)$ for a {\em generic} vertex $v \in \cluster_{i}$ (\Cref{lem:clustering,def:noncenter-subset}), but is an {\em exact nearest center} $\dist(v, c_{i}) = \dist(v, C)$ for a {\em restricted} vertex $v \in \cluster'_{i}$.}
We prove \Cref{claim:rmloss} by reasoning about a matched index $i \in M$ versus a lonely index $i \in L$ separately.

\vspace{.1in}
\noindent
{\bf Case~1: A matched index $i \in M$.}
We can reformulate the function value $\rmloss(c_{i})$, as follows.
\begin{align*}
    \rmloss(c_{i})
    & ~=~ 2^{2\eps z} \cdot \cost(\cluster_{i} \setminus \cluster_{i}^{*}, C - c_{i}) - \cost(\cluster_{i} \setminus \cluster_{i}^{*}, C) \\
    & ~=~ 2^{2\eps z} \cdot \big(\cost(\cluster_{i} \setminus \cluster_{i}^{*}, C - c_{i}) - \cost(\cluster_{i} \setminus \cluster_{i}^{*}, C)\big)
    + (2^{2\eps z} - 1) \cdot \cost(\cluster_{i} \setminus \cluster_{i}^{*}, C) \\
    & ~=~ 2^{2\eps z} \cdot \big(\cost(\cluster'_{i} \setminus \cluster_{i}^{*}, C - c_{i}) - \cost(\cluster'_{i} \setminus \cluster_{i}^{*}, C)\big)
    + (2^{2\eps z} - 1) \cdot \cost(\cluster_{i} \setminus \cluster_{i}^{*}, C).
    \hspace{.23cm}
\end{align*}
A vertex $v \in \cluster'_{i} \setminus \cluster_{i}^{*}$ locates in the center-$(c_{v}^{*} \eqdef \argmin_{c^{*} \in C^{*}} \dist(v, c^{*}))$ optimal cluster $\cluster_{v}^{*}$;
observe that $c_{v}^{*} \ne c_{i}^{*} \impliedby (v \in \cluster_{v}^{*}) \wedge (v \notin \cluster_{i}^{*})$.
This optimal center $c_{v}^{*}$ is captured by its nearest maintained center $c_{v} \eqdef \argmin_{c \in C} \dist(c_{v}^{*}, c)$;
observe that $c_{v} \in C - c_{i} \impliedby c_{v} \ne c_{i} \impliedby i \in M$,
namely $c_{v}^{*} \ne c_{i}^{*}$ but the latter $c_{i}^{*}$ (\Cref{def:noncenter-subset}) is the unique optimal center $\in C^{*}$ captured by the index-$(i \in M)$ {\em matched} maintained center $c_{i}$.
Consequently, we can deduce that, for any parameter $\lambda > 0$,
\begin{align*}
    \dist^{z}(v, C - c_{i})
    & ~\le~ \dist^{z}(v, c_{v})
    \tag{$c_{v} \in C - c_{i}$} \\
    & ~\le~ \big(\dist(v, c_{v}^{*}) + \dist(c_{v}^{*}, c_{v})\big)^{z}
    \tag{Triangle inequality} \\
    & ~\le~ \big(\dist(v, c_{v}^{*}) + \dist(c_{v}^{*}, c_{i})\big)^{z}
    \tag{Definition of $c_{v}$} \\
    & ~\le~ \big(2\dist(v, c_{v}^{*}) + \dist(v, c_{i})\big)^{z}
    \tag{Triangle inequality} \\
    & ~=~ \big(2\dist(v, C^{*}) + \dist(v, C)\big)^{z}
    \tag{Definitions of $c_{v}^{*}$ and $\cluster'_{i}$} \\
    & ~\le~ (1 + \lambda)^{z - 1} \cdot 2^{z} \cdot \dist^{z}(v, C^{*})
    + (1 + 1 / \lambda)^{z - 1} \cdot \dist^{z}(v, C).
    \hspace{3.21cm}
\end{align*}
Here, the last step applies \Cref{lem:triangle}.
Combining the above two equations gives
\begin{align*}
    \rmloss(c_{i})
    & ~\le~ 2^{2\eps z} \cdot \big((1 + \lambda)^{z - 1} \cdot 2^{z} \cdot \cost(\cluster'_{i} \setminus \cluster_{i}^{*}, C^{*})
    + \big((1 + 1 / \lambda)^{z - 1} - 1\big) \cdot \cost(\cluster'_{i} \setminus \cluster_{i}^{*}, C)\big) \\
    & \phantom{~=~}\qquad + (2^{2\eps z} - 1) \cdot \cost(\cluster_{i} \setminus \cluster_{i}^{*}, C) \\
    & ~\le~ 2^{2\eps z} \cdot \big((1 + \lambda)^{z - 1} \cdot 2^{z} \cdot \cost(\cluster_{i}, C^{*})
    + \big((1 + 1 / \lambda)^{z - 1} - 1\big) \cdot \cost(\cluster_{i}, C)\big) \\
    & \phantom{~=~}\qquad + (2^{2\eps z} - 1) \cdot \cost(\cluster_{i}, C) \\
    & ~=~ 2^{z + 2\eps z} \cdot (1 + \lambda)^{z - 1} \cdot \cost(\cluster_{i}, C^{*})
    + \big(2^{2\eps z} \cdot (1 + 1 / \lambda)^{z - 1} - 1\big) \cdot \cost(\cluster_{i}, C).
\end{align*}
This finishes the proof of \Cref{claim:rmloss}, in case of a {\em matched} index $i \in M$.

\vspace{.1in}
\noindent
{\bf Case~2: A lonely index $i \in L$.}
We likewise reformulate the function value $\rmloss(c_{i})$, as follows.
\begin{align*}
    \rmloss(c_{i})
& ~=~ 2^{2\eps z} \cdot \big(\cost(\cluster'_{i}, C - c_{i}) - \cost(\cluster'_{i}, C)\big)
    + (2^{2\eps z} - 1) \cdot \cost(\cluster_{i}, C).
    \hspace{2.73cm}
\end{align*}
A vertex $v \in \cluster'_{i}$ locates in the center-$(c_{v}^{*} \eqdef \argmin_{c^{*} \in C^{*}} \dist(v, c^{*}))$ optimal cluster $\cluster_{v}^{*}$; this optimal center $c_{v}^{*}$ is captured by its nearest maintained center $c_{v} \eqdef \argmin_{c \in C} \dist(c_{v}^{*}, c)$.
We observe that $c_{v} \in C - c_{i} \impliedby c_{v} \ne c_{i} \impliedby i \in L$, namely no optimal center $\in C^{*}$ (\Cref{def:noncenter-subset}) is captured by the index-$(i \in L)$ {\em lonely} maintained center $c_{i}$.
Consequently, reapplying the arguments for {\bf Case~1}, we can likewise deduce that, for any parameter $\lambda > 0$,
\begin{align*}
    \dist^{z}(v, C - c_{i})
    & ~\le~ (1 + \lambda)^{z - 1} \cdot 2^{z} \cdot \dist^{z}(v, C^{*})
    + (1 + 1 / \lambda)^{z - 1} \cdot \dist^{z}(v, C).
    \hspace{3.20cm}
\end{align*}
Combining the above two equations and following the same steps as in {\bf Case~1}, we likewise have
\begin{align*}
    \rmloss(c_{i})
    & ~\le~ 2^{z + 2\eps z} \cdot (1 + \lambda)^{z - 1} \cdot \cost(\cluster_{i}, C^{*})
    + \big(2^{2\eps z} \cdot (1 + 1 / \lambda)^{z - 1} - 1\big) \cdot \cost(\cluster_{i}, C).
    \hspace{1.38cm}
\end{align*}
This finishes the proof of \Cref{claim:rmloss}, in case of a {\em lonely} index $i \in L$.
\end{proof}

Now we are ready to establish \Cref{claim:rmcost:2}, by leveraging the above \Cref{claim:rmloss}.

\begin{proof}[Proof of \Cref{claim:rmcost:2}]
\Cref{claim:rmcost:2} holds if and only if $\sum_{i \notin M' \cup \Bar{M}'} \cost(\cluster_{i}^{*}, C) < (1 - \eps) \cdot \cost(V, C)$, so we shall upper-bound the $\LHS$ of this equation.

We have $\rmgain(c_{i}^{*}) = 2^{(1 + 4\eps) z} \cdot \cost(\cluster_{i}^{*}, \{c_{i}^{*}\}) - \cost(\cluster_{i}^{*}, C)$, for every optimal center $c_{i}^{*} \in C^{*}$ (\Cref{def:noncenter-subset}).
The defining condition of the {\em matched index subset} $M' \subseteq M$ (\Cref{def:noncenter-subset}) implies that, for every {\em matched} index $i \in M \setminus M'$,
\begin{align*}
    \rmloss(c_{i}) + \rmgain(c_{i}^{*})
    & ~\ge~ -2^{2\eps z} \cdot \eps \cdot \DSvolume[c_{i}] \cdot \cost(V, C) \\
    \implies\qquad\quad
    \cost(\cluster_{i}^{*}, C)
    & ~\le~ \rmloss(c_{i})
    + 2^{2\eps z} \cdot \eps \cdot \DSvolume[c_{i}] \cdot \cost(V, C)
    \hspace{1.56cm} \\
    & \phantom{~=~}\qquad\qquad
    + 2^{(1 + 4\eps) z} \cdot \cost(\cluster_{i}^{*}, \{c_{i}^{*}\}).
\end{align*}
Further, the defining condition of the {\em unmatched index subset} $\Bar{M}' \subseteq \Bar{M}$ (\Cref{def:noncenter-subset}) implies that, for every {\em unmatched} index $i \in \Bar{M} \setminus \Bar{M}'$,
\begin{align*}
    \rmloss(c_{j}) + \rmgain(c_{i}^{*})
    & ~\ge~ -2^{2\eps z} \cdot \eps \cdot \DSvolume[c_{j}] \cdot \cost(V, C),\qquad\qquad
    \forall j \in L, \\
    \implies\qquad\quad
    \cost(\cluster_{i}^{*}, C)
    & ~\le~ \tfrac{1}{|L|} \sum_{j \in L} \big(\rmloss(c_{j}) + 2^{2\eps z} \cdot \eps \cdot \DSvolume[c_{j}] \cdot \cost(V, C)\big) \\
    & \phantom{~=~}\qquad\qquad\qquad\qquad
    + 2^{(1 + 4\eps) z} \cdot \cost(\cluster_{i}^{*}, \{c_{i}^{*}\}). 
\end{align*}
From a simple counting argument about (\Cref{def:noncenter-subset}) the number of centers
$|\Bar{M}| + |M|
= |C|
= k
= |C^{*}|
\ge |M| + 0 \cdot |L| + 2 \cdot |\Bar{M} \setminus L|
= |M| + 2|\Bar{M}| - 2|L|$, we can conclude with $|L| \ge \frac{|\Bar{M}|}{2} \ge \frac{|\Bar{M} \setminus \Bar{M}'|}{2}$.
Also, we have $\sum_{c \in C} \DSvolume[c] = 1$ (\Cref{def:super-effective}).

\noindent
Combining everything together, we can deduce that, for any parameter $\lambda > 0$,
\begin{align}
    \sum_{i \notin M' \cup \Bar{M}'} \cost(\cluster_{i}^{*}, C)
    & ~\le~ \Big(\sum_{i \in M \setminus M'} \rmloss(c_{i}) + \tfrac{|\Bar{M} \setminus \Bar{M}'|}{|L|} \sum_{j \in L} \rmloss(c_{j})\Big)
    \notag \\
    & \phantom{~=~}\qquad + \Big(\sum_{i \in M \setminus M'} \DSvolume[c_{i}] + \tfrac{|\Bar{M} \setminus \Bar{M}'|}{|L|} \sum_{j \in L} \DSvolume[c_{j}]\Big) \cdot 2^{2\eps z} \cdot \eps \cdot \cost(V, C)
    \notag \\
    & \phantom{~=~}\qquad + 2^{(1 + 4\eps) z} \cdot \sum_{i \notin M' \cup \Bar{M}'} \cost(\cluster_{i}^{*}, \{c_{i}^{*}\})
    \notag \\
    & ~\le~ \sum_{i \in M \cup L} 2\rmloss(c_{i})
    + 2^{1 + 2\eps z} \cdot \eps \cdot \cost(V, C)
    + 2^{(1 + 4\eps) z} \cdot \OPT
    \notag \\
    & ~\le~ \sum_{i \in M \cup L} \Big(2^{1 + z + 2\eps z} \cdot (1 + \lambda)^{z - 1} \cdot \cost(\cluster_{i}, C^{*})
    \notag \\
    & \phantom{~=~}\qquad + (2^{1 + 2\eps z} \cdot (1 + 1 / \lambda)^{z - 1} - 2) \cdot \cost(\cluster_{i}, C)\Big)
    \notag \\
    & \phantom{~=~}\qquad + 2^{1 + 2\eps z} \cdot \eps \cdot \cost(V, C)
    + 2^{(1 + 4\eps) z} \cdot \OPT
    \notag \\
    & ~\le~ 2^{1 + z + 2\eps z} \cdot (1 + \lambda)^{z - 1} \cdot \OPT
    \notag \\
    & \phantom{~=~}\qquad + (2^{1 + 2\eps z} \cdot (1 + 1 / \lambda)^{z - 1} - 2) \cdot \cost(V, C)
    \notag \\
    & \phantom{~=~}\qquad + 2^{1 + 2\eps z} \cdot \eps \cdot \cost(V, C)
    + 2^{(1 + 4\eps) z} \cdot \OPT
    \notag \\
    & ~=~ \big(2^{1 + z + 2\eps z} \cdot (1 + \lambda)^{z - 1} + 2^{(1 + 4\eps) z}\big) \cdot \OPT
    \notag \\
    & \phantom{~=~}\qquad + \big(2^{1 + 2\eps z} \cdot ((1 + 1 / \lambda)^{z - 1} + \eps) - 2\big) \cdot \cost(V, C).
    \label{eq:rmcost:1:1}
\end{align}
Here, the third step applies \Cref{claim:rmloss}, for every {\em matched/lonely} index $i \in M \cup L$, and the last two steps stems from elementary algebra.

Rearranging \Cref{eq:rmcost:1:1} gives the premise of \Cref{claim:rmcost:2}, essentially, as follows.
\begin{align*}
    \text{\Cref{claim:rmcost:2}}
    & ~\impliedby~ \RHS~\text{of}~\eqref{eq:rmcost:1:1}
    ~\le~ (1 - \eps) \cdot \cost(V, C) \\
    & ~\iff~
    \frac{\cost(V, C)}{\OPT}
    ~\ge~ \alpha_{z}(\eps, \lambda)
    ~=~ \frac{2^{1 + z + 2\eps z} \cdot (1 + \lambda)^{z - 1} + 2^{(1 + 4\eps) z}}{3 - \eps - 2^{1 + 2\eps z} \cdot ((1 + 1 / \lambda)^{z - 1} + \eps)}
    ~>~ 0.
    \hspace{.03cm}
\end{align*}
As mentioned in \Cref{footnote:alpha}, the only constraint on the parameter $\lambda > 0$ is that $\alpha_{z}(\eps, \lambda) > 0 \iff \lambda > ((\frac{3 - \eps}{2^{1 + 2\eps z}} - \eps)^{1 / (z - 1)} - 1)^{-1}$.
We thus conclude with
\begin{align*}
    \text{\Cref{claim:rmcost:2}}
    & ~\impliedby~
    \frac{\cost(V, C)}{\OPT}
    ~\ge~ \min_{\lambda} \bigg\{\,\alpha_{z}(\eps, \lambda) \,\biggmid\, \lambda > \frac{1}{((3 - \eps) / 2^{1 + 2\eps z} - \eps)^{1 / (z - 1)} - 1}\,\bigg\} \\
    & ~\iff~
    \cost(V, C)
    ~\ge~ \alpha_{z}(\eps) \cdot \OPT.
    \qedhere
\end{align*}
\end{proof}

\subsection{Performance guarantees of our local search}
\label{subsec:local-search}

In this subsection, we will establish the performance guarantees of our algorithm {\LocalSearch}, by combining everything presented thus far with additional arguments.
The reader can reference \Cref{fig:LocalSearch} for the algorithm {\LocalSearch} and the subroutine {\TestEffectiveness}, as well as \Cref{fig:initialize,fig:insert,fig:delete,fig:sample-noncenter} for the operations {\Initialize}, {\Insert}, {\Delete}, and {\SampleNoncenter}.

For ease of presentation, we would establish the correctness of our algorithm {\LocalSearch} first (\Cref{thm:LS-correctness}) and its running time afterward (\Cref{thm:LS-runtime}).

\begin{theorem}[{\LocalSearch}; correctness]
\label{thm:LS-correctness}
\begin{flushleft}
For the randomized algorithm $\LocalSearch(G, k)$:
\begin{enumerate}[font = {\em\bfseries}]
    \item \label{thm:LS-correctness:feasible}
    It will return a feasible solution $\Cterminal \in V^{k}$.
    
    \item \label{thm:LS-correctness:approx}
    This random solution $\Cterminal$ is an $\alpha_{z}(\eps)$-approximation to the optimal solution $C^{*}$, with probability $\ge 1 - n^{-\Theta(1)}$,
    by setting the parameter $\DSnumber = \eps^{-\Theta(z)} \log(n)$ large enough.
\end{enumerate}
\end{flushleft}
\end{theorem}

\begin{proof}
The algorithm {\LocalSearch} starts by finding an initial $n^{z + 1}$-approximate feasible solution $\Cinitial \in V^{k}$, based on  \Cref{prop:naive_solution}, and initializing the data structure $\calD$ (\Cref{fig:initialize}), through the operation $\Initialize(G, \Cinitial)$; likewise for the other data structures $\{\calD'_{\sigma}\}_{\sigma \in [\DSnumber]}$.
\hfill
(Lines~\ref{alg:LS:naive} and \ref{alg:LS:initialize})

\noindent
Afterward, the algorithm {\LocalSearch} starts iterating Lines~\ref{alg:LS:sample} to \ref{alg:LS:return}:
\hfill
(Line~\ref{alg:LS:repeat})

\noindent
Specifically, a single iteration samples a number of $\DSnumber$ noncenters, $\cinsert_{\sigma} \gets \calD.\SampleNoncenter()$ for $\sigma \in [\DSnumber]$, and schedules a number of $\DSnumber$ subroutines, $\TestEffectiveness(\calD, \calD'_{\sigma}, \cinsert_{\sigma})$ for $\sigma \in [\DSnumber]$, step by step using the round-robin algorithm.
\hfill
(Lines~\ref{alg:LS:sample} and \ref{alg:LS:schedule})

\noindent
Depending on the outputs of these $\DSnumber$ subroutines, an iteration falls into either {\bf Case~1} or {\bf Case~2}:

\vspace{.1in}
\noindent
{\bf Case~1:}
One of the $\DSnumber$ subroutines, say $\sigma^{*} \in [\DSnumber]$, first returns a pair $(\cinsert_{\sigma^{*}}, \cdelete_{\sigma^{*}})$.
\hfill
(Line~\ref{alg:LS:positive})

\noindent
In this case, the algorithm terminates all (ongoing) subroutines and modifies the data structure $\calD$ through the operations $\calD.\Insert(\cinsert_{\sigma^{*}})$ and $\calD.\Delete(\cdelete_{\sigma^{*}})$; likewise for $\{\calD'_{\sigma}\}_{\sigma \in [\DSnumber]}$.
\hfill
(Lines~\ref{alg:LS:terminate} and \ref{alg:LS:swap})

\noindent
\Comment{Let us call such a ``{\bf Case~1}'' iteration a {\em positive} iteration.}

\noindent
After this {\em positive} iteration, the algorithm proceeds to the next iteration.

\vspace{.1in}
\noindent
{\bf Case~2:}
All of the $\DSnumber$ subroutines (terminate and) return {\failure}'s.
\hfill
(Line~\ref{alg:LS:negative})

\noindent
In this case, the algorithm returns $\Cterminal \gets \calD.C$ as its solution.
\hfill
(Line~\ref{alg:LS:return})

\noindent
\Comment{Let us call such a ``{\bf Case~2}'' iteration a {\em negative} iteration.}

\noindent
After this {\em negative} iteration, the algorithm terminates.

\vspace{.1in}
\noindent
To summarize, the algorithm {\LocalSearch} experiences first a number of $\ell \ge 0$ positive iterations -- we assume for the moment that $\ell$ is finite, but will verify this soon after -- and then the terminal negative iteration, thus returning a solution $\Cterminal \gets \calD.C$.

\newcommand{\Dinitial}{\calD^{\sf init}}
\newcommand{\Dterminal}{\calD^{\sf term}}

\noindent
\Comment{For ease of presentation, below we denote by $\Dinitial$ the data structure $\calD$ initialized in Line~\ref{alg:LS:initialize} and by $\Dterminal$ the (unmodified) data structure $\calD$ in the terminal negative iteration.}

As the whole process of the algorithm {\LocalSearch} now are clear, we can prove \Cref{thm:LS-correctness}.

\vspace{.1in}
\noindent
{\bf \Cref{thm:LS-correctness:feasible}.}
The initial solution $\Cinitial$ found in Line~\ref{alg:LS:naive} {\em is} feasible $|\Cinitial| = k$ (\Cref{prop:naive_solution}), and so is the maintained center set $|\Dinitial.C| = |\Cinitial| = k$ (\Cref{lem:initialize:center} of \Cref{lem:initialize}) of the data structure $\Dinitial$ initialized in Line~\ref{alg:LS:initialize}.
Afterward, every positive iteration swaps a noncenter $\cinsert_{\sigma^{*}} \notin \calD.C$ and a center $\cdelete_{\sigma^{*}} \in \calD.C$ (Lines~\ref{alg:LS:positive} to \ref{alg:LS:swap}), preserving the feasibility $|\calD.C + \cinsert_{\sigma^{*}} - \cdelete_{\sigma^{*}}| = |\calD.C| = k$.
Moreover, the terminal negative iteration returns $\Cterminal \gets \Dterminal.C$ the center set maintained by the terminal data structure $\Dterminal$, hence $|\Cterminal| = |\Dterminal.C| = k$.
In sum, the returned solution $\Cterminal$ {\em is} feasible.

\vspace{.1in}
\noindent
{\bf \Cref{thm:LS-correctness:approx}.}
To show that the algorithm {\LocalSearch} successes with high probability, we would first prove the following upper bound on the number $\ell$ of positive iterations. (Notice that $0 < \eps < \frac{1}{10z}$.)
\begin{align}
\label{eq:positive-iteration}
    \ell
    ~\le~ \Bar{\ell}
    ~\eqdef~ \lceil 8z \eps^{-1} m |\calJ| \ln(n) \rceil
    ~=~ \eps^{-O(1)} m \log^{2}(n).
\end{align}
In every positive iteration $i \ge 1$, the pair $(\cinsert_{i}, \cdelete_{i})$ (say) obtained in Line~\ref{alg:LS:positive} must be super-effective (\Cref{lem:test:1} of \Cref{lem:test}), so the subsequent invocation of the operations $\calD.\Insert(\cinsert_{i})$ and $\calD.\Delete(\cdelete_{i})$ in Line~\ref{alg:LS:swap} decreases $\calD.\DScost$ (\Cref{def:effective} and \Cref{lem:deletion-estimator:DSvolume} of \Cref{lem:deletion-estimator}) by a multiplicative factor of $\le 1 - (\eps / 2) \cdot \calD'.\DSvolume[\cdelete_{i}]
\le e^{-(\eps / 2) \cdot \calD'.\DSvolume[\cdelete_{i}]}$.
We thus deduce that
\begin{align}
    &\hspace{-1.09cm} \sum_{i \in [\ell]} \calD'.\DSvolume[\cdelete_{i}]
    ~\le~ 4z \eps^{-1} \ln(n)
    ~=~ \eps^{-O(1)} \log(n).
    \label{eq:total-volume} \\
    \impliedby \OPT
    & ~\le~ \cost(V, \Dterminal.C)
    \tag{Optimality of $\OPT$} \\
    & ~\le~ \Dterminal.\DScost
    \tag{\Cref{lem:DScost}} \\
    & ~\le~ \frac{\Dinitial.\DScost}{\exp((\eps / 2) \sum_{i \in [\ell]} \calD'.\DSvolume[\cdelete_{i}])}
    \tag{$\ell$ positive iterations} \\
    & ~=~ \frac{\cost(V, \Cinitial)}{\exp((\eps / 2) \sum_{i \in [\ell]} \calD'.\DSvolume[\cdelete_{i}])}
    \tag{\Cref{lem:initialize:center,lem:initialize:DScost} of \Cref{lem:initialize}} \\
    & ~\le~ \frac{n^{z + 1} \cdot \OPT}{\exp((\eps / 2) \sum_{i \in [\ell]} \calD'.\DSvolume[\cdelete_{i}])}.
    \tag{\Cref{prop:naive_solution}}
\end{align}
So we can infer \Cref{eq:positive-iteration} from a combination of \Cref{eq:total-volume} and that $\calD'.\DSvolume[\cdelete_{i}] \ge \frac{1}{2m |\calJ|}$, for every positive iteration $i \ge 1$ (\Cref{lem:deletion-estimator:DSvolume} of \Cref{lem:deletion-estimator}).
The above arguments also indicate that, if our algorithm {\LocalSearch} could experience the iteration $(\Bar{\ell} + 1)$, then:\\
(i)~The maintained center set $\calD.C$ must be an $\alpha_{z}(\eps)$-approximation to the optimal solution $C^{*}$.\\
(ii)~This iteration $(\Bar{\ell} + 1)$ must be the terminal negative iteration and returns the maintained center set $\Cterminal \gets \calD.C$ the solution. Thus, our algorithm {\LocalSearch} succeeds.

The returned random solution $\Cterminal \in V^{k}$ (Line~\ref{alg:LS:return}) fails to be an $\alpha_{z}(\eps)$-approximation to the optimal solution $C^{*}$ if and only if both events $\calE^{1}_{i}$ and $\calE^{2}_{i}$ occur in some iteration $i \in [\Bar{\ell}]$:
\begin{align*}
    \calE^{1}_{i}
    & ~\eqdef~ \big\{\,\text{$\cost(V, \calD.C) \geq \alpha_{z}(\eps) \cdot \OPT$ in the iteration $i$}\,\big\}. \\
    \calE^{2}_{i}
    & ~\eqdef~ \big\{\,\text{all of the $\DSnumber$ subroutines in the iteration $i$ return {\failure}'s}\,\big\}.
\end{align*}
I.e., the maintained center set $\calD.C \in V^{k}$
(provided $\calE^{1}_{i}$) has yet been an $\alpha_{z}(\eps)$-approximation to the optimal solution $C^{*}$ but
(provided $\calE^{2}_{i}$) this iteration {\em is} the terminal negative iteration and returns the maintained center set $\Cterminal \gets \calD.C$ the solution.

Combining all above observations, we can upper-bound the overall failure probability as follows:
\begin{align*}
    \Pr\big[\,\text{{\LocalSearch} fails}\,\big]
    & ~=~ \Pr\big[\,\cup_{i \in [\Bar{\ell}]} (\calE^{1}_{i} \cap \calE^{2}_{i})\,\big] \\
    & ~\le~ \Bar{\ell} \cdot (1 - \eps^{4z})^{\DSnumber}
    \tag{\Cref{lem:test,lem:sample} and union bound} \\
    & ~\le~ \eps^{-O(1)} m \log^{2}(n) \cdot (1 - \eps^{4z})^{\DSnumber}
    \tag{\Cref{eq:positive-iteration}} \\
    & ~\le~ n^{-\Theta(1)}.
\end{align*}
Here, the last step holds whenever the parameter $\DSnumber = \eps^{-\Theta(z)} \log(n)$ is large enough.

This finishes the proof of \Cref{thm:LS-correctness}.
\end{proof}

\newcommand{\Phiinitial}{\Phi^{\sf init}}
\newcommand{\Phiterminal}{\Phi^{\sf term}}

Finally, we turn to bounding the running time of our algorithm {\LocalSearch} (\Cref{thm:LS-runtime}).
For ease of reference, we rephrase our previous results about running time and potential.

\begin{restate}[\Cref{prop:naive_solution}]
\begin{flushleft}
An $n^{z + 1}$-approximate feasible solution $\Cinitial \in V^{k}$ to the {\kzC} problem can be found (deterministically) in time $O(m \log(n))$.
\end{flushleft}
\end{restate}

\begin{restate}[\Cref{lem:potential}]
\begin{flushleft}
$0 \leq \Phi \le \Phi_{\max}$, for the $\Phi_{\max} = 2m |\calJ| \cdot \log_{2}(1 + \dmax) = O(z m \log^{2}(n))$.
\end{flushleft}
\end{restate}

\begin{restate}[\Cref{lem:sample-noncenter}]
\begin{flushleft}
The operation $\calD.\SampleNoncenter()$ has worst-case running time $O(\log(n))$.
\end{flushleft}
\end{restate}

\begin{restate}[\Cref{lem:initialize}, \Cref{lem:initialize:runtime}]
\begin{flushleft}
The operation $\calD.\Initialize(\Cinitial)$ has worst-case running time $\Tinitialize = O(m \log(n) + n \log^{2}(n))$.
\end{flushleft}
\end{restate}

\begin{restate}[\Cref{lem:insert}, \Cref{lem:insert:runtime}]
\begin{flushleft}
The operation $\calD' \gets \calD.\Insert(\cinsert)$ has worst-case running time $\Tinsert = (\Phi - \Phi') \cdot \eps^{-O(1)} \beta \log(n) \le \Phi_{\max} \cdot \eps^{-O(1)} \beta \log(n)$.
\end{flushleft}
\end{restate}

\begin{restate}[\Cref{lem:delete}, \Cref{lem:delete:runtime,lem:delete:potential}]
\begin{flushleft}
The operation $\calD'' \gets \calD'.\Delete(\cdelete)$ has worst-case running time $\Tdelete = \DSvolume[\cdelete] \cdot O(m \log^{3}(n))$ and potential change $\Phi'' - \Phi' \le \DSvolume[\cdelete] \cdot \Phi_{\max}$.
\end{flushleft}
\end{restate}

\begin{theorem}[{\LocalSearch}; running time]
\label{thm:LS-runtime}
\begin{flushleft}
The randomized algorithm $\LocalSearch(G, k)$ has worst-case running time $T = \DSnumber \cdot \eps^{-O(1)} m \beta \log^{4}(n)$.
\end{flushleft}
\end{theorem}

\begin{proof}
Without ambiguity, here we readopt the terminologies and the notations from the proof of \Cref{thm:LS-correctness}.
As mentioned, the algorithm {\LocalSearch} works as follows:

First, we find a feasible solution $\Cinitial \in V^{k}$ and initialize $\DSnumber + 1$ data structures $\calD$ and $\{\calD'_{\sigma}\}_{\sigma \in [s]}$,
which takes time $O(m \log(n)) + (\DSnumber + 1) \cdot O(m \log(n) + n \log^{2}(n)) = \DSnumber \cdot O(m \log^{2}(n))$ by \Cref{prop:naive_solution} and \Cref{lem:initialize:runtime} of \Cref{lem:initialize}.
\hfill
(Lines~\ref{alg:LS:naive} and \ref{alg:LS:initialize})

Afterward, we start iterating Lines~\ref{alg:LS:sample} to \ref{alg:LS:return}:
\hfill
(Line~\ref{alg:LS:repeat})

\noindent
Specifically, a single iteration samples a number of $\DSnumber$ noncenters, $\cinsert_{\sigma} \gets \calD.\SampleNoncenter()$ for $\sigma \in [\DSnumber]$, and schedules a number of $\DSnumber$ subroutines, $\TestEffectiveness(\calD, \calD'_{\sigma}, \cinsert_{\sigma})$ for $\sigma \in [\DSnumber]$, step by step using the round-robin algorithm.
\hfill
(Lines~\ref{alg:LS:sample} and \ref{alg:LS:schedule})

\noindent
The following \Cref{claim:test-runtime} upper-bounds the running time of a single subroutine.

\begin{claim}
\label{claim:test-runtime}
\begin{flushleft}
The subroutine $\TestEffectiveness(\calD, \calD', \cinsert)$ has worst-case running time $O(\Tinsert_{\calD}) + O(\log^{2}(n))$, where the $\Tinsert_{\calD}$ denotes the running time of the operation $\calD.\Insert(\cinsert)$.
\end{flushleft}
\end{claim}

\begin{proof}
We go through the subroutine $\TestEffectiveness(\calD, \calD', \cinsert)$ step by step.

\noindent
First, modify the second input data structure $\calD'$ via the operation $\calD'.\Insert(\cinsert)$.
\hfill
(Line~\ref{alg:test:insert})\\
\Comment{Line~\ref{alg:test:insert} runs in time $\Tinsert_{\calD}$, since both data structures $\calD'$ and $\calD$ fed to the subroutine are identical.}

\noindent
Then, for every center group $\calD'.G_{\tau}$, $\forall \tau \in [\CGnumber]$:
\hfill
(Line~\ref{alg:test:for})

\noindent
(i)~Find the ``$\cinsert$-excluded $\calD'.\DSloss[c]$-minimizer'' $\cdelete_{\tau} = \argmin_{c \in \calD'.G_{\tau} - \cinsert} \calD'.\DSloss[c]$.\footnote{We simply skip this center group $\tau \in [\CGnumber]$, when its ``$\cinsert$-excluded $\calD'.\DSloss[c]$-minimizer'' $\cdelete_{\tau}$ does not exist, namely when $(\calD'.G_{\tau} = \emptyset) \vee (\calD'.G_{\tau} = \{\cinsert\})$.}
\hfill
(Line~\ref{alg:test:cdelete})\\
(ii)~Test the condition in Line~\ref{alg:test:if} for this center $\cdelete_{\tau} \in \calD'.G_{\tau}$.
\hfill
(Line~\ref{alg:test:if})\\
(iii)~Return the pair $(\cinsert, \cdelete) \gets (\cinsert, \cdelete_{\tau})$ if the test passes.
\hfill
(Line~\ref{alg:test:return-pair})\\
\Comment{Line~\ref{alg:test:cdelete} runs in time $O(\log(n))$, by \Cref{lem:grouping}; the center group $\calD'.G_{\tau}$ is stored in the red-black tree $\calD'.\calG_{\tau}$, which has size $|\calD'.\calG_{\tau}| = |\calD'.G_{\tau}| \le |\calD'.C| \le n$ and supports searching in time $O(\log(n))$. Lines~\ref{alg:test:if} and \ref{alg:test:return-pair} clearly run in time $O(1)$.}

\noindent
If no pair has ever be returned in Line~\ref{alg:test:return-pair}, then return a {\failure}.
\hfill
(Line~\ref{alg:test:return-failure})\\
\Comment{Line~\ref{alg:test:return-failure} clearly runs in time $O(1)$.}

\noindent
Finally, the second data structure $\calD'$ backtracks such that $\calD \equiv \calD'$ again.\\
\Comment{This backtracking can be easily implemented by a {\em stack} data structure and runs in time $O(\Tinsert_{\calD})$.}

Overall, given that $\CGnumber = \lceil \log_{2}(2m|\calJ|) + 1\rceil = O(\log(n))$ (Invariant~\ref{invar:grouping:group}), the subroutine takes time $\Tinsert_{\calD} + \CGnumber \cdot O(\log(n)) + O(\Tinsert_{\calD}) = O(\Tinsert_{\calD}) + O(\log^{2}(n))$.
This finishes the proof of \Cref{claim:test-runtime}.
\end{proof}

Depending on the outputs of these $\DSnumber$ subroutines, $\TestEffectiveness(\calD, \calD'_{\sigma}, \cinsert_{\sigma})$ for $\sigma \in [\DSnumber]$, an iteration $i \ge 1$ falls into either {\bf Case~1} or {\bf Case~2} (i.e., a positive or negative iteration):

\begin{flushleft}
{\bf Case~1:}
One of the $\DSnumber$ subroutines, say $\sigma^{*} \in [\DSnumber]$, first returns a pair $(\cinsert_{\sigma^{*}}, \cdelete_{\sigma^{*}})$.

Then, we terminate all (ongoing) subroutines and modify the data structure $\calD$, via the operations $\calD.\Insert(\cinsert_{\sigma^{*}})$ and $\calD.\Delete(\cdelete_{\sigma^{*}})$; likewise for $\{\calD'_{\sigma}\}_{\sigma \in [\DSnumber]}$.
\hfill
(Lines~\ref{alg:LS:sample} and \ref{alg:LS:schedule} $\to$ Lines~\ref{alg:LS:positive} to \ref{alg:LS:swap})

\vspace{.1in}
(Line~\ref{alg:LS:sample})~The sampling of all noncenters $\{\cinsert_{\sigma}\}_{\sigma \in [\DSnumber]}$ takes time $\DSnumber \cdot O(\log(n))$, by \Cref{lem:sample-noncenter}.

(Line~\ref{alg:LS:schedule} $\to$ Lines~\ref{alg:LS:positive} and \ref{alg:LS:terminate})~All subroutines together take time $\DSnumber \cdot O(\Tinsert_{i}) + \DSnumber \cdot O(\log^{2}(n))$, where the $\Tinsert_{i}$ denotes the running time of the operation $\calD.\Insert(\cinsert_{\sigma^{*}})$, by \Cref{claim:test-runtime} and the nature of the round-robin algorithm.

(Line~\ref{alg:LS:swap})~The operations $\calD.\Insert(\cinsert_{\sigma^{*}})$ and $\calD.\Delete(\cdelete_{\sigma^{*}})$ take time $\Tinsert_{i}$ and $\Tdelete_{i}$, respectively; likewise for each of the other $\DSnumber$ data structures $\{\calD'_{\sigma}\}_{\sigma \in [\DSnumber]}$.

\vspace{.1in}
{\bf Case~2:}
All of the $\DSnumber$ subroutines (terminate and) return {\failure}'s.

Then, we output $\Cterminal \gets \calD.C$ as the solution.
\hfill
(Lines~\ref{alg:LS:sample} and \ref{alg:LS:schedule} $\to$ Lines~\ref{alg:LS:negative} and \ref{alg:LS:return})

\vspace{.1in}
(Line~\ref{alg:LS:sample})~The sampling of all noncenters $\{\cinsert_{\sigma}\}_{\sigma \in [\DSnumber]}$ takes time $\DSnumber \cdot O(\log(n))$, by \Cref{lem:sample-noncenter}.

(Line~\ref{alg:LS:schedule} $\to$ Line~\ref{alg:LS:negative})~All subroutines together take time $\DSnumber \cdot \Phi_{\max} \cdot \eps^{-O(1)} \beta \log(n) + \DSnumber \cdot O(\log^{2}(n))$, by a combination of \Cref{claim:test-runtime} \Cref{lem:insert:runtime} from \Cref{lem:insert}.

(Line~\ref{alg:LS:return})~Returning $\Cterminal \gets \calD.C$ as the solution clearly takes time $O(n)$.
\end{flushleft}

\noindent
{\bf The total running time.}
\`{A} la the proof of \Cref{thm:LS-correctness}, denote by $\ell \ge 0$ the number of positive iterations, after which the algorithm {\LocalSearch} returns $\Cterminal \gets \calD.C$ as the solution in the terminal negative iteration.
Combining the above arguments, the total running time is
\begin{align*}
    T
    & ~=~ \underbrace{\DSnumber \cdot O(m \log^{2}(n))}_{\text{Lines~\ref{alg:LS:naive} and \ref{alg:LS:initialize}}} \\
& \phantom{~=~} + \sum_{i \in [\ell]} \big(\underbrace{\DSnumber \cdot O(\log(n))
    + \DSnumber \cdot O(\Tinsert_{i} + O(\log^{2}(n))
    + (\DSnumber + 1) \cdot (\Tinsert_{i} + \Tdelete_{i})}_{\text{Lines~\ref{alg:LS:sample} and \ref{alg:LS:schedule} $\to$ Lines~\ref{alg:LS:positive} to \ref{alg:LS:swap}}}\big) \\
    & \phantom{~=~} + \underbrace{\DSnumber \cdot O(\log(n))
    + \DSnumber \cdot \Phi_{\max} \cdot \eps^{-O(1)} \beta \log(n) + \DSnumber \cdot O(\log^{2}(n))
    + O(n)}_{\text{Lines~\ref{alg:LS:sample} and \ref{alg:LS:schedule} $\to$ Lines~\ref{alg:LS:negative} and \ref{alg:LS:return}}} \\
    & ~=~ \DSnumber \cdot O(m \log^{2}(n))
    ~+~ (\ell + 1) \cdot \DSnumber \cdot O(\log^{2}(n)) \\
    & \phantom{~=~} + O(s) \sum_{i \in [\ell]} (\Tinsert_{i} + \Tdelete_{i})
    ~+~ \DSnumber \cdot \Phi_{\max} \cdot \eps^{-O(1)} \beta \log(n).
\end{align*}

In every positive iteration $i \in [\ell]$, let us denote by $\Phi_{i}$, $\Phi'_{i}$, and $\Phi''_{i}$ (\Cref{fig:DS}) the potential of the data structure $\calD$ before both operations, immediately after the first operation $\calD.\Insert(\cinsert_{\sigma^{*}})$, and immediately after the second operation $\calD.\Delete(\cdelete_{\sigma^{*}})$, respectively; obviously, we have $\Phi''_{i} = \Phi_{i + 1}$, $\forall i \in [\ell - 1]$.
Also, let us simply write $\DSvolume'_{i} \ge 0$ for the deletion volume estimator $\calD.\DSvolume[\cdelete_{\sigma^{*}}]$ of the data structure $\calD$ immediately after the first operation $\calD.\Insert(\cinsert_{\sigma^{*}})$.
We observe that:

\noindent
(i)~Over all positive iterations $i \in [\ell]$, the data structure $\calD$ has total insertion time
\begin{align*}
    \sum_{i \in [\ell]} \Tinsert_{i}
    & ~=~ \sum_{i \in [\ell]} (\Phi_{i} - \Phi'_{i}) \cdot \eps^{-O(1)} \beta \log(n)
    \tag{\Cref{lem:insert:runtime} of \Cref{lem:insert}} \\
    & ~=~ \Big((\Phi_{1} - \Phi''_{\ell}) + \sum_{i \in [\ell]} (\Phi''_{i} - \Phi'_{i})\Big) \cdot \eps^{-O(1)} \beta \log(n)
    \hspace{1.26cm}
\tag{$\Phi''_{i} = \Phi_{i + 1},\ \forall i \in [\ell - 1]$} \\
    & ~=~ \Big(1 + \sum_{i \in [\ell]} \DSvolume'_{i}\Big) \cdot \Phi_{\max} \cdot \eps^{-O(1)} \beta \log(n)
    \tag{\Cref{lem:potential,lem:delete}} \\
    & ~=~ \Phi_{\max} \cdot \eps^{-O(1)} \beta \log^{2}(n).
    \tag{\Cref{eq:total-volume}}
\end{align*}
(ii)~Over all positive iterations $i \in [\ell]$, the data structure $\calD$ has total insertion time
\begin{align*}
    \sum_{i \in [\ell]} \Tdelete_{i}
    & ~=~ \sum_{i \in [\ell]} \DSvolume'_{i} \cdot O(m \log^{3}(n))
    \hspace{4.49cm}
\tag{\Cref{lem:delete:runtime} of \Cref{lem:delete}} \\
    & ~=~ \eps^{-O(1)} m \log^{4}(n).
    \tag{\Cref{eq:total-volume}}
\end{align*}
Combining the above three equations gives
\begin{align*}
    T
    & ~=~ \DSnumber \cdot \Big(O(m \log^{2}(n))
    ~+~ (\ell + 1) \cdot O(\log^{2}(n)) \\
    & \phantom{~=~}\phantom{\DSnumber \cdot \Big(} + \Phi_{\max} \cdot \eps^{-O(1)} \beta \log^{2}(n)
    ~+~ \eps^{-O(1)} m \log^{4}(n)\Big) \\
    & ~=~ \DSnumber \cdot \Big(O(m \log^{2}(n))
    ~+~ \eps^{-O(1)} m \log^{4}(n)
    \tag{\Cref{eq:positive-iteration}} \\
    & \phantom{~=~}\phantom{\DSnumber \cdot \Big(} + \eps^{-O(1)} m \beta \log^{4}(n)
    ~+~ \eps^{-O(1)} m \log^{4}(n)\Big)
    \hspace{1.93cm}
\tag{\Cref{lem:potential}} \\
    & ~=~ \DSnumber \cdot \eps^{-O(1)} m \beta \log^{4}(n).
\end{align*}
This finishes the proof of \Cref{thm:LS-runtime}.
\end{proof}

Combining \Cref{thm:LS-correctness,thm:LS-runtime} gives \Cref{thm:LS} (restated below for ease of reference).

\begin{restate}[\Cref{thm:LS}]
\begin{flushleft}
Provided \Cref{assumption:LS:hop-bounded,assumption:edge-weight}, the randomized algorithm {\LocalSearch} has worst-case running time $\eps^{-O(z)} m \beta \log^{5}(n)$ and, with probability $\ge 1 - n^{-\Theta(1)}$,\\
returns an $\alpha_{z}(\eps)$-approximate feasible solution $\Cterminal \in V^{k}$ to the {\kzC} problem.
\begin{align*}
    \alpha_{z}(\eps)
    ~=~ \min_{\lambda} \bigg\{\,\frac{2^{1 + z + 2\eps z} \cdot (1 + \lambda)^{z - 1} + 2^{(1 + 4\eps) z}}{3 - \eps - 2^{1 + 2\eps z} \cdot ((1 + 1 / \lambda)^{z - 1} + \eps)} \,\biggmid\, \lambda > \frac{1}{((3 - \eps) / 2^{1 + 2\eps z} - \eps)^{1 / (z - 1)} - 1}\,\bigg\}.
\end{align*}
\end{flushleft}
\end{restate}

\begin{remark}[Aspect ratio]
\label{remark:edge-weight}
The {\em aspect ratio} $\Delta = \Delta(G) \ge 1$ of a graph $G = (V, E, w)$ refers to the ratio $\Delta(G) \eqdef \frac{\max_{u \ne v \in V} \dist(u, v)}{\min_{u \ne v \in V} \dist(u, v)}$ of the maximum versus minimum pairwise distances.
Regarding our proof, the aspect ratio $\Delta \ge 1$ essentially only impacts (\Cref{lem:potential}) the potential upper bound $\Phi_{\max} = 2m |\calJ| \cdot \log_{2}(1 + \dmax / \dmin)$.

Given a generic aspect ratio $\Delta \ge 1$ instead of (\Cref{prop:distance}) $\Delta = \dmax / \dmin = n^{O(z)}$, the potential upper bound will be $\Phi_{\max} = O(m \log(n) \log(\Delta))$ instead of $\Phi_{\max} = O(z m \log^{2}(n))$.\footnote{Here and after, we follow the convention of simply writing $\log(\Delta)$ for $\log(1 + \Delta)$.}
Accordingly, the algorithm {\LocalSearch} will have worst-case running time $\eps^{-O(z)} m \beta \log^{4}(n) \log(\Delta)$ instead of $\eps^{-O(z)} m \beta \log^{5}(n)$, while the approximation ratio and the success probability keep the same.
In this manner, unlike \Cref{thm:LS}, we remove \Cref{assumption:LS:hop-bounded} and only impose \Cref{assumption:edge-weight}.
\end{remark}

\section{Applications}
\label{sec:applications}

In this section, we apply our \Cref{thm:LS} to various clustering scenarios. In \Cref{subsec:general-graph}, we address generic graphs (without restrictions on their hop-boundedness and edge weights).
In \Cref{subsec:metric-clustering}, we study canonical families of metric spaces, aiming to incorporate known results about metric spanner constructions into our local search algorithm.

\subsection{Clustering on general graphs}
\label{subsec:general-graph}

The following \Cref{cor:general-graph} strengthens \Cref{thm:LS}, by removing (\Cref{assumption:LS:hop-bounded,assumption:edge-weight}) the restrictions on hop-boundedness and edge weights.

\begin{corollary}[Clustering on general graphs]
\label{cor:general-graph}
\begin{flushleft}
Given constants $z \ge 1$ and $\delta > 0$, there is a randomized $(\alpha_{z}^{*} + \delta)$-approximation {\kzC} algorithm that has worst-case running time $m \cdot 2^{O(\sqrt{\log n \log\log n})}$ and success probability $\ge 1 - n^{-\Theta(1)}$.
\begin{align*}
    \alpha_{z}^{*}
    ~\eqdef~ \min_{\lambda} \bigg\{\,\frac{2^{1 + z} \cdot (1 + \lambda)^{z - 1} + 2^{z}}{3 - 2 \cdot (1 + 1 / \lambda)^{z - 1}} \,\biggmid\, \lambda > \frac{1}{(3 / 2)^{1 / (z - 1)} - 1}\,\bigg\}.
\end{align*}
\end{flushleft}
\end{corollary}

\begin{proof}
Basically, \Cref{cor:general-graph} follows from a combination of \Cref{thm:LS}, \Cref{prop:edge-weight}, and the following \Cref{prop:EN19}, which summarizes the hopset construction by \cite[Theorem~3.8]{ElkinN19}.\footnote{Specifically, \Cref{prop:EN19} follows by setting $\kappa = \sqrt{\log n}$ and $\rho = \sqrt{\frac{\log\log n}{2\log n}}$ for the parameters $\kappa$ and $\rho$ in the statement of \cite[Theorem~3.8]{ElkinN19}.}

\begin{proposition}[{\cite[Theorem~3.8]{ElkinN19}}]
\label{prop:EN19}
\begin{flushleft}
For any parameter $\eps \ge \log^{-\Theta(1)}(n)$,\textsuperscript{\textnormal{\ref{footnote:exponent}}},
there exists an $m \cdot 2^{O(\sqrt{\log n \log\log n})}$-time algorithm that, with constant probability, converts a given graph $G = (V,E, w)$ to a new graph $G'=(V,E',w')$ such that
\begin{enumerate}[font = {\em\bfseries}]
    \item \label{prop:EN19:solution}
    $\dist_{G'}(u, v) = \dist_{G}(u, v)$, for any pair of vertices $u, v \in V$.
    
    \item \label{prop:EN19:hop} $G'$ is a $(\beta, \eps)$-hop-bounded graph, where the parameter $\beta = 2^{O(\sqrt{\log n \log\log n})}$.
    
    \item \label{prop:EN19:edge} $|E'| \le m + n \cdot 2^{O(\sqrt{\log n})}$.
\end{enumerate}
\end{flushleft}
\end{proposition}

Our algorithm for \Cref{cor:general-graph} takes four steps, as follows.

Firstly, we identify a suitable constant $0 < \eps' < \frac{1}{10z}$.
The function $\alpha_{z}(\eps)$ given in the statement of \Cref{thm:LS} clearly is continuous on the range $0 \le \eps \le \frac{1}{10z}$ and its left endpoint $\alpha_{z}(0) = \alpha_{z}^{*}$.
Thus, there must exist a small enough constant $0 < \eps' < \frac{1}{10z}$ such that $2^{\eps'} \cdot \alpha_{z}(\eps') \le \alpha_{z}^{*} + \delta$.\\
\Comment{Finding such a constant $\eps'$ clearly takes time $O(1)$.}

Secondly, based on the above constant $\eps'$, we utilize \Cref{prop:edge-weight} to convert the given graph $G = (V, E, w)$ into an interim graph $G' = (V, E, w')$ that satisfies \Cref{assumption:edge-weight}.
Regarding the {\kzC} problem, (\Cref{prop:edge-weight:solution} of \Cref{prop:edge-weight}) any $\alpha_{z}(\eps')$-approximate solution $C \in V^{k}$ to (the shortest-path metric induced by) the interim graph $G' = (V, E, w')$ is a $2^{\eps'}\alpha_z(\eps')$-approximate solution to (the shortest-path metric induced by) the input graph $G = (V, E, w)$.\\
\Comment{(\Cref{prop:edge-weight}) This step runs in time $O(m \log(n))$.}

Thirdly, based on the above constant $\eps'$, we utilize \Cref{prop:EN19} to convert the interim graph $G' = (V, E, w')$ into a $(\beta, \eps')$-hop-bounded graph $G'' = (V, E'', w'')$, with $\beta = 2^{O(\sqrt{\log n \log \log n})}$ and $|E''| \le m + n \cdot 2^{O(\sqrt{\log n})}$, that satisfies both \Cref{assumption:edge-weight,assumption:DS:hop-bounded}.
Regarding the {\kzC} problem, (\Cref{prop:EN19:solution} of \Cref{prop:EN19}) any $\alpha_{z}(\eps')$-approximate solution $C \in V^{k}$ to this $(\beta, \eps')$-hop-bounded graph $G'' = (V, E'', w'')$ is also an $\alpha_{z}(\eps')$-approximate solution to the interim graph $G' = (V, E, w')$, thus a $2^{\eps'}\alpha_z(\eps')$-approximate solution to the given graph $G = (V, E, w)$.\\
\Comment{(\Cref{prop:EN19}) This step runs in time $m \cdot 2^{O(\sqrt{\log n \log \log n})}$ and succeeds with constant probability.}

Fourthly, based on the above constant $\eps'$, we run the algorithm for \Cref{thm:LS} on the $(\beta, \eps')$-hop-bounded graph $G'' = (V, E'', w'')$, aiming to obtain an $\alpha_{z}(\eps')$-approximate solution $C \in V^{k}$ to the {\kzC} problem.\\
\Comment{(\Cref{thm:LS}) This step runs in time $O(|E''|\beta\log^5(n)) = m \cdot 2^{O(\sqrt{\log n \log\log n})}$ and succeeds with probability $\ge 1 - n^{-\Theta(1)}$.}

Overall, the above algorithm has worst-case running time $m \cdot 2^{O(\sqrt{\log n \log\log n})}$ and, with constant probability, obtains a $2^{\eps'} \cdot \alpha_{z}(\eps') \le \alpha_{z}^{*} + \delta$-approximate solution $C \in V^{k}$ to the {\kzC} problem on the given graph $G = (V, E, w)$.
It is easy to boost the success probability to $1 - n^{-\Theta(1)}$, by repeating $\Theta(\log(n))$ times and utilizing Dijkstra's algorithm to find the best among the $\Theta(\log(n))$ candidate solutions.
This finishes the proof of \Cref{cor:general-graph}.
\end{proof}

\begin{remark}[Clustering on general graphs]
    \label{remark:general-graph}
    Specifically, the ratio for {\kMedian} ($z = 1$) is $\alpha_{1}^{*} = 6$ and the ratio for {\kMeans} ($z = 2$) is $\alpha_{2}^{*} = \min_{\lambda} \big\{8(\lambda - 2) + \frac{56}{\lambda - 2} + 44 \bigmid \lambda > 2\big\} = 44 + 16\sqrt{7} \approx 86.33$.
\end{remark}

\subsection{Clustering in canonical families of metric spaces}
\label{subsec:metric-clustering}

We further apply \Cref{cor:general-graph} (and \Cref{thm:LS}) to the metric clustering problem in various canonical families of metric spaces.
In metric clustering problem, we allow algorithms to access the given metric space $(V, \dist)$ through queries to a distance oracle, which returns $\dist(u, v)$ for any pair of points $u, v \in V$ in constant time. This differs from the graph clustering setting, where we do not assume the presence of a distance oracle. 
Moreover, our goal is to cluster a given dataset $X \subseteq V$ with $n$ points, rather than considering the entire space $V$ as in shortest-path metrics.

We apply our algorithms to metric clustering by running them on top of {\em metric spanners}.
We provide a formal definition of metric spanners below.

\begin{definition}[Metric spanner] \label{def:spanner}
    Given a stretch parameter $t \geq 1$, a (metric) $t$-spanner for a dataset $X \subseteq V$ from a metric space $(V, \dist)$ is a weighted undirected graph $H = (X, E, w)$, where $E \subseteq X \times X$ and the edge weights are defined as $w(x, y) \eqdef \dist(x, y)$ for $(x, y) \in E$, such that the shortest-path distance function $\dist_H$ satisfies:
    \begin{equation*}
    \forall x,y\in X, \quad \dist_H(x, y) \leq t \cdot \dist(x, y).
    \end{equation*}
\end{definition}

The defining condition of \Cref{def:spanner} ensures that any $\alpha$-approximation to {\kzC} on a $t$-spanner $H$ for $X$ will be an $\alpha \cdot t^z$-approximation in the original metric space. 
Next, we summarize our results in various metric spaces by applying \Cref{cor:general-graph} (or \Cref{thm:LS}) to the existing metric spanners.

\vspace{.1in}
\noindent
{\bf Metric spaces that admit LSH.}
We apply \Cref{cor:general-graph} to a wide range of metric spaces, specifically those that admit Locality-Sensitive Hashing.
Notably, this includes (high-dimensional) Euclidean spaces, general $\ell_p$ spaces, and Jaccard metrics.
Below, we give the formal definition of the LSH family.

\begin{definition}[LSH families \cite{IndykM98, Har-PeledIM12}]
\label{def:lsh}
\begin{flushleft}
For a metric space $(V, \dist)$, a family $\calH$ of hash functions $h \colon V \to U$ is called {\em $(r, c r, p_{1}, p_{2})$-sensitive} when, for a uniform random hash function $h' \sim \calH$ and any pair of points $u, v \in V$:
\begin{itemize}
    \item $\Pr_{h' \sim \calH}[h'(u) = h'(v)] \ge p_{1}$ if $\dist(u,v) \le r$.
    \item $\Pr_{h' \sim \calH}[h'(u) = h'(v)] \le p_{2}$ if $\dist(u,v) \ge c r$.
\end{itemize}
\end{flushleft}
\end{definition}

We employ the LSH-based spanner construction from~\cite{HIS13}, which was originally designed for Euclidean spaces but can be easily generalized to any metric space that admits LSH.
We formalize this generalized construction in \Cref{thm:spanner-lsh} and, for completeness, provide a proof in \Cref{sec:proof-spanner-lsh}.

\begin{proposition}[{Spanners via LSH \cite[Theorem 3.1]{HIS13}}]
\label{thm:spanner-lsh}
\begin{flushleft}
Fix any $c\ge 1$. 
Suppose that for any $r > 0$, the underlying metric space $(V, \dist)$ admits an $(r, c r, p_{1}, p_{2})$-sensitive family of hash functions such that 
\begin{enumerate}[font = {\em\bfseries}]
    \item \label{thm:spanner-lsh:condition1} $1 > p_{1} > p_{2} > 0$, $p_1^{-1} = o(n)$ and $\rho \eqdef \frac{\log(1 / p_{1})}{\log(1 / p_{2})} < 1/3$; and
    \item \label{thm:spanner-lsh:condition2} each function in this family can be sampled and evaluated in time $\TLSH$.
\end{enumerate}
There exists a randomized algorithm that, given as input an $n$-point dataset $X \subseteq V$ with aspect ratio $\Delta \ge 1$, runs in time
\begin{equation*}
    O\left(\frac{n^{1+3\rho}\log^2(n)\log(\Delta)}{ p_1\log(1/p_2)}\cdot \TLSH\right)
\end{equation*}
to compute an $O(c)$-spanner for $X$ with 
\begin{equation*}
    O\left(\frac{n^{1 + 3\rho}\log(n)\log(\Delta)}{ p_1}\right)
\end{equation*}
edges.
The algorithm succeeds with probability $\ge 1 - \frac{1}{n^{\Theta(1)}}$.

\end{flushleft}
\end{proposition}
\begin{proof}
The proof can be found in \Cref{sec:proof-spanner-lsh}.
\end{proof}

The parameter $\rho = \frac{\log(1 / p_{1})}{\log(1 / p_{2})}$ is a key measure of LSH performance, and optimizing this parameter is one of the main research directions in LSH.
In \Cref{thm:spanner-lsh}, the factors $p_{1}^{-1}$ and $(\log(1 / p_{2}))^{-1}$ appear in the running time and spanner size, which is common in the LSH literature; similar dependencies are observed in the query/space complexity of approximate nearest neighbor algorithms based on LSH (see~\cite[Theorem 3.4 \& Remark 3.5]{Har-PeledIM12}).

The following corollary directly follows from applying \Cref{cor:general-graph} to the spanner constructed by \Cref{thm:spanner-lsh}, which is a fast clustering algorithm for metric spaces that admit LSH.

\begin{corollary}[Clustering on metrics that admit LSH]
\begin{flushleft}
    \label{cor:metric-lsh}
    Fix any $c \geq 1$. Suppose the underlying metric space $(V, \dist)$ admits an $(r, cr, p_1, p_2)$-sensitive family of hash functions that satisfies Conditions~\ref{thm:spanner-lsh:condition1} and \ref{thm:spanner-lsh:condition2} of \Cref{thm:spanner-lsh} for any $r > 0$.
    Given a constant $z\ge 1$, there is a randomized $O(c^z)$-approximation {\kzC} algorithm that has worst-case running time 
    \begin{equation*}
        O\left(\frac{n^{1+3\rho + o(1)}\log(\Delta)}{ p_1\log(1/p_2)}\cdot \TLSH \right)
    \end{equation*}
    and success probability $\ge 1 - n^{-\Theta(1)}$. Here, $\rho$ and $\TLSH$ are defined as in \Cref{thm:spanner-lsh}, and the factor $n^{o(1)}$ in the running time is $2^{O(\sqrt{\log n\log\log n})}$.
\end{flushleft}
\end{corollary}
\begin{proof}
    Our algorithm first constructs an $O(c)$-spanner $G$ using \Cref{thm:spanner-lsh}, which contains $m \eqdef O\left(\frac{n^{1 + 3\rho}\log(\Delta)\log(n)}{ p_1}\right)$ edges, in time $O\left(\frac{n^{1+3\rho}\log^2(n)\log(\Delta)}{ p_1\log(1/p_2)}\cdot \TLSH\right)$.
    Then, it runs the algorithm from \Cref{cor:general-graph} on $G$ with parameter $\delta = 1$.
    The overall running time is
    \begin{equation*}
        O\left(\frac{n^{1+3\rho}\log^2(n)\log(\Delta)}{ p_1\log(1/p_2)}\cdot \TLSH\right) + m\cdot 2^{O\left(\sqrt{\log n \log\log n}\right)} ~\le~ O\left(\frac{n^{1+3\rho + o(1)}\log(\Delta)}{ p_1\log(1/p_2)}\cdot \TLSH \right), 
    \end{equation*}
    where the factor $n^{o(1)} = 2^{O(\sqrt{\log n\log\log n})}$.

    Let $C$ denote the solution returned by \Cref{cor:general-graph}, and let $C^{\star}$ denote the optimal $k$-point center set selected from the dataset $X$, i.e., $C^{\star} \eqdef \argmin_{C \in X^k} \cost(X, C)$, in the original metric space. 
    It is well-known that $C^{\star}$ is an $O(2^z)$-approximation to the optimal center set in the metric space $(V,\dist)$, i.e., $\cost(X,C^{\star}) \le O(2^z)\cdot \OPT = O(2^z) \cdot \min_{C \in V^k} \cost(X,C)$.
    The spanner guarantee (\Cref{def:spanner}) ensures that 
    \begin{equation*}
        \cost(X, C) ~\le~ \cost^{(G)}(X, C),\quad \cost^{(G)}(X, C^{\star}) \le O(c^z)\cdot \cost(X, C^{\star}) \le O(c^z)\cdot \OPT,
    \end{equation*}
    where we are using $\cost^{(G)}(X,C):=\sum_{x\in X} \dist_G^z(x, C)$ to denote the clustering objective on graph $G$. Since the solution $C$ satisfies $\cost^{(G)}(X, C) \le O(2^z)\cdot \cost^{(G)}(X, C^{\star})$ by \Cref{cor:general-graph}, we conclude that $C$ is an $O(c^z)$-approximate solution in the original metric space.
\end{proof}

\begin{remark}[Hop-diameter]
\label{remark:hop-diameter}
The work of \cite{HIS13} constructs metric spanners with a {\em hop-diameter} of $2$. However, their notion of hop-diameter differs from the hop-boundedness we consider (\Cref{def:hop_bounded}).
Specifically, their result states that for any pair of points, there exists a path with at most $2$ edges that serves as a $t$-spanner path, i.e., a path whose length approximates the distance in the {\em original metric space} within a factor of $t$.
In contrast, our notion of $(\beta, \epsilon)$-hop-boundedness requires that for any pair of points, there exists a path with at most $\beta$ edges that approximates the {\em shortest-path distance in the spanner itself} within a factor of $(1 + \epsilon)$.

\end{remark}

\vspace{.1in}
\noindent
{\bf Implications to various metric spaces.}
LSH schemes are known to exist for various metric spaces, including Euclidean spaces \cite{DatarIIM04, AndoniI06},
$\ell_{p}$ spaces \cite{IndykM98, DatarIIM04, AndoniI06Efficient, Har-PeledIM12}, and
Jaccard metric spaces~\cite{Broder97, BroderGMZ97};
see also the survey \cite{AndoniI08}.
By incorporating the corresponding LSH results into \Cref{cor:metric-lsh}, we obtain fast clustering algorithms for all these metric spaces, as summarized below. 
We postpone the proof to \Cref{sec:proof-cor-lsh}, which includes a discussion on the choice of LSH and the specification of LSH parameters (e.g., $p_1$, $\TLSH$) for different metric spaces.

\begin{corollary}[Clustering on various metric spaces] \label{cor:lsh}
\begin{flushleft}
    Given constants $c\ge 1$ and $z\ge 1$, there is a randomized $O(c^z)$-approximation algorithm for {\kzC} that has success probability $\ge 1 - n^{-\Theta(1)}$ and worst-case running time:
    \begin{itemize}
        \item {\bf (Euclidean space)} $O(dn^{1 + 1 / c^2 + o(1)}\log(\Delta))$ if the underlying metric space $(V,\dist) = (\R^d, \ell_2)$ is a $d$-dimensional Euclidean space;
        
        \item {\bf ($\ell_p$ metric)} $O(dn^{1 + 1 / c + o(1)}\log(\Delta))$ if the underlying metric space $(V,\dist) = (\R^d, \ell_p)$ for constant $p\in [1,2)$;
        
        \item {\bf (Jaccard metric)} $O(n^{1 + 1 / c + o(1)}\log(\Delta)\cdot |U|^2)$ if the underlying metric space is a Jaccard metric $(2^U, \dist)$, where $U$ is some universe and $\dist$ is the Jaccard distance.
    \end{itemize}
\end{flushleft}
\end{corollary}
\begin{proof}
The proof can be found in \Cref{sec:proof-cor-lsh}.
\end{proof}

\noindent
{\bf Metric spaces with low doubling dimensions.}
We then consider a special class of metric spaces with low {\em doubling dimensions} \cite{GuptaKL03}, which is an important generalization of Euclidean space that also includes other well-known metrics, such as $\ell_{p}$ spaces. 
The doubling dimension of a metric space is the smallest integer $\ddim \ge 1$, such that every ball can be covered by at most $2^{\ddim}$ balls of half the radius.
For example, a $d$-dimensional Euclidean space has a doubling dimension $\ddim = \Theta(d)$.

For metric spaces with low doubling dimensions, spanner constructions have been extensively studied, and they achieve different parameter tradeoffs; see \cite{GaoGN06, Har-PeledM06, GottliebR08, ChanG09, ChanLN15, ElkinS15, Solomon14, ChanLNS15, ChanGMZ16, BorradaileLW19, KahalonLMS22, LeST23} and the references therein.
Among these known results,
we choose to use the construction from~\cite{Solomon14} which achieves a suitable parameter tradeoff for our application, restated below.\footnote{In fact, \cite{Solomon14} considers a more powerful $k$-fault-tolerant spanner, but for our purposes, a non-fault-tolerant version of their result (i.e., $k = O(1)$) is sufficient.
}

\begin{proposition}[{Spanners in doubling metrics \cite[Theorem~1.1]{Solomon14}}]
\begin{flushleft}
\label{prop:Solomon14}
Given a metric space $(V, \dist)$ with a doubling dimension $\ddim \ge 1$,
for any parameter $0 < \eps < 1$ and every $n$-point dataset $X \subseteq V$ with an aspect ratio $\Delta(X) \ge 1$,
an $(O(\log(n)), \eps)$-hop-bounded $(1+\eps)$-spanner $H$ with $\eps^{-O(\ddim)} n$ edges and an aspect ratio of $\Delta(H) \le (1+\eps) \cdot \Delta(X)$ can be found in time $\eps^{-O(\ddim)} n \log(n)$. 
\end{flushleft}
\end{proposition}

Here, we apply \Cref{thm:LS} (instead of \Cref{cor:general-graph}) to the spanner constructed using \Cref{prop:Solomon14}, as this spanner inherently satisfies the hop-boundedness assumption (\Cref{assumption:LS:hop-bounded}). Therefore, applying \Cref{thm:LS} yields a running time that is better by a factor of $n^{o(1)}$ compared to applying \Cref{cor:general-graph}.

\begin{corollary}[Clustering in doubling metrics]
\label{cor:doubling}
\begin{flushleft}
Given constants $z \ge 1$, $\delta > 0$ and a metric space $(V, \dist)$ with a doubling dimension $\ddim \ge 1$, 
there is a randomized algorithm that computes an $(\alpha_{z}^{*} + \delta)$-approximate solution to {\kzC} for every $n$-point dataset $X \subseteq V$ with an aspect ratio $\Delta = \Delta(X) \ge 1$, which has a worst-case running time of $O(2^{O(\ddim)} n \log^{8}(n))$ and a success probability of $1 - n^{-\Theta(1)}$.

\end{flushleft}
\end{corollary}

\begin{proof}
Similar to the proof of \Cref{cor:general-graph}, we can find a constant $0 < \eps' < \frac{1}{10z}$ such that $(1+\eps')^{z}\cdot\alpha_{z}(\eps')\le \alpha_{z}^{*} + \delta$ in constant time $O(1)$. Then, we apply \Cref{prop:Solomon14} with input $\eps'$ and the dataset $X$ to compute a $(1+\eps')$-spanner $H$ with $m \eqdef \eps'^{-O(\ddim)} n = 2^{O(\ddim)}n$ edges and an aspect ratio $\Delta(H) \le O(\Delta)$, in time $\eps'^{-O(\ddim)}n\log(n) = 2^{O(\ddim)}n\log(n)$.

    This spanner $H$ is guaranteed to be $(\beta,\eps')$-hop-bounded with $\beta = O(\log(n))$, allowing us to run \Cref{thm:LS} on it to compute an $\alpha_{z}(\eps')$-approximate solution to the \kzC problem on $H$. 
The nature of the spanner $H$ (\Cref{def:spanner}) ensures that the final solution is a $((1+\eps')^z \cdot \alpha_{z}(\eps') \le \alpha_{z}^{*} + \delta)$-approximate solution to the \kzC problem on the dataset $X$ in the original metric space $(V, \dist)$. The overall running time is 
    \begin{align*}
        \eps'^{-O(\ddim)}n\log(n) + \eps'^{-O(z)}m\beta \log^4(n)\log(\Delta(H)) = 2^{O(\ddim)}n\log^{5}(n)\log(\Delta),
    \end{align*}
    according to \Cref{remark:edge-weight}. Finally, we can use \cite[Lemmas~32 and 33]{Cohen-AddadFS21} to transform the above algorithm into one that achieves the same approximation but replaces the $\log(\Delta)$ factor with $\log^3(n)$ in the running time, resulting in the time complexity shown in \Cref{cor:doubling}.
\end{proof}

\ifanonym
\else
\section*{Acknowledgements}
We are grateful to Hengjie Zhang for helpful discussions on Locality Sensitive Hashing.
\fi

\bibliographystyle{alphaurl}
\bibliography{ref}

\newcommand{\etalchar}[1]{$^{#1}$}
\begin{thebibliography}{GMMO00}

\bibitem[ABB{\etalchar{+}}23]{DBLP:conf/focs/AbbasiBBCGKMSS23}
Fateme Abbasi, Sandip Banerjee, Jaroslaw Byrka, Parinya Chalermsook, Ameet Gadekar, Kamyar Khodamoradi, D{\'{a}}niel Marx, Roohani Sharma, and Joachim Spoerhase.
\newblock Parameterized approximation schemes for clustering with general norm objectives.
\newblock In {\em {FOCS}}, pages 1377--1399. {IEEE}, 2023.

\bibitem[ACKS15]{AwasthiCKS15}
Pranjal Awasthi, Moses Charikar, Ravishankar Krishnaswamy, and Ali~Kemal Sinop.
\newblock The hardness of approximation of euclidean k-means.
\newblock In {\em SoCG}, volume~34 of {\em LIPIcs}, pages 754--767. Schloss Dagstuhl - Leibniz-Zentrum f{\"{u}}r Informatik, 2015.

\bibitem[ACW16]{DBLP:conf/focs/AlmanCW16}
Josh Alman, Timothy~M. Chan, and R.~Ryan Williams.
\newblock Polynomial representations of threshold functions and algorithmic applications.
\newblock In {\em {FOCS}}, pages 467--476. {IEEE} Computer Society, 2016.

\bibitem[AGK{\etalchar{+}}04]{AryaGKMMP04}
Vijay Arya, Naveen Garg, Rohit Khandekar, Adam Meyerson, Kamesh Munagala, and Vinayaka Pandit.
\newblock Local search heuristics for k-median and facility location problems.
\newblock {\em {SIAM} J. Comput.}, 33(3):544--562, 2004.

\bibitem[AI06a]{AndoniI06Efficient}
Alexandr Andoni and Piotr Indyk.
\newblock Efficient algorithms for substring near neighbor problem.
\newblock In {\em {SODA}}, pages 1203--1212. {ACM} Press, 2006.

\bibitem[AI06b]{AndoniI06}
Alexandr Andoni and Piotr Indyk.
\newblock Near-optimal hashing algorithms for approximate nearest neighbor in high dimensions.
\newblock In {\em {FOCS}}, pages 459--468. {IEEE} Computer Society, 2006.

\bibitem[AI08]{AndoniI08}
Alexandr Andoni and Piotr Indyk.
\newblock Near-optimal hashing algorithms for approximate nearest neighbor in high dimensions.
\newblock {\em Commun. {ACM}}, 51(1):117--122, 2008.

\bibitem[Ali96]{Alimonti96}
Paola Alimonti.
\newblock New local search approximation techniques for maximum generalized satisfiability problems.
\newblock {\em Inf. Process. Lett.}, 57(3):151--158, 1996.

\bibitem[ANSW20]{ANSW20}
Sara Ahmadian, Ashkan Norouzi{-}Fard, Ola Svensson, and Justin Ward.
\newblock Better guarantees for k-means and euclidean k-median by primal-dual algorithms.
\newblock {\em {SIAM} J. Comput.}, 49(4), 2020.

\bibitem[AR15]{AndoniR15}
Alexandr Andoni and Ilya~P. Razenshteyn.
\newblock Optimal data-dependent hashing for approximate near neighbors.
\newblock In {\em {STOC}}, pages 793--801. {ACM}, 2015.

\bibitem[ASZ20]{AndoniSZ20}
Alexandr Andoni, Clifford Stein, and Peilin Zhong.
\newblock Parallel approximate undirected shortest paths via low hop emulators.
\newblock In {\em {STOC}}, pages 322--335. {ACM}, 2020.

\bibitem[AV07]{ArthurV07}
David Arthur and Sergei Vassilvitskii.
\newblock k-means++: the advantages of careful seeding.
\newblock In {\em {SODA}}, pages 1027--1035. {SIAM}, 2007.

\bibitem[BCF25]{BhattacharyaCF24}
Sayan Bhattacharya, Mart{\'{\i}}n Costa, and Ermiya Farokhnejad.
\newblock Fully dynamic \emph{k}-median with near-optimal update time and recourse.
\newblock In {\em {STOC}}. {ACM}, 2025.
\newblock To appear.

\bibitem[BCG{\etalchar{+}}24]{BhattacharyaCGLP24}
Sayan Bhattacharya, Mart{\'{\i}}n Costa, Naveen Garg, Silvio Lattanzi, and Nikos Parotsidis.
\newblock Fully dynamic k-clustering with fast update time and small recourse.
\newblock In {\em {FOCS}}, pages 216--227. {IEEE}, 2024.

\bibitem[BCLP23]{0001CLP23}
Lorenzo Beretta, Vincent Cohen{-}Addad, Silvio Lattanzi, and Nikos Parotsidis.
\newblock Multi-swap k-means++.
\newblock In {\em NeurIPS}, 2023.

\bibitem[BCP{\etalchar{+}}24]{0001CPSS24}
Nikhil Bansal, Vincent Cohen{-}Addad, Milind Prabhu, David Saulpic, and Chris Schwiegelshohn.
\newblock Sensitivity sampling for k-means: Worst case and stability optimal coreset bounds.
\newblock In {\em {FOCS}}, pages 1707--1723. {IEEE}, 2024.

\bibitem[BEF{\etalchar{+}}23]{BateniEFHJMW23}
MohammadHossein Bateni, Hossein Esfandiari, Hendrik Fichtenberger, Monika Henzinger, Rajesh Jayaram, Vahab Mirrokni, and Andreas Wiese.
\newblock Optimal fully dynamic \emph{k}-center clustering for adaptive and oblivious adversaries.
\newblock In {\em {SODA}}, pages 2677--2727. {SIAM}, 2023.

\bibitem[Ber09]{Bernstein09}
Aaron Bernstein.
\newblock Fully dynamic {(2} + epsilon) approximate all-pairs shortest paths with fast query and close to linear update time.
\newblock In {\em {FOCS}}, pages 693--702. {IEEE} Computer Society, 2009.

\bibitem[BGMZ97]{BroderGMZ97}
Andrei~Z. Broder, Steven~C. Glassman, Mark~S. Manasse, and Geoffrey Zweig.
\newblock Syntactic clustering of the web.
\newblock {\em Comput. Networks}, 29(8-13):1157--1166, 1997.

\bibitem[BLW19]{BorradaileLW19}
Glencora Borradaile, Hung Le, and Christian Wulff{-}Nilsen.
\newblock Greedy spanners are optimal in doubling metrics.
\newblock In {\em {SODA}}, pages 2371--2379. {SIAM}, 2019.

\bibitem[BPR{\etalchar{+}}17]{ByrkaPRST17}
Jaroslaw Byrka, Thomas~W. Pensyl, Bartosz Rybicki, Aravind Srinivasan, and Khoa Trinh.
\newblock An improved approximation for \emph{k}-median and positive correlation in budgeted optimization.
\newblock {\em {ACM} Trans. Algorithms}, 13(2):23:1--23:31, 2017.

\bibitem[Bro97]{Broder97}
Andrei~Z. Broder.
\newblock On the resemblance and containment of documents.
\newblock In {\em {SEQUENCES}}, pages 21--29. {IEEE}, 1997.

\bibitem[CDR{\etalchar{+}}25]{Cohen-AddadD0SS25}
Vincent Cohen{-}Addad, Andrew Draganov, Matteo Russo, David Saulpic, and Chris Schwiegelshohn.
\newblock A tight vc-dimension analysis of clustering coresets with applications.
\newblock In {\em {SODA}}, pages 4783--4808. {SIAM}, 2025.

\bibitem[CEMN22]{Cohen-AddadEMN22}
Vincent Cohen{-}Addad, Hossein Esfandiari, Vahab~S. Mirrokni, and Shyam Narayanan.
\newblock Improved approximations for euclidean \emph{k}-means and \emph{k}-median, via nested quasi-independent sets.
\newblock In {\em {STOC}}, pages 1621--1628. {ACM}, 2022.

\bibitem[CFS21]{Cohen-AddadFS21}
Vincent Cohen{-}Addad, Andreas~Emil Feldmann, and David Saulpic.
\newblock Near-linear time approximation schemes for clustering in doubling metrics.
\newblock {\em J. {ACM}}, 68(6):44:1--44:34, 2021.

\bibitem[CG99]{CharikarG99}
Moses Charikar and Sudipto Guha.
\newblock Improved combinatorial algorithms for the facility location and k-median problems.
\newblock In {\em {FOCS}}, pages 378--388. {IEEE} Computer Society, 1999.

\bibitem[CG09]{ChanG09}
T.{-}H.~Hubert Chan and Anupam Gupta.
\newblock Small hop-diameter sparse spanners for doubling metrics.
\newblock {\em Discret. Comput. Geom.}, 41(1):28--44, 2009.

\bibitem[CGH{\etalchar{+}}22]{Cohen-AddadGHOS22}
Vincent Cohen{-}Addad, Anupam Gupta, Lunjia Hu, Hoon Oh, and David Saulpic.
\newblock An improved local search algorithm for k-median.
\newblock In {\em {SODA}}, pages 1556--1612. {SIAM}, 2022.

\bibitem[CGL{\etalchar{+}}25]{CGLSO25}
Vincent Cohen{-}Addad, Fabrizio Grandoni, Euiwoong Lee, Chris Schwiegelshohn, and Ola Svensson.
\newblock A $(2+\varepsilon)$-approximation algorithm for metric $k$-median.
\newblock In {\em {STOC}}. {ACM}, 2025.
\newblock To appear.

\bibitem[CGLS23]{Cohen-Addad0LS23}
Vincent Cohen{-}Addad, Fabrizio Grandoni, Euiwoong Lee, and Chris Schwiegelshohn.
\newblock Breaching the 2 {LMP} approximation barrier for facility location with applications to \emph{k}-median.
\newblock In {\em {SODA}}, pages 940--986. {SIAM}, 2023.

\bibitem[CGMZ16]{ChanGMZ16}
T.{-}H.~Hubert Chan, Anupam Gupta, Bruce~M. Maggs, and Shuheng Zhou.
\newblock On hierarchical routing in doubling metrics.
\newblock {\em {ACM} Trans. Algorithms}, 12(4):55:1--55:22, 2016.

\bibitem[CGPR20]{ChooGPR20}
Davin Choo, Christoph Grunau, Julian Portmann, and V{\'{a}}clav Rozhon.
\newblock k-means++: few more steps yield constant approximation.
\newblock In {\em {ICML}}, volume 119 of {\em Proceedings of Machine Learning Research}, pages 1909--1917. {PMLR}, 2020.

\bibitem[CGTS99]{CharikarGTS99}
Moses Charikar, Sudipto Guha, {\'{E}}va Tardos, and David~B. Shmoys.
\newblock A constant-factor approximation algorithm for the \emph{k}-median problem (extended abstract).
\newblock In {\em {STOC}}, pages 1--10. {ACM}, 1999.

\bibitem[CGTS02]{CharikarGTS02}
Moses Charikar, Sudipto Guha, {\'{E}}va Tardos, and David~B. Shmoys.
\newblock A constant-factor approximation algorithm for the k-median problem.
\newblock {\em J. Comput. Syst. Sci.}, 65(1):129--149, 2002.

\bibitem[Che09]{DBLP:journals/siamcomp/Chen09}
Ke~Chen.
\newblock On coresets for k-median and k-means clustering in metric and euclidean spaces and their applications.
\newblock {\em {SIAM} J. Comput.}, 39(3):923--947, 2009.

\bibitem[CHH{\etalchar{+}}23]{CharikarHHVW23}
Moses Charikar, Monika Henzinger, Lunjia Hu, Maximilian V{\"{o}}tsch, and Erik Waingarten.
\newblock Simple, scalable and effective clustering via one-dimensional projections.
\newblock In {\em NeurIPS}, 2023.

\bibitem[CK19]{Cohen-AddadS19}
Vincent Cohen{-}Addad and {Karthik {C. S.}}
\newblock Inapproximability of clustering in lp metrics.
\newblock In {\em {FOCS}}, pages 519--539. {IEEE} Computer Society, 2019.

\bibitem[CKL21]{Cohen-AddadSL21}
Vincent Cohen{-}Addad, {Karthik {C. S.}}, and Euiwoong Lee.
\newblock On approximability of clustering problems without candidate centers.
\newblock In D{\'{a}}niel Marx, editor, {\em Proceedings of the 2021 {ACM-SIAM} Symposium on Discrete Algorithms, {SODA} 2021, Virtual Conference, January 10 - 13, 2021}, pages 2635--2648. {SIAM}, 2021.
\newblock \href {https://doi.org/10.1137/1.9781611976465.156} {\path{doi:10.1137/1.9781611976465.156}}.

\bibitem[CKL22]{Cohen-AddadSL22}
Vincent Cohen{-}Addad, {Karthik {C. S.}}, and Euiwoong Lee.
\newblock Johnson coverage hypothesis: Inapproximability of k-means and k-median in {\(\mathscr{l}\)}\({}_{\mbox{p}}\)-metrics.
\newblock In {\em {SODA}}, pages 1493--1530. {SIAM}, 2022.

\bibitem[CKM19]{Cohen-AddadKM19}
Vincent Cohen{-}Addad, Philip~N. Klein, and Claire Mathieu.
\newblock Local search yields approximation schemes for k-means and k-median in euclidean and minor-free metrics.
\newblock {\em {SIAM} J. Comput.}, 48(2):644--667, 2019.

\bibitem[CL12]{CharikarL12}
Moses Charikar and Shi Li.
\newblock A dependent lp-rounding approach for the k-median problem.
\newblock In {\em {ICALP} {(1)}}, volume 7391 of {\em Lecture Notes in Computer Science}, pages 194--205. Springer, 2012.

\bibitem[CLN15]{ChanLN15}
T.{-}H.~Hubert Chan, Mingfei Li, and Li~Ning.
\newblock Sparse fault-tolerant spanners for doubling metrics with bounded hop-diameter or degree.
\newblock {\em Algorithmica}, 71(1):53--65, 2015.

\bibitem[CLN{\etalchar{+}}20]{CLNSS20}
Vincent Cohen{-}Addad, Silvio Lattanzi, Ashkan Norouzi{-}Fard, Christian Sohler, and Ola Svensson.
\newblock Fast and accurate $k$-means++ via rejection sampling.
\newblock In {\em NeurIPS}, 2020.

\bibitem[CLNS15]{ChanLNS15}
T.{-}H.~Hubert Chan, Mingfei Li, Li~Ning, and Shay Solomon.
\newblock New doubling spanners: Better and simpler.
\newblock {\em {SIAM} J. Comput.}, 44(1):37--53, 2015.

\bibitem[CLRS22]{CLRS22}
Thomas~H Cormen, Charles~E Leiserson, Ronald~L Rivest, and Clifford Stein.
\newblock {\em Introduction to algorithms}.
\newblock MIT press, 2022.

\bibitem[CLS{\etalchar{+}}22]{Cohen-AddadLSSS22}
Vincent Cohen{-}Addad, Kasper~Green Larsen, David Saulpic, Chris Schwiegelshohn, and Omar~Ali Sheikh{-}Omar.
\newblock Improved coresets for euclidean k-means.
\newblock In {\em NeurIPS}, 2022.

\bibitem[CLSS22]{Cohen-AddadLSS22}
Vincent Cohen{-}Addad, Kasper~Green Larsen, David Saulpic, and Chris Schwiegelshohn.
\newblock Towards optimal lower bounds for k-median and k-means coresets.
\newblock In {\em {STOC}}, pages 1038--1051. {ACM}, 2022.

\bibitem[Coh18]{Cohen-Addad18}
Vincent Cohen{-}Addad.
\newblock A fast approximation scheme for low-dimensional \emph{k}-means.
\newblock In {\em {SODA}}, pages 430--440. {SIAM}, 2018.

\bibitem[CSS21]{Cohen-AddadSS21}
Vincent Cohen{-}Addad, David Saulpic, and Chris Schwiegelshohn.
\newblock A new coreset framework for clustering.
\newblock In {\em {STOC}}, pages 169--182. {ACM}, 2021.

\bibitem[DIIM04]{DatarIIM04}
Mayur Datar, Nicole Immorlica, Piotr Indyk, and Vahab~S. Mirrokni.
\newblock Locality-sensitive hashing scheme based on p-stable distributions.
\newblock In {\em {SCG}}, pages 253--262. {ACM}, 2004.

\bibitem[DSS24]{DraganovSS24}
Andrew Draganov, David Saulpic, and Chris Schwiegelshohn.
\newblock Settling time vs. accuracy tradeoffs for clustering big data.
\newblock {\em Proc. {ACM} Manag. Data}, 2(3):173, 2024.

\bibitem[EN19]{ElkinN19}
Michael Elkin and Ofer Neiman.
\newblock Hopsets with constant hopbound, and applications to approximate shortest paths.
\newblock {\em {SIAM} J. Comput.}, 48(4):1436--1480, 2019.

\bibitem[ES15]{ElkinS15}
Michael Elkin and Shay Solomon.
\newblock Optimal euclidean spanners: Really short, thin, and lanky.
\newblock {\em J. {ACM}}, 62(5):35:1--35:45, 2015.

\bibitem[ES23]{ElkinS23}
Michael Elkin and Idan Shabat.
\newblock Path-reporting distance oracles with logarithmic stretch and size o(n log log n).
\newblock In {\em {FOCS}}, pages 2278--2311. {IEEE}, 2023.

\bibitem[FL11]{FeldmanL11}
Dan Feldman and Michael Langberg.
\newblock A unified framework for approximating and clustering data.
\newblock In {\em {STOC}}, pages 569--578. {ACM}, 2011.

\bibitem[FMS07]{FeldmanMS07}
Dan Feldman, Morteza Monemizadeh, and Christian Sohler.
\newblock A {PTAS} for k-means clustering based on weak coresets.
\newblock In {\em {SCG}}, pages 11--18. {ACM}, 2007.

\bibitem[FRS19]{FriggstadRS19}
Zachary Friggstad, Mohsen Rezapour, and Mohammad~R. Salavatipour.
\newblock Local search yields a {PTAS} for k-means in doubling metrics.
\newblock {\em {SIAM} J. Comput.}, 48(2):452--480, 2019.

\bibitem[FSS20]{FeldmanSS20}
Dan Feldman, Melanie Schmidt, and Christian Sohler.
\newblock Turning big data into tiny data: Constant-size coresets for k-means, pca, and projective clustering.
\newblock {\em {SIAM} J. Comput.}, 49(3):601--657, 2020.

\bibitem[FW12]{FilmusW12}
Yuval Filmus and Justin Ward.
\newblock The power of local search: Maximum coverage over a matroid.
\newblock In {\em {STACS}}, volume~14 of {\em LIPIcs}, pages 601--612. Schloss Dagstuhl - Leibniz-Zentrum f{\"{u}}r Informatik, 2012.

\bibitem[FW14]{FilmusW14}
Yuval Filmus and Justin Ward.
\newblock Monotone submodular maximization over a matroid via non-oblivious local search.
\newblock {\em {SIAM} J. Comput.}, 43(2):514--542, 2014.

\bibitem[GGK{\etalchar{+}}18]{0001G0MS0V18}
Martin Gro{\ss}, Anupam Gupta, Amit Kumar, Jannik Matuschke, Daniel~R. Schmidt, Melanie Schmidt, and Jos{\'{e}} Verschae.
\newblock A local-search algorithm for steiner forest.
\newblock In {\em {ITCS}}, volume~94 of {\em LIPIcs}, pages 31:1--31:17. Schloss Dagstuhl - Leibniz-Zentrum f{\"{u}}r Informatik, 2018.

\bibitem[GGN06]{GaoGN06}
Jie Gao, Leonidas~J. Guibas, and An~Thai Nguyen.
\newblock Deformable spanners and applications.
\newblock {\em Comput. Geom.}, 35(1-2):2--19, 2006.

\bibitem[GKL03]{GuptaKL03}
Anupam Gupta, Robert Krauthgamer, and James~R. Lee.
\newblock Bounded geometries, fractals, and low-distortion embeddings.
\newblock In {\em {FOCS}}, pages 534--543. {IEEE} Computer Society, 2003.

\bibitem[GMMO00]{GuhaMMO00}
Sudipto Guha, Nina Mishra, Rajeev Motwani, and Liadan O'Callaghan.
\newblock Clustering data streams.
\newblock In {\em {FOCS}}, pages 359--366. {IEEE} Computer Society, 2000.

\bibitem[GOR{\etalchar{+}}22]{GrandoniORSV22}
Fabrizio Grandoni, Rafail Ostrovsky, Yuval Rabani, Leonard~J. Schulman, and Rakesh Venkat.
\newblock A refined approximation for euclidean k-means.
\newblock {\em Inf. Process. Lett.}, 176:106251, 2022.

\bibitem[GPST23]{GowdaPST23}
Kishen~N. Gowda, Thomas~W. Pensyl, Aravind Srinivasan, and Khoa Trinh.
\newblock Improved bi-point rounding algorithms and a golden barrier for \emph{k}-median.
\newblock In {\em {SODA}}, pages 987--1011. {SIAM}, 2023.

\bibitem[GR08]{GottliebR08}
Lee{-}Ad Gottlieb and Liam Roditty.
\newblock Improved algorithms for fully dynamic geometric spanners and geometric routing.
\newblock In {\em {SODA}}, pages 591--600. {SIAM}, 2008.

\bibitem[GT08]{GuptaT08}
Anupam Gupta and Kanat Tangwongsan.
\newblock Simpler analyses of local search algorithms for facility location.
\newblock {\em CoRR}, abs/0809.2554, 2008.

\bibitem[HIM12]{Har-PeledIM12}
Sariel Har{-}Peled, Piotr Indyk, and Rajeev Motwani.
\newblock Approximate nearest neighbor: Towards removing the curse of dimensionality.
\newblock {\em Theory Comput.}, 8(1):321--350, 2012.

\bibitem[HIS13]{HIS13}
Sariel Har{-}Peled, Piotr Indyk, and Anastasios Sidiropoulos.
\newblock Euclidean spanners in high dimensions.
\newblock In {\em {SODA}}, pages 804--809. {SIAM}, 2013.

\bibitem[HK07]{Har-PeledK07}
Sariel Har{-}Peled and Akash Kushal.
\newblock Smaller coresets for k-median and k-means clustering.
\newblock {\em Discret. Comput. Geom.}, 37(1):3--19, 2007.

\bibitem[HKN16]{HenzingerKN16}
Monika Henzinger, Sebastian Krinninger, and Danupon Nanongkai.
\newblock A deterministic almost-tight distributed algorithm for approximating single-source shortest paths.
\newblock In {\em {STOC}}, pages 489--498. {ACM}, 2016.

\bibitem[HKN18]{HenzingerKN18}
Monika Henzinger, Sebastian Krinninger, and Danupon Nanongkai.
\newblock Decremental single-source shortest paths on undirected graphs in near-linear total update time.
\newblock {\em J. {ACM}}, 65(6):36:1--36:40, 2018.

\bibitem[HLW24]{Huang0024}
Lingxiao Huang, Jian Li, and Xuan Wu.
\newblock On optimal coreset construction for euclidean (k, z)-clustering.
\newblock In {\em {STOC}}, pages 1594--1604. {ACM}, 2024.

\bibitem[HM04]{Har-PeledM04}
Sariel Har{-}Peled and Soham Mazumdar.
\newblock On coresets for k-means and k-median clustering.
\newblock In {\em {STOC}}, pages 291--300. {ACM}, 2004.

\bibitem[HM06]{Har-PeledM06}
Sariel Har{-}Peled and Manor Mendel.
\newblock Fast construction of nets in low-dimensional metrics and their applications.
\newblock {\em {SIAM} J. Comput.}, 35(5):1148--1184, 2006.

\bibitem[HV20]{HuangV20}
Lingxiao Huang and Nisheeth~K. Vishnoi.
\newblock Coresets for clustering in euclidean spaces: importance sampling is nearly optimal.
\newblock In {\em {STOC}}, pages 1416--1429. {ACM}, 2020.

\bibitem[IM98]{IndykM98}
Piotr Indyk and Rajeev Motwani.
\newblock Approximate nearest neighbors: Towards removing the curse of dimensionality.
\newblock In {\em {STOC}}, pages 604--613. {ACM}, 1998.

\bibitem[JL84]{JL84}
William~B. Johnson and Joram Lindenstrauss.
\newblock Extensions of lipschitz mappings into hilbert space.
\newblock {\em Contemporary mathematics}, 26:189--206, 1984.

\bibitem[JMM{\etalchar{+}}03]{JainMMSV03}
Kamal Jain, Mohammad Mahdian, Evangelos Markakis, Amin Saberi, and Vijay~V. Vazirani.
\newblock Greedy facility location algorithms analyzed using dual fitting with factor-revealing {LP}.
\newblock {\em J. {ACM}}, 50(6):795--824, 2003.

\bibitem[JMS02]{JainMS02}
Kamal Jain, Mohammad Mahdian, and Amin Saberi.
\newblock A new greedy approach for facility location problems.
\newblock In {\em {STOC}}, pages 731--740. {ACM}, 2002.

\bibitem[JV01]{JainV01}
Kamal Jain and Vijay~V. Vazirani.
\newblock Approximation algorithms for metric facility location and \emph{k}-median problems using the primal-dual schema and lagrangian relaxation.
\newblock {\em J. {ACM}}, 48(2):274--296, 2001.

\bibitem[KLMS22]{KahalonLMS22}
Omri Kahalon, Hung Le, Lazar Milenkovic, and Shay Solomon.
\newblock Can't see the forest for the trees: Navigating metric spaces by bounded hop-diameter spanners.
\newblock In {\em {PODC}}, pages 151--162. {ACM}, 2022.

\bibitem[KMN{\etalchar{+}}04]{KanungoMNPSW04}
Tapas Kanungo, David~M. Mount, Nathan~S. Netanyahu, Christine~D. Piatko, Ruth Silverman, and Angela~Y. Wu.
\newblock A local search approximation algorithm for k-means clustering.
\newblock {\em Comput. Geom.}, 28(2-3):89--112, 2004.

\bibitem[KMSV98]{KhannaMSV98}
Sanjeev Khanna, Rajeev Motwani, Madhu Sudan, and Umesh~V. Vazirani.
\newblock On syntactic versus computational views of approximability.
\newblock {\em {SIAM} J. Comput.}, 28(1):164--191, 1998.

\bibitem[KSS04]{KumarSS04}
Amit Kumar, Yogish Sabharwal, and Sandeep Sen.
\newblock A simple linear time (1+{\(\acute{\epsilon}\)})-approximation algorithm for k-means clustering in any dimensions.
\newblock In {\em {FOCS}}, pages 454--462. {IEEE} Computer Society, 2004.

\bibitem[KT06]{KT06}
Jon~M. Kleinberg and {\'{E}}va Tardos.
\newblock {\em Algorithm design}.
\newblock Addison-Wesley, 2006.

\bibitem[LS16]{LiS16}
Shi Li and Ola Svensson.
\newblock Approximating k-median via pseudo-approximation.
\newblock {\em {SIAM} J. Comput.}, 45(2):530--547, 2016.

\bibitem[LS19]{LattanziS19}
Silvio Lattanzi and Christian Sohler.
\newblock A better k-means++ algorithm via local search.
\newblock In {\em {ICML}}, volume~97 of {\em Proceedings of Machine Learning Research}, pages 3662--3671. {PMLR}, 2019.

\bibitem[LST23]{LeST23}
Hung Le, Shay Solomon, and Cuong Than.
\newblock Optimal fault-tolerant spanners in euclidean and doubling metrics: Breaking the {\(\Omega\)} (log n) lightness barrier.
\newblock In {\em {FOCS}}, pages 77--97. {IEEE}, 2023.

\bibitem[LSW17]{LeeSW17}
Euiwoong Lee, Melanie Schmidt, and John Wright.
\newblock Improved and simplified inapproximability for k-means.
\newblock {\em Inf. Process. Lett.}, 120:40--43, 2017.

\bibitem[lTS24]{lTS24}
Max~Dupr{\'{e}} la~Tour and David Saulpic.
\newblock Almost-linear time approximation algorithm to euclidean k-median and k-means.
\newblock {\em CoRR}, abs/2407.11217, 2024.

\bibitem[MMR19]{MMR19}
Konstantin Makarychev, Yury Makarychev, and Ilya~P. Razenshteyn.
\newblock Performance of johnson-lindenstrauss transform for \emph{k}-means and \emph{k}-medians clustering.
\newblock In {\em {STOC}}, pages 1027--1038. {ACM}, 2019.

\bibitem[MP03]{MettuP03}
Ramgopal~R. Mettu and C.~Greg Plaxton.
\newblock The online median problem.
\newblock {\em {SIAM} J. Comput.}, 32(3):816--832, 2003.

\bibitem[MP04]{MettuP04}
Ramgopal~R. Mettu and C.~Greg Plaxton.
\newblock Optimal time bounds for approximate clustering.
\newblock {\em Mach. Learn.}, 56(1-3):35--60, 2004.

\bibitem[MPVX15]{MillerPVX15}
Gary~L. Miller, Richard Peng, Adrian Vladu, and Shen~Chen Xu.
\newblock Improved parallel algorithms for spanners and hopsets.
\newblock In {\em {SPAA}}, pages 192--201. {ACM}, 2015.

\bibitem[Nan14]{Nanongkai14}
Danupon Nanongkai.
\newblock Distributed approximation algorithms for weighted shortest paths.
\newblock In {\em {STOC}}, pages 565--573. {ACM}, 2014.

\bibitem[Sol14]{Solomon14}
Shay Solomon.
\newblock From hierarchical partitions to hierarchical covers: optimal fault-tolerant spanners for doubling metrics.
\newblock In {\em {STOC}}, pages 363--372. {ACM}, 2014.

\bibitem[SW18]{DBLP:conf/focs/SohlerW18}
Christian Sohler and David~P. Woodruff.
\newblock Strong coresets for k-median and subspace approximation: Goodbye dimension.
\newblock In {\em {FOCS}}, pages 802--813. {IEEE} Computer Society, 2018.

\bibitem[Tho04]{Thorup04}
Mikkel Thorup.
\newblock Quick k-median, k-center, and facility location for sparse graphs.
\newblock {\em {SIAM} J. Comput.}, 34(2):405--432, 2004.

\end{thebibliography}

\appendix

\section{Missing Proofs in \texorpdfstring{\Cref{sec:prelim}}{}}
\label{sec:appendix}

\subsection{Finding a naive solution}
\label{subsec:naive_solution}

Below we present the proof of \Cref{prop:naive_solution} (which is restated for ease of reference).

\begin{restate}[{\Cref{prop:naive_solution}}]
\begin{flushleft}
An $n^{z + 1}$-approximate feasible solution $\Cinitial \in V^{k}$ to the {\kzC} problem can be found in time $O(m \log(n))$.
\end{flushleft}
\end{restate}

Our strategy for constructing a new graph $G' = (V, E, w')$ needs a $\poly(n)$-approximate solution $C'$ to \kzC on $G$. Such a solution can be constructed

\newcommand{\Gaux}{G^{\sf aux}}

\begin{proof}
Regarding an undirected, connected, and non-singleton graph $G = (V, E, w)$ and its induced shortest-path metric space $(V,\dist)$, we find a solution $\Cinitial \in V^{k}$ in time $O(m \log(n))$, as follows.

\noindent
(i)~Initialize an edgeless {\em auxiliary subgraph} $\Gaux \gets (V, \emptyset, w)$ on the same vertices $V$; so its initialized number of components $c(\Gaux) = n$. \\
\Comment{This step trivially takes time $O(n)$.}

\noindent
(ii)~Sort and reindex all edges $E = \{e_{1}, e_{2}, \dots, e_{m}\}$ in increasing order of their nonnegative weights $0 \le w(e_{1}) \le w(e_{2}) \le \dots \le w(e_{m})$, using an arbitrary tie-breaking rule. \\
\Comment{This step trivially takes time $O(m \log(m)) = O(m \log(n))$.}

\noindent
(iii)~Add edges $e_{1}, e_{2}, \dots, e_{m}$ one by one to the auxiliary subgraph $\Gaux$ -- every addition decreases its number of components $c(\Gaux)$ by $1$ -- until $c(\Gaux) = k$. \\
\Comment{This step takes time $O(m + n \log(n))$. Namely, we repeats at most $m$ iterations and can implicitly maintain the components of the auxiliary subgraph $\Gaux$, using a {\em disjoint-set} data structure, in time $O(m + n \log(n))$ \cite[Chapter~19]{CLRS22}.}

\noindent
(iv)~Build our feasible solution $\Cinitial = \{c_{(i)}\}_{i \in [k]}$ by choosing one arbitrary vertex $c_{(i)}$ from each of the $k$ components (akin to $k$ clusters) of the ultimate auxiliary subgraph $\Gaux$. \\
\Comment{This step trivially takes time $O(n)$.}

It remains to establish that, regarding the {\kzC} problem, our solution $\Cinitial$ is an $n^{z + 1}$-approximation to the optimal solution $C^{*} \eqdef \argmin_{C \in V^{k}} \cost(V, C)$.
Consider the {\em maximum edge weight} $\Tilde{w}$ of the ultimate auxiliary subgraph $\Gaux$.
Regarding our solution $\Cinitial$, we must have $\dist(v, \Cinitial) \le (n - 1) \cdot \Tilde{w}$, for every vertex $v \in V$. It follows that
\begin{align*}
    \cost(V, \Cinitial)
    ~\le~ n \cdot (n - 1)^{z} \cdot \Tilde{w}^{z}
    ~=~ n^{z + 1} \cdot \Tilde{w}^{z}.
\end{align*}
Second, regarding the optimal solution $C^{*}$, we have $\dist(v, C^{*}) \ge \Tilde{w}$, for at least one vertex $v \in V$, since the restriction $\Tilde{G} = (V, \Tilde{E}, w)$ to the weight-$(< \Tilde{w})$ edges $\Tilde{E} \eqdef \{e \in E \mid w(e) < \Tilde{w}\}$ even cannot enable $k$ clusters, namely this restricted subgraph $\Tilde{G}$ must have {\em strictly more} components $c(\Tilde{G}) > k$.
It follows that
\begin{align*}
    \OPT ~=~ \cost(V, C^{*}) ~\ge~ \Tilde{w}^{z}.
\end{align*}
Combining the above two equations completes the proof of \Cref{prop:naive_solution}.
\end{proof}

\subsection{Bounding the edge weights}
\label{subsec:edge-weight}

Below we present the proof of \Cref{prop:edge-weight} (which is restated for ease of reference).

\begin{restate}[\Cref{prop:edge-weight}]
\begin{flushleft}
For any $0 < \eps < 1$, a graph $G = (V, E, w)$
can be converted into a new graph $G' = (V, E, w')$
in time $O(m \log(n))$, such that:
\begin{enumerate}[font = {\em\bfseries}]
    \item 
    $G'$ satisfies \Cref{assumption:edge-weight}
    with parameters
$\wmin = 1$ and $\wmax \le 32z^{2} \eps^{-2} n^{z + 5}$, i.e., all edges $(u, v) \in E$ have bounded weight $w'(u, v) \in [\wmin, \wmax]$, where (up to scale) the parameters $\wmin = 1$ and $\wmax \le 32z^{2} \eps^{-2} n^{z + 6}$.
    
    \item 
    Any $\alpha \in [1, n^{z + 1}]$-approximate solution $C$ to {\kzC} for $G'$
is a $2^{\eps}\alpha$-approximate solution
    to that for $G$.
\end{enumerate}

\end{flushleft}
\end{restate}

\newcommand{\costinitial}{\cost^{\sf init}}

\begin{proof}
In the original metric space $(V, \dist)$, consider the $n^{z + 1}$-approximate solution $\Cinitial$ found by \Cref{prop:naive_solution} in time $O(m \log(n))$; we can derive its clustering objective $\costinitial \eqdef \cost(V, \Cinitial)$, based on Dijkstra's algorithm, in time $O(m + n \log(n))$ \cite[Chapter~22]{CLRS22}.
Hence, the clustering objective $\OPT$ of the optimal solution $C^{*} \eqdef \argmin_{C \in V^{k}} \cost(V, C)$ is bounded between
\begin{align*}
    \OPT
    ~\eqdef~ \cost(V, C^{*})
    ~\in~ [\tfrac{1}{n^{z + 1}} \cdot \costinitial,\ \costinitial].
\end{align*}

We define the new edge weights $w'(u, v)$ as follows, based on two parameters $w_{\min} \le w_{\max}$.
\begin{align*}
    w_{\min}
    & ~\eqdef~ (\tfrac{1}{(1 + 3z / \eps) \cdot n^{2}})^{1 + 1 / z} \cdot (\costinitial)^{1 / z}, \\
    w_{\max}
    & ~\eqdef~ (2^{\eps}\alpha \cdot \costinitial)^{1 / z}, \\
    w'(u, v)
    & ~\eqdef~ \min\big(\max(w(u, v),\ w_{\min}),\ w_{\max}\big),
    && \forall (u, v) \in E.
\end{align*}
In total, our construction of the new graph $G'$ takes time $O(m \log n)$. We next verify \Cref{prop:edge-weight:bound,prop:edge-weight:solution}.

\vspace{.1in}
\noindent
{\bf \Cref{prop:edge-weight:bound}.}
The new edge weights $w'(u, v)$, for $(u, v) \in E$, differ at most by a multiplicative factor of
\begin{align*}
    \frac{w_{\max}}{w_{\min}}
    ~=~ 2^{\eps}\alpha \cdot ((1 + 3z / \eps) \cdot n^{2})^{1 + 1 / z}
    ~\le~ 32z^{2} \eps^{-2} n^{z + 5},
\end{align*}
given that $z \ge 1$, $0 < \eps < 1$, and $\alpha \in [1, n^{z + 1}]$.
Clearly, \Cref{prop:edge-weight:bound} follows, after scaling both parameters $w_{\min} \le w_{\max}$ and the new edge weights $w'(u, v)$.

\vspace{.1in}
\noindent
{\bf \Cref{prop:edge-weight:solution}.}
In the new metric space $(V, \dist')$, likewise, let $\cost'(V, C) \eqdef \sum_{v \in V} {\dist'}^{z}(v, C)$ be the new clustering objective and $\OPT' \eqdef \min_{C \in V^{k}} \cost'(V, C)$ be the new optimum.
It suffices to prove the following: For every solution $C \in V^{k}$ such that $\cost(V,C) > 2^{\eps}\alpha \cdot \OPT$,
\begin{align}
    \cost'(V, C)
    & ~>~ 2^{\eps}\alpha \cdot \OPT
    \label{eq:edge-weight:1} \\
    & ~\ge~ \alpha \cdot \OPT'.
    \label{eq:edge-weight:2}
\end{align}
We establish both \Cref{eq:edge-weight:1,eq:edge-weight:2} in the rest of this proof.

\vspace{.1in}
\noindent
{\bf \Cref{eq:edge-weight:1}.}
We focus on the case that $\dist'(v, C) < w_{\max}$, for every vertex $v \in V$; the opposite case is trivial $\cost'(V,C) \ge \max_{v \in V} {\dist'}^{z}(v, C) \ge w_{\max}^{z} = 2^{\eps}\alpha \cdot \costinitial \ge 2^{\eps}\alpha \cdot \OPT$.

Regarding the considered solution $C$ in the new metric space $(V, \dist')$, for every vertex $v \in V$, consider its shortest path $(u_{0} \equiv v), u_{1}, \dots, (u_{\ell} \in C)$.
In our focal case, all these edges $(u_{i - 1}, u_{i})$, for $i \in [\ell]$, have upper-bounded new weights $w'(u_{i - 1}, u_{i}) \le \dist'(v, C) < w_{\max}$;
given our construction, their original weights must be smaller $w(u_{i - 1}, u_{i}) \le w'(u_{i - 1}, u_{i}) < w_{\max}$.
As a consequence,
\begin{align*}
    \dist(v, C)
    ~\le~ \sum_{i\in [\ell]} w(u_{i - 1}, u_{i})
    ~\le~ \sum_{i\in [\ell]} w(u_{i - 1}, u_{i})
    ~=~ \dist'(v, C).
\end{align*}
This equation holds for every vertex $v \in V$.
Hence, we can infer \Cref{eq:edge-weight:1} as follows.
\begin{align*}
    \cost'(V, C)
    ~=~ \sum_{v \in V} {\dist'}^{z}(v, C)
    ~\ge~ \sum_{v \in V} \dist^{z}(v, C)
    ~=~ \cost(V, C)
    ~>~ 2^{\eps}\alpha \cdot \OPT.
\end{align*}
Here, the last step applies the premise of the considered solution $C \in V^{k}$.

\vspace{.1in}
\noindent
{\bf \Cref{eq:edge-weight:2}.}
Regarding the optimal solution $C^{*}$ in the original metric space $(V, \dist)$, for every vertex $v \in V$, consider its original shortest path $(v \equiv u_{0}), u_{1}, \dots, (u_{\ell} \in C^{*})$; these edges $(u_{i - 1}, u_{i})$, for $i \in [\ell]$, all have upper-bounded original weights $w(u_{i - 1}, u_{i}) \le (\OPT)^{1 / z} \le (\costinitial)^{1 / z} \le w_{\max}$.
Given our construction, the new metric space $(V, \dist')$ ensures that
\begin{align*}
    \dist'(v, C^{*})
    ~\le~ \sum_{i\in [\ell]} w'(u_{i - 1}, u_i)
    ~\le~ \sum_{i\in [\ell]} \big(w(u_{i - 1}, u_i) + w_{\min}\big)
    ~\le~ \dist(v, C^{*}) + n w_{\min}.
\end{align*}
This equation holds for every vertex $v \in V$.
Hence, we can infer \Cref{eq:edge-weight:2} as follows.
\begin{align*}
    \OPT'
    ~\le~ \sum_{v \in V} {\dist'}^{z}(v,C^{*})
    & ~\le~ \sum_{v \in V} \big(\dist(v, C^{*}) + n w_{\min}\big)^{z} \\
    & ~\le~ (1 + \tfrac{\eps}{3z})^{z - 1} \cdot \OPT
    + n \cdot (1 + \tfrac{3z}{\eps})^{z - 1} \cdot (n w_{\min})^{z} \\
    & ~\le~ (1 + \tfrac{\eps}{3z})^{z - 1} \cdot \OPT
    + \tfrac{\eps}{3} \cdot \OPT \\
    & ~\le~ 2^{\eps} \cdot \OPT.
\end{align*}
Here, the second step applies \Cref{lem:triangle}, by setting the parameter $\lambda = \frac{\eps}{3z}$.
The third step uses the formula of $w_{\min}$ and the fact $\costinitial \le n^{z + 1} \cdot \OPT$ (\Cref{prop:naive_solution}).
And the last step follows from elementary algebra.
Combining both \Cref{eq:edge-weight:1,eq:edge-weight:2} gives \Cref{prop:edge-weight:solution}.

This finishes the proof of \Cref{prop:edge-weight}.
\end{proof}

\section{Missing Proofs in \texorpdfstring{\Cref{sec:applications}}{}}
\label{sec:spanner}

\subsection{\texorpdfstring{Proof of \Cref{thm:spanner-lsh}}{}}
\label{sec:proof-spanner-lsh}

\begin{restate}[{\Cref{thm:spanner-lsh}}]
\begin{flushleft}
Fix any $c\ge 1$. 
Suppose that for any $r > 0$, the underlying metric space $(V, \dist)$ admits an $(r, c r, p_{1}, p_{2})$-sensitive family of hash functions such that 
\begin{itemize}
    \item $1 > p_{1} > p_{2} > 0$, $p_1^{-1} = o(n)$ and $\rho \eqdef \frac{\log(1 / p_{1})}{\log(1 / p_{2})} < 1/3$; and
    \item each function in this family can be sampled and evaluated in time $\TLSH$.
\end{itemize}
There exists a randomized algorithm that, given as input an $n$-point dataset $X \subseteq V$ with aspect ratio $\Delta \ge 1$, runs in time
\begin{equation*}
    O\left(\frac{n^{1+3\rho}\log^2(n)\log(\Delta)}{ p_1\log(1/p_2)}\cdot \TLSH\right)
\end{equation*}
to compute an $O(c)$-spanner for $X$ with 
\begin{equation*}
    O\left(\frac{n^{1 + 3\rho}\log(n)\log(\Delta)}{ p_1}\right)
\end{equation*}
edges.
The algorithm succeeds with probability $\ge 1 - \frac{1}{n^{\Theta(1)}}$.
\end{flushleft}
\end{restate}

\begin{proof}
    We assume w.l.o.g. that $\min_{u\neq v\in V} \dist(u,v) = 1$ and $\max_{u,v\in V} \dist(u,v) = \Delta$. Let $L\eqdef \log(\Delta) + 1\le O(\log(\Delta))$.

    \vspace{.1in}
    \noindent
    {\bf A weaker spanner.} We begin by constructing a weaker spanner such that, for a fixed pair of points $u, v \in X$, the resulting spanner preserves the distance between $u$ and $v$ with probability $n^{-\Omega(1)}$.
    The construction is as follows. 
    
    Initially, we construct a graph $H = (V,\emptyset)$.
For every integer scale $0 \leq \ell < L$, we perform the following:
We sample $M \eqdef \lceil \log(n^3) / \log(1/p_2) \rceil$ independent hash functions from the $(2^{\ell}, c\cdot 2^{\ell}, p_1, p_2)$-sensitive hashing family, denoted by $\{h^{(\ell)}_{i}\}_{i\in [M]}$, and define a new hash function $\hat{h^{(\ell)}}$ based on these, specifically: 
    \begin{equation}
        \label{eq:LSH-combine}
        \forall x\in V,\quad \hat{h^{(\ell)}}(x) := \left(h^{(\ell)}_1(x),\dots, h^{(\ell)}_{M}(x)\right).
    \end{equation}
    This hash function maps points in $X$ into buckets
    \begin{equation*}
        \calB^{(\ell)}:=\left\{ \hat{h^{(\ell)}}^{-1}(y): y \in \hat{h^{(\ell)}}(X)  \right\}.
    \end{equation*}
    Finally, for each bucket $B\in \calB^{(\ell)}$, we select an arbitrary point $s_B\in B$ as the center point, and add an edge $(s,u)$ with weight $\dist(s,u)$ to the graph $H$ for every point $u\in B - s_B$.
    
    The following claim summarizes the weaker spanner guarantee of the above construction.
    
    \begin{claim}
        \label{claim:weaker-spanner}
        For every pair of points $u,v\in X$, the graph $H$ constructed as above satisfies the following with probability $\Theta(p_1\cdot n^{-3\rho})$:
        \begin{equation*}
            \dist_H(u,v) \le 8c\cdot \dist(u,v).
        \end{equation*}
    \end{claim}
    
    We defer the proof of \Cref{claim:weaker-spanner} to later, and now we show how to use it to construct a spanner. The construction is simple: we construct $N = O(p_1^{-1}\cdot n^{3\rho} \cdot \log(n))$ independent weaker spanners using the above construction. Denote these weaker spanners by $H_i = (X, E_i)$, for $i \in [N]$. Our final spanner is then $G = (X, \bigcup_{i \in [N]} E_i)$. We then verify that with high probability, $G$ preserves every pairwise distances.
    
    To see this, consider a fixed pair of points $u, v \in V$. 
    By \Cref{claim:weaker-spanner}, we have
    \begin{equation*}
        \Pr\bigg[\forall i\in [N], \dist_{H_i}(u,v) > 8c \cdot \dist(u,v) \bigg] \le \left(1 - \Theta(p_1\cdot n^{-3\rho}) \right)^N \le e^{-\Theta(p_1\cdot n^{-3\rho}\cdot N)}.
    \end{equation*}
    Notice that $\dist_G(u,v) \le \dist_{H_i}(u,v)$ for all $i\in [N]$. The above implies that 
    \begin{equation*}
        \Pr\bigg[\dist_{G}(u,v) > 8c \cdot \dist(u,v) \bigg] \le e^{-\Theta(p_1\cdot n^{-3\rho}\cdot N)}.
    \end{equation*}
    Hence, by applying the union bound over all pairs of points in $X$ and choosing $N = O(p_1^{-1}\cdot n^{3\rho} \cdot \log(n))$ to be sufficiently large, we obtain
    \begin{equation*}
        \Pr\bigg[\exists u,v\in X, \dist_G(u,v) > 8c \cdot \dist(u,v) \bigg]\le e^{-\Theta(p_1\cdot n^{-3\rho}\cdot N)}\cdot n^2 \le \frac{1}{n^{\Theta(1)}}.
    \end{equation*}
    Thus $G$ is a $8c$-spanner with high probability. Finally, according to the construction, the edge number of $G$ is at most 
    \begin{equation*}
        N\cdot L\cdot n = O\left(\frac{n^{1 + 3\rho}\log(\Delta)\log(n)}{ p_1}\right),
    \end{equation*}
    and the running time of constructing $G$ is dominated by
    \begin{equation*}
        N\cdot L\cdot M\cdot n\cdot \TLSH = O\left(\frac{n^{1+3\rho}\log^2(n)\log(\Delta)}{ p_1\log(1/p_2)}\right)\cdot \TLSH.
    \end{equation*}
    
    This finishes the proof.
    \end{proof}
    \begin{proof}[Proof of \Cref{claim:weaker-spanner}]
        Suppose that $u \neq v$ and $2^{\ell - 1} \le \dist(u, v) \le 2^{\ell}$ for some $0 \le \ell < L$. 
        Note that such an integer $\ell$ must exist, since $1 \le \dist(u, v) \le \Delta \le 2^{L - 1}$.
        Let us consider the buckets $\calB^{(\ell)}$ at scale $\ell$, and the following two events:
    \begin{itemize}
        \item {\bf Event A:} $\hat{h^{(\ell)}}(u) =\hat{h^{(\ell)}}(v)$, i.e., $u$ and $v$ are mapped into the same bucket. By the definition of $\hat{h^{(\ell)}}$ (see \eqref{eq:LSH-combine}) and the definition of $(2^\ell, c\cdot 2^\ell, p_1, p_2)$-sensitive (\Cref{def:lsh}), we have 
        \begin{equation*}
            \Pr\left[\hat{h^{(\ell)}}(u) =\hat{h^{(\ell)}}(v)\right] = \Pr\left[\forall i\in [M], h^{(\ell)}_{i}(u) = h_{i}^{(\ell)}(v)\right] \ge p_1^M \ge p_1\cdot n^{-3\rho}
        \end{equation*}
        
        \item {\bf Event B:} Every point $u' \in X$ with $\dist(u, u') > c \cdot 2^{\ell}$ is not mapped to the same bucket as $u$.
        Fix such a point $u'\in X$ with $\dist(u, u') > c\cdot 2^{\ell}$. Again, by the definition of $\hat{h^{(\ell)}}$ (see \eqref{eq:LSH-combine}) and the definition of $(2^\ell, c\cdot 2^\ell, p_1, p_2)$-sensitive (\Cref{def:lsh}), we have 
        \begin{equation*}
            \Pr\left[\hat{h^{(\ell)}}(u) =\hat{h^{(\ell)}}(u')\right] = \Pr\left[\forall i\in [M], h^{(\ell)}_{i}(u) = h_{i}^{(\ell)}(u')\right]\le p_2^M \le n^{-3}.
        \end{equation*}
    \end{itemize}
    Notice that {\bf Event B} is a consequence of the event ``all such distant points are mapped into buckets different from that of $u$''. By the union bound over all such distant points (at most $n$ points), we have that the latter event happens with probability at least $1 - n^{-2}$, thus $\Pr[{\bf Event B}] \ge 1 - n^{-2}$.
    
    Now, assume that both {\bf Event A} and {\bf Event B} occur simultaneously, which happens with probability at least $p_1\cdot n^{-3\rho} - n^{-2} = \Theta(p_1n^{-3\rho})$ (recalling that $p_1^{-1} = o(n)$ and $\rho < 1 / 3$). 
Let $B \in \calB^{(\ell)}$ be the bucket that contains both $u$ and $v$. {\bf Event B} implies that $\diam(B) \le c\cdot 2^{\ell + 1}$. Let $s_B$ denote the selected center point from $B$.
    According to our construction, the edges $(u, s_B)$ and $(v, s_B)$ are added to the graph $H$, and therefore
    \begin{equation*} 
        \dist_H(u, v) \le \dist(u, s_B) + \dist(v, s_B) \le 2\cdot c\cdot 2^{\ell + 1} \le 8c \cdot \dist(u, v),
    \end{equation*} 
    where the last step follows from $\dist(u, v) \ge 2^{\ell - 1}$. This finishes the proof.
    \end{proof}

\subsection{\texorpdfstring{Proof of \Cref{cor:lsh}}{}}
\label{sec:proof-cor-lsh}

\begin{restate}[{\Cref{cor:lsh}}]
    Given constants $c\ge 1$ and $z\ge 1$, there is a randomized $O(c^z)$-approximation {\kzC} algorithm that has success probability $\ge 1 - n^{-\Theta(1)}$ and worst-case running time:
    \begin{itemize}
        \item {\bf (Euclidean space)} $O(dn^{1 + 1 / c^2 + o(1)}\log(\Delta))$ if the underlying metric space $(V,\dist) = (\R^d, \ell_2)$ is a $d$-dimensional Euclidean space;
        
        \item {\bf ($\ell_p$ metric)} $O(dn^{1 + 1 / c + o(1)}\log(\Delta))$ if the underlying metric space $(V,\dist) = (\R^d, \ell_p)$ for constant $p\in [1,2)$;
        
        \item {\bf (Jaccard metric)} $O(n^{1 + 1 / c + o(1)}\log(\Delta)\cdot |U|^2)$ if the underlying metric space is a Jaccard metric $(2^U, \dist)$, where $U$ is some universe and $\dist$ is the Jaccard distance.
    \end{itemize}
\end{restate}

\begin{proof}
We prove each case separately by plugging the corresponding LSH bounds into \Cref{cor:metric-lsh}.

\vspace{.1in}
\noindent
{\bf Euclidean spaces.} 
We employ the LSH construction from~\cite{AndoniI06}, which provides an $(r, c \cdot r, p_1, p_2)$-sensitive hashing family with $p_1^{-1} = 2^{O(\log^{2/3}n)}$, $(\log(1/p_2))^{-1} = O(1)$, $\rho = \frac{\log(1/p_1)}{\log(1/p_2)} \le 1/c^2 + O(\log\log n / \log^{1/3} n)$, and an evaluation time of $\TLSH = d\cdot 2^{O(\log^{2/3}n \log\log n)}$, for any $r > 0$. 

Note that Condition~\ref{thm:spanner-lsh:condition1} of \Cref{thm:spanner-lsh} is satisfied when $c > \sqrt{3}$ and $n$ is sufficiently large.
By plugging this LSH construction into \Cref{cor:metric-lsh} and scaling $c$ by a constant, we obtain an $O(c^z)$-approximation algorithm that runs in time $O(dn^{1 + 1 / c^2 + o(1)}\log(\Delta))$, where the factor $n^{o(1)}$ is $n^{O(\log\log n / \log^{1/3} n)}$.

\vspace{.1in}
\noindent
{\bf $\ell_p$ metric spaces.}
In an $\ell_p$ metric space with $p \in [1,2]$, we use the LSH construction from~\cite{DatarIIM04}, which provides an $(r, c \cdot r, p_1, p_2)$-sensitive hashing family, where both $p_1$ and $p_2$ are constants that depend only on $c$ and $p$. This construction achieves $\rho = \frac{\log(1/p_1)}{\log(1/p_2)} \le 2 c^{-1}$, with an evaluation time of $O(d)$ for any constant $c \ge 1$.
By plugging this LSH construction into \Cref{cor:metric-lsh} and scaling $c$ by a constant, we obtain a running time of $O\left(d n^{1 + 1 / c + o(1)}\cdot \log(\Delta)\right)$,
where the factor $n^{o(1)}$ is $2^{O(\sqrt{\log n\log\log n})}$.

\vspace{.1in}
\noindent
{\bf Jaccard metric spaces.}
In a Jaccard metric space, for any sets $A, B \in V \eqdef 2^U$, where $U$ is a universe, the distance is defined as $\dist(A,B) = 1 - \frac{|A\cap B|}{|A\cup B|}$. Therefore, we must have $\min_{A\neq B \in V} \dist(A,B)\ge \frac{1}{|U|}$, and thus $\Delta \le |U|$.
We use a classical \emph{min-hash}~\cite{Broder97, BroderGMZ97}, which is proven to be $(r, c\cdot r, 1 - r, 1 - cr)$-sensitive with $\rho = \frac{\log(1/p_1)}{\log(1/p_2)} \le 1 /c $  for any constant $c\ge 1$ and $r\in [\frac{1}{|U|},\frac{1}{2c}]$ (see e.g. \cite[Proposition 33]{BateniEFHJMW23}). 
The evaluation time for each hash function is $O(|U|)$.
However, since it only works for a specific range of values of $r$, the spanner construction of \Cref{thm:spanner-lsh} cannot be directly applied as a black box. To address this, we use a slightly modified spanner construction, which we sketch below:

For any $c > 3$, we apply the spanner construction from \Cref{thm:spanner-lsh} only for integer scales $\ell$ such that $\frac{1}{|U|} \le 2^\ell \le \frac{1}{2c}$. For every such scale $\ell\ge 0$, we use the $(2^{\ell}, c \cdot 2^{\ell}, 1 - 2^{\ell}, 1 - c \cdot 2^{\ell})$-sensitive min-hash family. 
This min-hash family is an LSH that satisfies $\rho = \frac{\log(1/p_1)}{\log(1/p_2)} \le \frac{1}{c}$, $p_1^{-1} \le \frac{1}{1 - 1/(2c)} = O(1)$, and $(\log(1/p_2))^{-1} \le \left(\log\left(\frac{1}{1 - c/|U|}\right)\right)^{-1} \le O(|U|)$.
Therefore, this construction gives a spanner $G$ with
\begin{equation*}
    O\left(\frac{n^{1 + 3\rho} \log(n) \log(\Delta)}{{p_1}}\right) = O\left(n^{1+3/c} \log(n) \log(\Delta)\right)
\end{equation*}
edges, and runs in time 
\begin{equation*}
O\left(\frac{n^{1+3\rho}\log^2(n)\log(\Delta)}{ p_1\log(1/p_2)}\cdot \TLSH\right) =  O\left( n^{1+3/c} \log^2(n) \log(\Delta)\cdot |U|^2\right).
\end{equation*}

Using the same argument as in \Cref{thm:spanner-lsh}, this spanner $G$ satisfies that for any $u, v \in V$ with $\dist(u, v) \le \frac{1}{4c}$, we have $\dist_G(u, v) \le O(c)\cdot \dist(u, v)$. However, for pairs of points $u, v \in V$ with $\dist(u, v) > \frac{1}{4c}$, we do not have such guarantees, and these two points may not even be connected. 
To address this, we pick an arbitrary point $s \in V$ and add an edge $(s, v)$ to $G$ for every $v \in P$, weighted by $\dist(s, v)$. This adds $n$ edges to $G$. After this modification, for any pair of points $u, v \in V$ that are far apart, with $\dist(u, v) > \frac{1}{4c}$, we have
\begin{equation*}
    \dist_G(u,v) \le \dist_G(u,s) + \dist_G(v,s) = \dist(u,s) + \dist(v,s)\le 2\le O(c)\cdot \dist(u,v),
\end{equation*}
where we use the fact that $\dist(u, s) \le 1$ in the Jaccard metric. Therefore, $G$ is now a $O(c)$-spanner. 

Finally, following a similar approach as in the Euclidean and $\ell_p$ spaces, we run \Cref{cor:general-graph} on the spanner $G$ and obtain a running time of $O\left(n^{1 + 1 / c + o(1)}\cdot \log(\Delta) \cdot |U|^2 \right)$,
where the factor $n^{o(1)}$ is $2^{O(\sqrt{\log n\log\log n})}$.
\end{proof}

\end{document}